\documentclass[journal,onecolumn]{IEEEtran}
\usepackage{ifpdf}
\usepackage{cite}
\usepackage{tikz}
\usepackage{graphicx}
\usepackage{lipsum}
\usepackage{subfigure}
\usepackage{mathtools}
\usepackage{multirow}
\usepackage{cuted}
\usepackage{balance}
\usetikzlibrary{decorations, decorations.text,}
\usepackage{float}
\usepackage{amsmath}
\usepackage{amsthm}
\usepackage{mathtools, cuted}
\usepackage{array}
\usepackage{textcomp}
\usepackage{url}
\usepackage{epstopdf}
\usepackage{paralist,esint,amssymb,booktabs,cases,bm,galois}

\newcommand\Tx[1]{\mathrm{#1}}
\newcommand\Se[1]{\mathcal{#1}}

\newcommand\Nm[1]{\lvert #1\rvert}

\newcommand\MB[1]{\left[#1\right]}

\newcommand{\RN}[1]{\textup{\uppercase\expandafter{\romannumeral#1}}}
\usepackage{cleveref}
\usetikzlibrary{calc}
\usetikzlibrary{shadings}
\newtheorem{theo}{Theorem}

\newtheorem{lemma}{Lemma}
\newtheorem{exam}{Example}
\newtheorem{defi}{Definition}
\newtheorem{rem}{Remark}

\makeatletter
\def\old@comma{,}
\catcode`\,=13
\def,{%
  \ifmmode%
    \old@comma\discretionary{}{}{}%
  \else%
    \old@comma%
  \fi%
}
\makeatother

\usepackage{algorithm}
\usepackage[noend]{algpseudocode}
\makeatletter
\def\BState{\State\hskip-\ALG@thistlm}
\makeatother

\definecolor{apple green}{rgb}{0.17,0.75,0.13}

\definecolor{edit}{rgb}{0.0,0.0,0.0}
\definecolor{edit_ah}{rgb}{0.0,0.0,0.0}

\errorcontextlines\maxdimen
\makeatletter
\newcommand*{\algrule}[1][\algorithmicindent]{\makebox[#1][l]{\hspace*{.5em}\vrule height 0.9 \baselineskip depth 0.3\baselineskip}}%

\newcount\ALG@printindent@tempcnta
\def\ALG@printindent{%
    \ifnum \theALG@nested>0
        \ifx\ALG@text\ALG@x@notext
            \addvspace{0pt}
        \else
            \unskip
            \ALG@printindent@tempcnta=1
            \loop
                \algrule[\csname ALG@ind@\the\ALG@printindent@tempcnta\endcsname]%
                \advance \ALG@printindent@tempcnta 1
            \ifnum \ALG@printindent@tempcnta<\numexpr\theALG@nested+1\relax
            \repeat
        \fi
    \fi
    }%
\usepackage{etoolbox}
\patchcmd{\ALG@doentity}{\noindent\hskip\ALG@tlm}{\ALG@printindent}{}{\errmessage{failed to patch}}
\makeatother



\begin{document}

\title{Breaking the Computational Bottleneck:\\ Design of Near-Optimal High-Memory Spatially-Coupled Codes }
\author{\IEEEauthorblockN{Siyi Yang, \IEEEmembership{Student Member, IEEE}, Ahmed Hareedy, \IEEEmembership{Member, IEEE}, Robert Calderbank, \IEEEmembership{Fellow, IEEE}, \\ and Lara Dolecek, \IEEEmembership{Senior Member, IEEE}}
\thanks{S. Yang and L. Dolecek are with the Electrical and Computer Engineering Department, University of California, Los Angeles, Los Angeles, CA 90095 USA (e-mail: siyiyang@ucla.edu and dolecek@ee.ucla.edu).}
\thanks{A. Hareedy and R. Calderbank are with the Electrical and Computer Engineering Department, Duke University, Durham, NC 27708 USA (e-mail: ahmed.hareedy@duke.edu and robert.calderbank@duke.edu).}
\thanks{The research was supported in part by UCLA Dissertation Year Fellowship, in part by the Air Force Office of Scientific Research (AFOSR) under Grant FA 9550-20-1-0266, and by the National Science Foundation (NSF) under Grants CCF-FET no. 2008728 and CCF 2106213. Part of the paper was presented at the 2021 IEEE International Symposium on Information Theory (ISIT) \cite{Yang2020GRADE}.}
}
\maketitle

\begin{abstract}
Spatially-coupled (SC) codes, known for their threshold saturation phenomenon and low-latency windowed decoding algorithms, are ideal for streaming applications and data storage systems. SC codes are constructed by partitioning an underlying block code, followed by rearranging and concatenating the partitioned components in a convolutional manner. The number of partitioned components determines the memory of SC codes. In this paper, we investigate the relation between the performance of SC codes and the density distribution of partitioning matrices. While adopting higher memories results in improved SC code performance, obtaining finite-length, high-performance SC codes with high memory is known to be computationally challenging. We break this computational bottleneck by developing a novel probabilistic framework that obtains (locally) optimal density distributions via gradient descent. Starting from random partitioning matrices abiding by the obtained distribution, we perform low-complexity optimization algorithms that minimize the number of detrimental objects to construct high-memory, high-performance quasi-cyclic SC codes. We apply our framework to various objects of interests, from the simplest short cycles, to more sophisticated objects such as concatenated cycles aiming at finer-grained optimization. Simulation results show that codes obtained through our proposed method notably outperform state-of-the-art SC codes with the same constraint length and optimized SC codes with uniform partitioning. The performance gain is shown to be universal over a variety of channels, from canonical channels such as additive white Gaussian noise and binary symmetric channels, to practical channels underlying flash memory and magnetic recording systems.

\end{abstract} 


\IEEEpeerreviewmaketitle

\section{Introduction}
\label{sectoin: introduction}

Spatially-coupled (SC) codes, also known as low-density parity-check (LDPC) codes with convolutional structures, are an ideal choice for streaming applications and data storage devices thanks to their threshold saturation phenomenon \cite{5695130,kumar2014threshold,olmos2015scaling,hareedy2017high,lentmaier2010iterative} and amenability to low-latency windowed decoding \cite{Iyengar2013windowed}. SC codes are constructed by partitioning the parity-check matrix of an underlying block code, followed by rearranging the component matrices in a \textit{convolutional} manner. In particular, component matrices are vertically concatenated into a \textit{replica}, and then multiple replicas are horizontally placed together, resulting in a \textit{coupled} code. The number of component matrices minus one is referred to as the \textit{memory} of the SC code \cite{mitchell2015spatially,esfahanizadeh2018finite,hareedy2020channel,pusane2011deriving}.

It is known that the performance of an SC code improves as its memory increases. This is a byproduct of improved node expansion and additional degrees of freedom that can be utilized to decrease the number of short cycles and detrimental objects \cite{esfahanizadeh2018finite,hareedy2020channel,dolecek2010analysis,naseri2020spatially,naseri2021construction}. A plethora of existing works \cite{esfahanizadeh2018finite,hareedy2020channel,battaglioni2017design,9112247} focus on minimizing the number of short cycles in the graph of the SC code. Although the optimization problem of designing SC codes with memory less than $4$ has been efficiently solved \cite{esfahanizadeh2018finite,hareedy2020channel}, there is still an absence of efficient algorithms that construct good enough SC codes with high memories systematically. Esfahanizadeh \textit{et al.} \cite{esfahanizadeh2018finite} proposed a combinatorial framework to develop optimal quasi-cyclic (QC) SC codes, comprising so-called optimal overlap (OO) to search for the optimal partitioning matrices, and lifting optimization (CPO) to optimize the lifting parameters, which was extended by Hareedy \textit{et al.} \cite{hareedy2020channel}. However, this method is hard to execute in practice for high-memory codes due to the increasing computational complexity. Heuristic methods that search for good SC codes with high memories are derived in \cite{battaglioni2017design,beemer2017generalized,9112247,naseri2021construction}. However, high-memory codes designed by purely heuristic methods are unable to reach the potential performance gain that can be achieved through high memories due to lack of theoretical properties; several of these codes can even be beat by optimally designed QC-SC codes with lower memories under the same constraint length \cite{9112247}. Therefore, a method that theoretically identifies an avenue to a near-optimal construction of SC codes with high memories is of significant interest. 

Inspired by the excellent performance and the low computational complexity offered by approaches comprising theoretical analysis guiding heuristic methods, we propose a two-step hybrid optimization framework that has these advantages. The framework first specifies a search subspace that is theoretically proved to be locally optimal, followed by a semi-greedy algorithm within this targeted search space. By analogy with threshold optimization approaches that search for LDPC ensembles with the optimal degree distribution, our first step is to obtain an SC ensemble with the optimal edge distribution (i.e., density distribution of component matrices). The associated metric is the expected number of targeted detrimental objects in the protograph of the code. Having reached a locally optimal edge distribution through gradient descent, we then apply a semi-greedy algorithm to search for a locally optimal partitioning matrix that satisfies this edge distribution. Our probabilistic framework is referred to as \textbf{gradient-descent distributor, algorithmic optimizer (GRADE-AO)}. 

Preliminary version of this work was presented in \cite{Yang2020GRADE}, where we focused only on the minimization of the number of short cycles. In this work, we develop a general framework that handles arbitrary objects. While cycles are detrimental in codes with low variable node (VN) degrees, such as regular codes with VN degree $2$ or $3$, objects that dominate the error profiles in higher-degree codes and irregular codes are typically more advanced. In particular, we focus on the concatenation of two short cycles in this paper. These concatenated cycles are common subgraphs of the detrimental objects, which are absorbing sets (ASs) \cite{dolecek2010analysis}), that govern the performance of LDPC codes with VN degree $\geq 3$ in error floor region. These detrimental objects are also the major source of undesirable dependencies that undermine the performance in the waterfall region. While focusing on cycles for simplicity, which is the case for the majority of existing works, unnecessary degrees of freedom could be exhausted on isolated cycles that are much less problematic. To the best of our knowledge, we are the first to provide a framework that systematically eliminates objects other than cycles from the Tanner graph of an SC code with mathematical guarantees, which is important for a variety of applications including storage systems. Hareedy \textit{et al.} \cite{hareedy2016general,hareedy2019combinatorial} proposed the so-called weight consistency matrix (WCM) framework to search for edge-weight assignments that minimize the number of ASs in non-binary (NB) LDPC codes with a given underlying topology, and demonstrated performance gains in data storage systems. Simulation results show that our framework leads to codes with excellent performance in flash memory and magnetic recording systems. The proposed GRADE-AO framework not only opens a door to fine-grained optimization over detrimental objects in SC codes, but also can be applied in optimizing the unweighted graphs of non-binary (NB) SC codes, which can lead to excellent NB-SC codes when combined with the WCM framework. Because of the improved threshold and waterfall performance, GRADE-AO has potential to produce SC codes for communication systems as well.

In this paper, we propose a probabilistic framework that efficiently searches for near-optimal SC codes with high memories. In \Cref{section: preliminaries}, we introduce preliminaries of SC codes and the performance-related metrics. In \Cref{section: framework}, we develop the theoretical basis of GRADE, which derives an edge distribution that determines a locally optimal SC ensemble. In \Cref{section: generalization of GRADE}, we introduce the theoretical details of how GRADE is generalized to more sophisticated objects. The distribution obtained through GRADE leads to effective initialization and specifies the search space of the semi-greedy algorithm adopted in AO afterwards. In \Cref{section: construction}, we introduce two examples of GRADE-AO that result in near-optimal SC codes: the so-called \textbf{gradient-descent (GD) codes} and \textbf{topologically-coupled (TC) codes}. In summary, we generalize GRADE to a rich class of relevant objects and present examples of GRADE-AO that focus on concatenated cycles. Our proposed framework is supported in \Cref{section: simulation} by simulation results of seven groups of codes, with the best code in each obtained from GRADE-AO. Finally, we make concluding remarks and introduce possible future work in \Cref{section: conclusion}.

\section{Preliminaries}
\label{section: preliminaries}

In this section, we recall the typical construction of SC codes with quasi-cyclic (QC) structure. Any QC code with a parity-check matrix $\mathbf{H}$ is obtained by replacing each nonzero (zero) entry of some binary matrix $\mathbf{H}^{\textup{P}}$ with a circulant (zero) matrix of size $z$, $z\in\mathbb{N}$. The matrix $\mathbf{H}^{\textup{P}}$ and $z$ are referred to as the protograph and the circulant size of the code, respectively. In particular, the protograph $\mathbf{H}^{\textup{P}}_{\textup{SC}}$ of an SC code has a convolutional structure composed of $L$ replicas, as presented in Fig.~\ref{fig: SC protograph}. Each replica is obtained by stacking the disjoint component matrices $\{\mathbf{H}^{\textup{P}}_i\}_{i=0}^{m}$, where $m$ is the memory and $\bm{\Pi}=\mathbf{H}^{\textup{P}}_{0}+\mathbf{H}^{\textup{P}}_1+\cdots+\mathbf{H}^{\textup{P}}_m$ is the protograph of the underlying block code. 

In this paper, we constrain $\bm{\Pi}$ to be an all-one matrix of size $\gamma\times \kappa$, $\gamma,\kappa\in\mathbb{N}$. An SC code is then uniquely represented by its partitioning matrix $\mathbf{P}$ and lifting matrix $\mathbf{L}$, where $\mathbf{P}$ and $\mathbf{L}$ are all $\gamma\times \kappa$ matrices. The matrix $\mathbf{P}$ has $(\mathbf{P})_{i,j}=a$ if $(\mathbf{H}^{\textup{P}}_a)_{i,j}=1$. The matrix $\mathbf{L}$ is determined by replacing each circulant matrix by its associated exponent. Here, this exponent represents the power to which the matrix $\bm{\sigma}$ defined by $(\bm{\sigma})_{i,i+1}=1$ is raised, where $(\bm{\sigma})_{z,z+1}=(\bm{\sigma})_{z,1}$.

\begin{figure}
\centering
\includegraphics[width=0.4\textwidth]{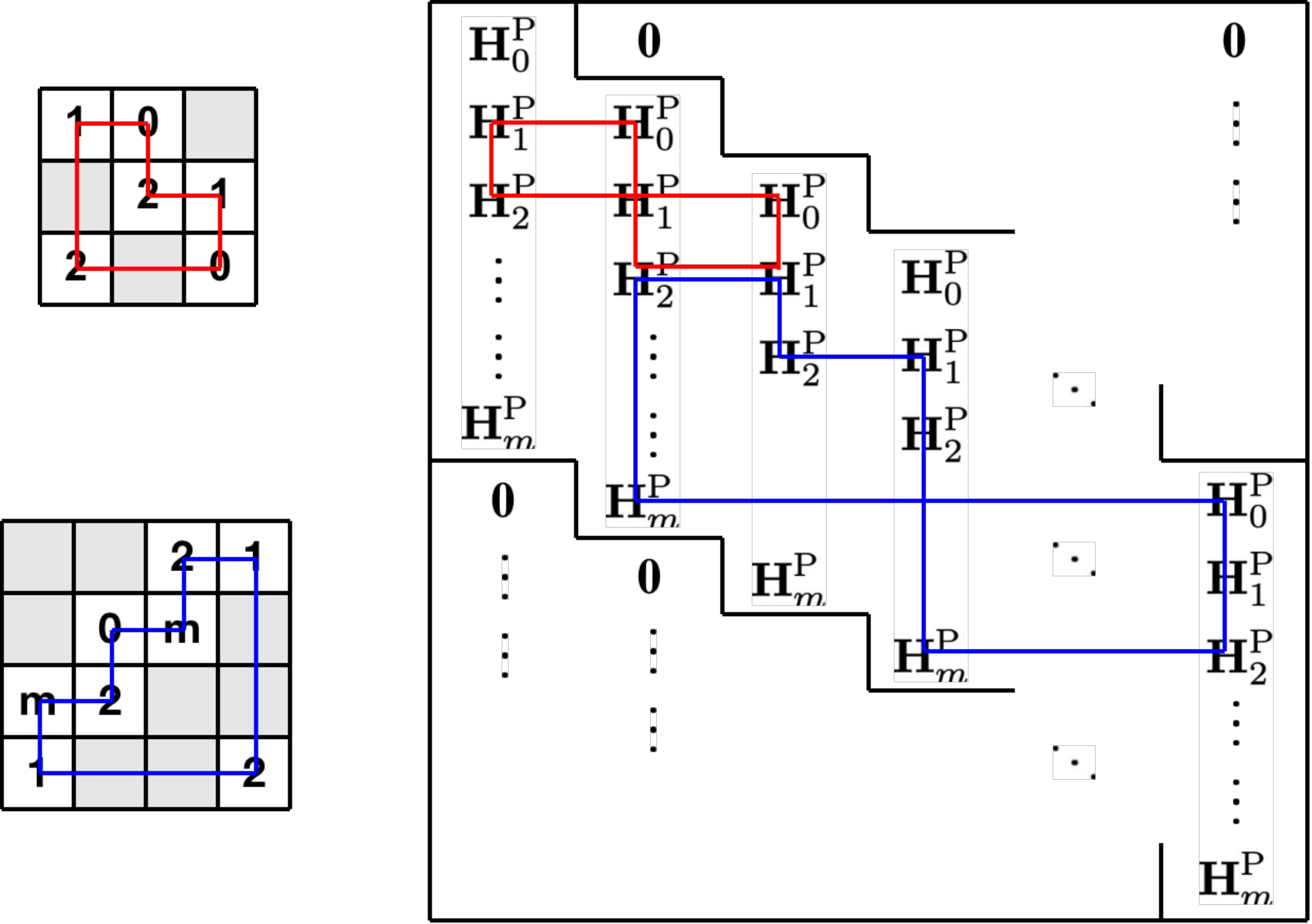}
\caption{Cycles in the protograph (right panel) and their corresponding structures in the partitioning matrices (left panel).}
\label{fig: SC protograph}
\end{figure}

The performance of finite-length LDPC codes is strongly affected by the number of detrimental objects that are subgraphs with certain structures in the Tanner graphs of those codes. Two major classes of detrimental objects are trapping sets and absorbing sets. Since enumerating and minimizing the number of detrimental objects is complicated, existing work typically focuses on common substructures of these objects: the short cycles \cite{esfahanizadeh2018finite,hareedy2020channel,battaglioni2017design}. A cycle-$2g$ candidate in $\mathbf{H}_{\textup{SC}}^{\textup{P}}$ ($\bm{\Pi}$) is a path of traversing a structure to generate cycles of length $2g$ after lifting (partitioning) \cite{hareedy2020channel}. In an SC code, each cycle in the Tanner graph corresponds to a cycle candidate in the protograph $\mathbf{H}_{\textup{SC}}^{\textup{P}}$, and each cycle candidate in $\mathbf{H}_{\textup{SC}}^{\textup{P}}$ corresponds to a cycle candidate $C$ in the base matrix $\mathbf{\Pi}$. \Cref{lemma: cycle condition} specifies a necessary and sufficient condition for a cycle candidate in $\bm{\Pi}$ to become a cycle candidate in the protograph and then a cycle in the final Tanner graph.

\begin{lemma}\label{lemma: cycle condition} Let $C$ be a cycle-$2g$ candidate in the base matrix, where $g\in\mathbb{N}$, $g\geq 2$. Denote $C$ by $(j_1,i_1,j_2,i_2,\dots,j_g,i_g)$, where $(i_{k},j_{k})$, $(i_{k},j_{k+1})$, $1\leq k\leq g$, $j_{g+1}=j_1$, are nodes of $C$ in $\bm{\Pi}$, $\mathbf{P}$, and $\mathbf{L}$. Then $C$ becomes a cycle candidate in the protograph if and only if the following condition follows \cite{battaglioni2017design}:
\begin{equation}\label{eqn: partition cycle}
\sum\nolimits_{k=1}^{g}\mathbf{P}(i_{k},j_{k})=\sum\nolimits_{k=1}^{g}\mathbf{P}(i_{k},j_{k+1}).
\end{equation}
This cycle candidate becomes a cycle in the Tanner graph if and only if \cite{fossorier2004quasicyclic}:
\begin{equation}\label{eqn: final cycle}
\sum\nolimits_{k=1}^{g}\mathbf{L}(i_{k},j_{k})\equiv\sum\nolimits_{k=1}^{g}\mathbf{L}(i_{k},j_{k+1}) \mod z.
\end{equation}
\end{lemma}

As shown in Fig.~\ref{fig: SC protograph}, a cycle-$6$ candidate and a cycle-$8$ candidate in the partitioning matrix with assignments satisfying condition (\ref{eqn: partition cycle}), and their corresponding cycle candidates in the protograph are marked by red and blue, respectively. An optimization of a QC-SC code is typically divided into two major steps: optimizing $\mathbf{P}$ to minimize the number of cycle candidates in the protograph, and optimizing $\mathbf{L}$ to further reduce that number in the Tanner graph given the optimized $\mathbf{P}$ \cite{esfahanizadeh2018finite,hareedy2020channel}. The latter goal has been achieved in \cite{esfahanizadeh2018finite} and \cite{hareedy2020channel}, using an algorithmic method called lifting optimization (CPO), while the former goal is yet to be achieved for large $m$. We note that the step separation highlighted above notably reduces the overall optimization complexity.

In the remainder of this paper, we first focus on QC-SC codes for the additive white Gaussian noise (AWGN) channel, where the most detrimental objects are the low weight absorbing sets (ASs) \cite{esfahanizadeh2018finite}. The ASs are defined in \Cref{defi: AS}.

\begin{defi}\label{defi: AS} \textbf{(Absorbing Sets)} Consider a subgraph induced by a subset $\mathcal{V}$ of VNs in the Tanner graph of a code. Set all the VNs in $\mathcal{V}$ to values in $\textup{GF}(q)$$\setminus$$\{0\}$ and set all other VNs to $0$. The set $\mathcal{V}$ is said to be an $(a, b)$ \textbf{absorbing set (AS)} over $\textup{GF}(q)$ if the size of $\mathcal{V}$ is $a$, the number of unsatisfied neighboring CNs of $\mathcal{V}$ is $b$, and each VN in $\mathcal{V}$ is connected to strictly more satisfied than unsatisfied neighboring CNs, for some set of VN values.

An $(a, b)$ \textbf{elementary AS} $\mathcal{V}$ over $\textup{GF}(q)$ is an $(a, b)$ AS  with the additional property that all the satisfied (resp., unsatisfied (if any)) neighboring CNs of $\mathcal{V}$ have degree $2$ (resp., degree $1$); otherwise the AS is referred to as an $(a, b)$ \textbf{non-elementary AS}.

Consider a subgraph induced by a subset $\mathcal{V}$ of VNs in the Tanner graph of a binary code. The set $\mathcal{V}$ is said to be an $(a, b)$ \textbf{binary AS} if the size of $\mathcal{V}$ is $a$, the number of odd-degree neighboring CNs of $\mathcal{V}$ is $b$, and each VN in $\mathcal{V}$ is connected to strictly more even-degree than odd-degree neighboring CNs.

Observe that the unlabeled configuration (all edge weights set to $1$) underlying an $(a, b)$ non-binary elementary AS is itself an $(a, b)$ binary elementary AS.

\end{defi}

Consequently, a simplified optimization focuses on cycle candidates of lengths $4$, $6$, and $8$ \cite{esfahanizadeh2018finite,hareedy2020channel}. Existing literature shows that the optimal $\mathbf{P}$ for an SC code with $m\leq 2$ typically has a balanced (uniform) edge distribution among component matrices \cite{esfahanizadeh2018finite}. However, in the remaining sections, we show that the edge distribution for optimal SC codes with large $m$ is not uniform, and we propose the GRADE-AO framework that explores a locally optimal solution. With the success in cycle optimization, we step forward to a finer-grained optimization over more advanced objects, which can be applied in higher degree codes and irregular codes. While GRADE can be generalized for arbitrary objects, we present constructions obtained though GRADE-AO focusing on concatenated cycles. Simulation results show that our proposed codes have excellent performance on practical channel models derived from flash memories and magnetic recording (MR), in both waterfall and error floor region.

\section{A Probabilistic Optimization Framework}
\label{section: framework}

In this section, we present a probabilistic framework that searches for a locally optimal edge distribution for the partitioning matrices of SC codes with given memories through the gradient-descent algorithm. 

\begin{defi}\label{defi: coupling pattern} Let $\gamma,\kappa,m,m_t\in\mathbb{N}$ and $\mathbf{a}=\left(a_0,a_1,\dots,a_{m_t}\right)$, where $0=a_0<a_1<\cdots<a_{m_t}=m$. A $(\gamma,\kappa)$ SC code with memory $m$ is said to have \textbf{coupling pattern} $\mathbf{a}$ if and only if $\mathbf{H}^{\textup{P}}_i\neq\mathbf{0}^{\gamma\times \kappa}$, for all $i\in\{a_0,a_1,\dots,a_{m_t}\}$, and $\mathbf{H}^{\textup{P}}_i=\mathbf{0}^{\gamma\times \kappa}$, otherwise. The value $m_t$ is called the \textbf{pseudo-memory} of the SC code.
\end{defi}

\subsection{Probabilistic Metric}
\label{subsec: metric}
In this subsection, we define metrics relating the edge distribution to the expected number of cycle candidates in the protograph in \Cref{theo: cycle 6 probability} and \Cref{theo: cycle 8 probability}. While Schmalen \textit{et al.} have shown in \cite{Schmalen} that nonuniform coupling (nonuniform edge distribution in our paper) yields an improved threshold, our work differs in two areas: 1) Explicit optimal coupling graphs were exhaustively searched and were restricted to small memories in \cite{Schmalen}, whereas our method produces near-optimal SC protographs for arbitrary memories. 2) Work \cite{Schmalen} focused on the asymptotic analysis for the threshold region, while our framework is dedicated to the finite-length construction and has additional demonstrable gains in the error floor region.

\begin{defi}\label{defi: coupling polynommial}
Let $m, m_t\in\mathbb{N}$ and $\mathbf{a}=\left(a_0,a_1,\dots,a_{m_t}\right)$, where $0=a_0<a_1<\cdots<a_{m_t}=m$. Let $\mathbf{p}=\left(p_0,p_1\dots,p_{m_t}\right)$, where $0<p_i\leq 1$, $p_0+p_1+\cdots+p_{m_t}=1$: each $p_i$ specifies the probability of a `$1$' in $\bm{\Pi}$ going to the component matrix $\mathbf{H}_{a_i}^{\textup{P}}$, thus $\mathbf{p}$ is referred to as \textbf{edge distribution} under random partition later on. Then, the following $f(X;\mathbf{a},\mathbf{p})$, which is abbreviated to $f(X)$ when the context is clear, is called the \textbf{coupling polynomial} of an SC code with coupling pattern $\mathbf{a}$, associated with probability distribution~$\mathbf{p}$:
\begin{equation}\label{eqn: coupling polynomial}
f(X;\mathbf{a},\mathbf{p})\triangleq \sum\nolimits_{0\leq i\leq m_t} p_i X^{a_i}.
\end{equation}
\end{defi}

\begin{theo}\label{theo: cycle 6 probability} Let $\left[\cdot\right]_i$ denote the coefficient of $X^{i}$ of a polynomial. Denote by $P_{6}(\mathbf{a},\mathbf{p})$ the probability of a cycle-$6$ candidate in the base matrix becoming a cycle-$6$ candidate in the protograph under random partitioning with edge distribution $\mathbf{p}$. Then,
\begin{equation} 
P_{6}(\mathbf{a},\mathbf{p})=\left[f^3(X)f^3(X^{-1})\right]_0.
\end{equation}
\end{theo}

\begin{proof} According to \Cref{lemma: cycle condition}, suppose the cycle-$6$ candidate in the base matrix is represented by $C(j_1,i_1,j_2,i_2,j_3,i_3)$. Then,
\begin{equation*}\label{eqn: equation}
\begin{split}
&P_{6}(\mathbf{a},\mathbf{p})=\mathbb{P}\left[\sum\nolimits_{k=1}^{3}\mathbf{P}(i_{k},j_{k})=\sum\nolimits_{k=1}^{3}\mathbf{P}(i_{k},j_{k+1})\right]\\
=&\sum\limits_{\sum\nolimits_{k=1}^{3}x_k=\sum\nolimits_{k=1}^{3}y_k} \prod_{k=1}^3\mathbb{P}\left[ \mathbf{P}(i_{k},j_{k})=x_k,\mathbf{P}(i_{k},j_{k+1})=y_k \right]\\
=&\sum\limits_{\sum\nolimits_{k=1}^{3}x_k=\sum\nolimits_{k=1}^{3}y_k} p_{x_1}p_{x_2}p_{x_3}p_{y_1}p_{y_2}p_{y_3}\\
=&\left[\sum\limits_{x_k,y_k\in \Tx{vals}(\mathbf{a})} p_{x_1}p_{x_2}p_{x_3}p_{y_1}p_{y_2}p_{y_3} X^{x_1+x_2+x_3-y_1-y_2-y_3}\right]_0\\
=&\left[f^3(X)f^3(X^{-1})\right]_0,\\
\end{split}
\end{equation*}
where $\Tx{vals}(\mathbf{a})$ is the set $\{a_0, a_1, \dots, a_{m_t}\}$. Thus, the theorem is proved.
\end{proof}

\begin{exam}\label{exam: SC probability} Consider SC codes with full memories and uniform partition, i.e., $\mathbf{a}=(0,1,\dots,m)$ and $\mathbf{p}=\frac{1}{m+1} \mathbf{1}_{m+1}$. When $m=2$, $P_6(\mathbf{a},\mathbf{p})=0.1934$; when $m=4$, $P_6(\mathbf{a},\mathbf{p})=0.1121$.
\end{exam}

\begin{exam}\label{exam: SC_m2_opt_edge_dist} First, consider SC codes with $m=m_t=2$. Let $\mathbf{a}_1=(0,1,2)$ and $\mathbf{p}_1=(2/5,1/5,2/5)$. According to \Cref{theo: cycle 6 probability}, $f(X)=(2+X+2X^2)/5$, $f^3(X)f^3(X^{-1})=0.0041(X^6+X^{-6})+0.0123(X^5+X^{-5})+0.0399(X^4+X^{-4})+0.0717(X^3+X^{-3})+0.1267(X^2+X^{-2})+0.1544(X+X^{-1})+0.1818$. Therefore, $P_6(\mathbf{a}_1,\mathbf{p}_1)=0.1818$. Second, consider SC codes with $m=m_t=4$. Let $\mathbf{a}_2=(0,1,2,3,4)$ and $\mathbf{p}_2=(0.31,0.13,0.12,0.13,0.31)$. According to \Cref{theo: cycle 6 probability}, $P_6(\mathbf{a}_2,\mathbf{p}_2)=0.0986$.
\end{exam}

\begin{figure}
\centering
\includegraphics[width=0.6\textwidth]{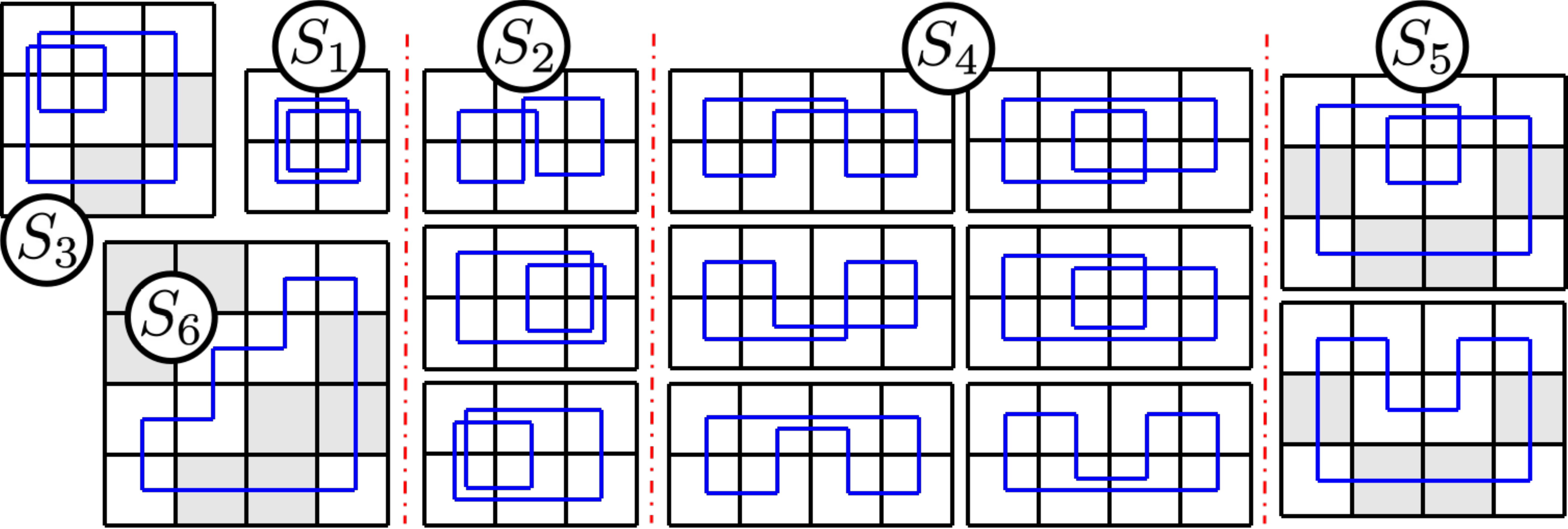}
\caption{Structures and cycle candidates for cycle-$8$.}
\label{fig: cycle8}
\end{figure}

After we have derived the metric for cycle-$6$ candidates in the protograph, we now turn to the case of cycle-$8$ candidates. As shown in Fig.~\ref{fig: cycle8}, cycle candidates in the base matrix that result in cycle-$8$ candidates in the protograph can be categorized into $6$ different structures, labeled $S_1,\dots,S_6$. Different cases are differentiated by the number of rows and columns (without order) the structures span in the partitioning matrix \cite{hareedy2020channel}. Specifically, $S_1,\dots,S_6$ denote the structures that span submatrices of size $2\times 2$, $2\times 3$ or $3\times 2$, $3\times 3$, $2\times 4$ or $4\times 2$, $3\times 4$ or $4\times 3$, and $4\times 4$, respectively. Any structure that belongs to $S_2,S_4, S_5$ has multiple cycle-$8$ candidates, and these distinct candidates are marked by blue in Fig.~\ref{fig: cycle8}. 

\begin{lemma}\label{lemma: cycle 8 pattern probability} Let $P_{8;i}(\mathbf{a},\mathbf{p})$, $1\leq i\leq 6$, denote the probability of a cycle-$8$ candidate of structure $S_i$ in the base matrix becoming a cycle-$8$ candidate in the protograph, under random partition with edge distribution $\mathbf{p}$. Then,
\begin{equation}\label{eqn: lemma cycle 8 pattern probability}
\begin{split}
\hspace{-0.5em}P_{8;1}(\mathbf{a},\mathbf{p})&=\left[f^2(X)f^2(X^{-1})\right]_0,\\
\hspace{-0.5em}P_{8;2}(\mathbf{a},\mathbf{p})&=\left[f(X^2)f(X^{-2})f^2(X)f^2(X^{-1})\right]_0,\\
\hspace{-0.5em}P_{8;3}(\mathbf{a},\mathbf{p})&=\left[f(X^2)f^2(X)f^4(X^{-1})\right]_0,\text{ and}\\
\hspace{-0.5em}P_{8;4}(\mathbf{a},\mathbf{p})&=P_{8;5}(\mathbf{a},\mathbf{p})=P_{8;6}(\mathbf{a},\mathbf{p})=\left[f^4(X)f^4(X^{-1})\right]_0. \nonumber
\end{split}
\end{equation}
\end{lemma}

\begin{proof} For structures where the nodes of the cycle-$8$ candidates are pairwise different, namely, $S_4,S_5,S_6$, the result can be derived by following the logic in the proof of \Cref{theo: cycle 6 probability}. 

For $S_1$, suppose the indices of the rows and columns are $i_1,i_2$, and $j_1,j_2$, respectively. Then, the cycle condition in \Cref{lemma: cycle condition} is $\mathbf{P}(i_1,j_1)+\mathbf{P}(i_2,j_2)=\mathbf{P}(i_1,j_2)+\mathbf{P}(i_2,j_1)$. 

For $S_2$, suppose the indices of the rows and columns are $i_1,i_2$, and $j_1,j_2,j_3$, respectively. Then, the cycle condition in \Cref{lemma: cycle condition} is $2\mathbf{P}(i_1,j_1)-2\mathbf{P}(i_2,j_1)+\mathbf{P}(i_2,j_2)+\mathbf{P}(i_2,j_3)-\mathbf{P}(i_1,j_2)-\mathbf{P}(i_1,j_3)=0$. 

For $S_3$, suppose the indices of the rows and columns are $i_1,i_2,i_3$, and $j_1,j_2,j_3$, respectively. Then, the cycle condition in \Cref{lemma: cycle condition} is $2\mathbf{P}(i_1,j_1)+\mathbf{P}(i_2,j_2)+\mathbf{P}(i_3,j_3)-\mathbf{P}(i_1,j_2)-\mathbf{P}(i_2,j_1)-\mathbf{P}(i_1,j_3)-\mathbf{P}(i_3,j_1)=0$. 

Following the logic in the proof of \Cref{theo: cycle 6 probability}, the case for $S_1,S_2,S_3$ can be proved.
\end{proof}

\begin{theo}\label{theo: cycle 8 probability} Denote $N_8(\mathbf{a},\mathbf{p})$ as the expectation of the number of cycle-$8$ candidates in the protograph. Then,
\begin{equation}\label{eqn: lemma cycle 8 probability}
\begin{split}
N_{8}(\mathbf{a},\mathbf{p})&=w_1\left[f^2(X)f^2(X^{-1})\right]_0+w_2\left[f(X^2)f(X^{-2})f^2(X)f^2(X^{-1})\right]_0\\
&+w_3\left[f(X^2)f^2(X)f^4(X^{-1})\right]_0+w_4\left[f^4(X)f^4(X^{-1})\right]_0,\\
\end{split}
\end{equation}
where $w_1=\binom{\gamma}{2}\binom{\kappa}{2}$, $w_2=3\binom{\gamma}{2}\binom{\kappa}{3}+3\binom{\gamma}{3}\binom{\kappa}{2}$, $w_3=18\binom{\gamma}{3}\binom{\kappa}{3}$, $w_4=6\binom{\gamma}{2}\binom{\kappa}{4}+6\binom{\gamma}{4}\binom{\kappa}{2}+36\binom{\gamma}{3}\binom{\kappa}{4}+36\binom{\gamma}{4}\binom{\kappa}{3}+24\binom{\gamma}{4}\binom{\kappa}{4}$.
\end{theo}

\begin{proof} Provided the results in \Cref{lemma: cycle 8 pattern probability}, we just need to prove that the numbers of cycle candidates of structures $S_1,S_2,\dots,S_6$ in a $\gamma\times \kappa$ base matrix are $\binom{\gamma}{2}\binom{\kappa}{2}$, $3\binom{\gamma}{2}\binom{\kappa}{3}+3\binom{\gamma}{3}\binom{\kappa}{2}$, $18\binom{\gamma}{3}\binom{\kappa}{3}$, $6\binom{\gamma}{2}\binom{\kappa}{4}+6\binom{\gamma}{4}\binom{\kappa}{2}$, $36\binom{\gamma}{3}\binom{\kappa}{4}+36\binom{\gamma}{4}\binom{\kappa}{3}$, and $24\binom{\gamma}{4}\binom{\kappa}{4}$, respectively. 

Take $i=5$ as an example. The number of cycle candidates of structure $S_5$ in any $3\times 4$ or $4\times 3$ matrix is $3\cdot \binom{4}{2} \cdot 2=36$. The total number of $3\times 4$ or $4\times 3$ matrices in a $\gamma\times \kappa$ base matrix is $\binom{\gamma}{3}\binom{\kappa}{4}+\binom{\gamma}{4}\binom{\kappa}{3}$. Therefore, the total number of cycle candidates of structure $S_5$ is $36\binom{\gamma}{3}\binom{\kappa}{4}+36\binom{\gamma}{4}\binom{\kappa}{3}$. By a similar logic, we can prove the result for the remaining structures.
\end{proof}

\begin{rem} Note that each cycle candidate satisfying the cycle condition in the partitioning matrix can result in multiple cycle candidates in the protograph; the multiplicity is determined by its width, i.e., the number of replicas each resultant cycle candidate spans in the protograph. We ignore the number of replicas a cycle candidate spans in $\mathbf{H}_{\textup{SC}}^{\textup{P}}$. We address this number in the CPO stage.
\end{rem}


\subsection{Gradient-Descent Distributor}
\label{subsec: gradient descent}
By contrasting Examples \ref{exam: SC probability} and \ref{exam: SC_m2_opt_edge_dist} it is clear that for a given coupling pattern, an optimal edge distribution is not necessarily reached by a uniform partition. In this subsection, we develop an algorithm that obtains a locally optimal distribution by gradient descent.

\begin{lemma}\label{lemma: cycle 6 distribution} Given $m_t\in\mathbb{N}$ and $\mathbf{a}=(a_0,a_1,\dots,a_{m_t})$, a necessary condition for $P_6(\mathbf{a},\mathbf{p})$ to reach its minimum value is that the following equation holds for some $c_0\in\mathbb{R}$:
\begin{equation}\label{eqn: cycle 6 condition}
\left[f^3(X)f^2(X^{-1})\right]_{a_i}=c_0,\ \forall i, 0\leq i\leq m_t.
\end{equation}
\end{lemma}

\begin{proof} Consider the gradient of $L_6(\mathbf{a},\mathbf{p})=P_6(\mathbf{a},\mathbf{p})+c(1-p_{0}-p_{1}-\dots-p_{{m_t}})$. 
\begin{equation}\label{eqn: lemma cycle 6 gradient}
\begin{split}
&\nabla_{\mathbf{p}}L_6(\mathbf{a},\mathbf{p})\\
=&\nabla_{\mathbf{p}}\left(P_6(\mathbf{a},\mathbf{p})+c(1-p_{0}-p_{1}-\dots-p_{{m_t}})\right)\\
=&\nabla_{\mathbf{p}}\left[f^3(X)f^3(X^{-1})\right]_0-c\mathbf{1}_{m_t+1}\\
=&\left[\nabla_{\mathbf{p}}\left(f^3(X)f^3(X^{-1})\right)\right]_0-c\mathbf{1}_{m_t+1}\\
=&3\left[f^2(X)f^2(X^{-1})f(X)\nabla_{\mathbf{p}}f(X^{-1})\right]_0+3\left[f^2(X)f^2(X^{-1})f(X^{-1})\nabla_{\mathbf{p}}f(X)\right]_0-c\mathbf{1}_{m_t+1}\\
=&6\left[f^3(X)f^2(X^{-1})\left(X^{-a_{0}},X^{-a_{1}},\dots,X^{-a_{m_t}}\right)\right]_0-c\mathbf{1}_{m_t+1}.
\end{split}
\end{equation}
When $P_8(\mathbf{a},\mathbf{p})$ reaches its minimum, $\nabla_{\mathbf{p}}\left[L(\mathbf{a},\mathbf{p})\right]=\mathbf{0}_{m_t+1}$, which is equivalent to (\ref{eqn: cycle 6 condition}) by defining $c_0=c/6$.
\end{proof}

\begin{lemma}\label{lemma: cycle 8 distribution} Given $\gamma,\kappa,m_t\in\mathbb{N}$ and $\mathbf{a}=(a_0,a_1,\dots,a_{m_t})$, a necessary condition for $N_8(\mathbf{a},\mathbf{p})$ to reach its minimum value is that the following equation holds for some $c_0\in\mathbb{R}$:
\begin{equation*}\label{eqn: cycle 8 condition}
\begin{split}
&\left[4f^2(X)f(X^{-1})\right])_{a_i}+\bar{w}_2\left[2f(X^{2})f^2(X)f^2(X^{-1})\right]_{2a_i}+\bar{w}_2\left[4f(X^2)f(X^{-2})f^2(X)f(X^{-1})\right]_{a_i}\\
+&\bar{w}_3\left[f^2(X)f^4(X^{-1})\right]_{-2a_i}+\bar{w}_3\left[2f(X^2)f(X)f^4(X^{-1})\right]_{-a_i}+\bar{w}_3\left[4f(X^2)f^2(X)f^3(X^{-1})\right]_{a_i}\\
+&\bar{w}_4\left[8f^4(X)f^3(X^{-1})\right]_{a_i}=c_0,\ \forall i, 0\leq i\leq m_t,
\end{split}
\end{equation*}
where $\bar{w}_2=\gamma+\kappa-4$, $\bar{w}_3=2(\gamma-2)(\kappa-2)$, and $\bar{w}_4=\frac{1}{2}\left[(\gamma-2)(\gamma-3)+(\kappa-2)(\kappa-3)\right]+(\gamma-2)(\kappa-2)(\gamma+\kappa-6)+\frac{1}{6}(\gamma-2)(\gamma-3)(\kappa-2)(\kappa-3)$.
\end{lemma}

\begin{proof} Consider the gradient of $L_8(\mathbf{a},\mathbf{p})=N_8(\mathbf{a},\mathbf{p})+c(1-p_{0}-p_{1}-\dots-p_{{m_t}})$. 
\begin{equation}\label{eqn: lemma cycle 8 gradient}
\begin{split}
&\nabla_{\mathbf{p}}L_8 (\mathbf{a},\mathbf{p})\\
=&\nabla_{\mathbf{p}}\left(N_8(\mathbf{a},\mathbf{p})+c(1-p_{0}-p_{1}-\dots-p_{{m_t}})\right)\\
=&w_1\left[\nabla_{\mathbf{p}}\left(f^2(X)f^2(X^{-1})\right)\right]_0+w_2\left[\nabla_{\mathbf{p}}\left(f(X^2)f(X^{-2})f^2(X)f^2(X^{-1})\right)\right]_0\\
&+w_3\left[\nabla_{\mathbf{p}}\left(f(X^2)f^2(X)f^4(X^{-1})\right)\right]_0+w_4\left[\nabla_{\mathbf{p}}\left(f^4(X)f^4(X^{-1})\right)\right]_0-c\mathbf{1}_{m_t+1}\\
=&w_1\{\left[4f^2(X)f(X^{-1})(X^{-a_0},X^{-a_1},\dots,X^{-a_{m_t}})\right])_0+\bar{w}_2\left[2f(X^2)f^2(X)f^2(X^{-1})(X^{-2a_0},\dots,X^{-2a_{m_t}})\right]_0\\
&+\bar{w}_2\left[4f(X^2)f(X^{-2})f^2(X)f(X^{-1})(X^{-a_0},\dots,X^{-a_{m_t}})\right]_0+\bar{w}_3\left[f^2(X)f^4(X^{-1})(X^{2a_0},X^{2a_1},\dots,X^{2a_{m_t}})\right]_0\\
&+\bar{w}_3\left[2f(X^2)f(X)f^4(X^{-1})(X^{a_0},X^{a_1},\dots,X^{a_{m_t}})\right]_0+\bar{w}_3\left[4f(X^2)f^2(X)f^3(X^{-1})(X^{-a_0},X^{-a_1},\dots,X^{-a_{m_t}})\right]_0\\
&+\bar{w}_4\left[8f^4(X)f^3(X^{-1})(X^{-a_0},X^{-a_1},\dots,X^{-a_{m_t}})\right]_0\}-c\mathbf{1}_{m_t+1}.
\end{split}
\end{equation}
When $P_8(\mathbf{a},\mathbf{p})$ reaches its minimum, $\nabla_{\mathbf{p}}\left[L(\mathbf{a},\mathbf{p})\right]=\mathbf{0}_{m_t+1}$, which is equivalent to (\ref{eqn: cycle 8 condition}) by defining $c_0=c/w_1$.
\end{proof}

Based on \Cref{lemma: cycle 6 distribution} and \Cref{lemma: cycle 8 distribution}, we adopt the gradient-descent algorithm to obtain a locally optimal edge distribution for SC codes with coupling pattern $\mathbf{a}$, starting from the uniform distribution inside $\mathbf{P}$ as presented in \Cref{algo: GRADE}. Note that $\Tx{conv}(\cdot)$ and $\Tx{flip}(\cdot)$ refer to convolution and reverse of vectors, respectively.

\begin{algorithm}
\caption{Gradient-Descent Distributor (GRADE) for Cycle Optimization} \label{algo: GRADE}
\begin{algorithmic}[1]
\Require 
\Statex $\gamma,\kappa,m_t,m,\mathbf{a}$: parameters of the SC code;
\Statex $w$: weight of each cycle-$6$ candidate;
\Statex $\epsilon,\alpha$: accuracy and step size of gradient descent;
\Ensure
\Statex $\mathbf{p}$: a locally optimal edge distribution over $\Tx{vals}(\mathbf{a})$;
\State $\bar{w}_1\gets \frac{2w}{3}(\gamma-2)(\kappa-2)$, obtain $\{\bar{w}_i\}_{i=2}^4$ in \Cref{lemma: cycle 8 distribution}; 
\State $v_{prev}=1$; $v_{cur}=1$;
\State $\mathbf{p},\mathbf{g}\gets \mathbf{0}_{m_t+1}$, $\mathbf{f},\bar{\mathbf{f}}\gets \mathbf{0}_{m+1}$, $\mathbf{f}_2,\bar{\mathbf{f}}_2\gets \mathbf{0}_{2m+1}$;
\State $\mathbf{p}\gets \frac{1}{m_t+1} \mathbf{1}_{m_t+1}$;
\State $\mathbf{f}\left[a_0,\dots,a_{m_t}\right]\gets\mathbf{p}$, $\bar{\mathbf{f}}\gets \Tx{flip}(\mathbf{f})$;
\State $\mathbf{f}_2\left[1,3,\dots,2m+1\right]\gets\mathbf{f}$, $\bar{\mathbf{f}}_2\gets \Tx{flip}(\mathbf{f}_2)$;
\State $\mathbf{q}_1 \gets \bar{w}_1\Tx{conv}(\mathbf{f},\mathbf{f},\mathbf{f},\bar{\mathbf{f}},\bar{\mathbf{f}},\bar{\mathbf{f}})$, $\mathbf{q}_2\gets \Tx{conv}(\mathbf{f},\mathbf{f},\bar{\mathbf{f}},\bar{\mathbf{f}})$;
\State $\mathbf{q}_3\gets \bar{w}_2\Tx{conv}(\mathbf{f}_2,\bar{\mathbf{f}}_2,\mathbf{f},\mathbf{f},\bar{\mathbf{f}},\bar{\mathbf{f}})+\bar{w}_3\Tx{conv}(\mathbf{f}_2,\mathbf{f},\mathbf{f},\bar{\mathbf{f}},\bar{\mathbf{f}},\bar{\mathbf{f}},\bar{\mathbf{f}})+\bar{w}_4\Tx{conv}(\mathbf{f},\mathbf{f},\mathbf{f},\mathbf{f},\bar{\mathbf{f}},\bar{\mathbf{f}},\bar{\mathbf{f}},\bar{\mathbf{f}})$;
\State $v_{prev}=v_{cur}$, $v_{cur}=\mathbf{q}_1\left[3m\right]+\mathbf{q}_2\left[2m\right]+\mathbf{q}_3\left[4m\right]$;
\State $\mathbf{g}_1\gets 6\bar{w}_1\Tx{conv}(\mathbf{f},\mathbf{f},\mathbf{f},\bar{\mathbf{f}},\bar{\mathbf{f}})$, $\mathbf{g}_2\gets 4\Tx{conv}(\mathbf{f},\mathbf{f},\bar{\mathbf{f}})$;
\State $\mathbf{g}_3\gets 4\bar{w}_2\Tx{conv}(\mathbf{f}_2,\bar{\mathbf{f}}_2,\mathbf{f},\mathbf{f},\bar{\mathbf{f}})+2\bar{w}_3\Tx{conv}(\bar{\mathbf{f}}_2,\mathbf{f},\mathbf{f},\mathbf{f},\mathbf{f},\bar{\mathbf{f}})+4\bar{w}_3\Tx{conv}(\mathbf{f}_2,\mathbf{f},\mathbf{f},\bar{\mathbf{f}},\bar{\mathbf{f}},\bar{\mathbf{f}})+8\bar{w}_4\Tx{conv}(\mathbf{f},\mathbf{f},\mathbf{f},\mathbf{f},\bar{\mathbf{f}},\bar{\mathbf{f}},\bar{\mathbf{f}})$;
\State $\mathbf{g}_4\gets 2\bar{w}_2\Tx{conv}(\mathbf{f}_2,\mathbf{f},\mathbf{f},\bar{\mathbf{f}},\bar{\mathbf{f}})+\bar{w}_3\Tx{conv}(\mathbf{f},\mathbf{f},\mathbf{f},\mathbf{f},\bar{\mathbf{f}},\bar{\mathbf{f}})$;
\State $\mathbf{g}\gets \mathbf{g}_1\left[2m+\mathbf{a}\right]+\mathbf{g}_2\left[m+\mathbf{a}\right]+\mathbf{g}_3\left[3m+\mathbf{a}\right]+\mathbf{g}_4\left[2m+2\mathbf{a}\right]$, $\mathbf{g}\gets \mathbf{g}-\Tx{mean}(\mathbf{g})$;
\If{$|v_{prev}-v_{cur}|>\epsilon$} 
\State $\mathbf{p}\gets \mathbf{p}-\alpha\frac{\mathbf{g}}{||\mathbf{g}||}$;
\State \textbf{goto} step 5;
\EndIf
\State \textbf{return} $\mathbf{p}$;
\end{algorithmic}
\end{algorithm}

\section{Generalization of GRADE}
\label{section: generalization of GRADE}

We have explained the basic idea of GRADE in optimizing the edge distribution of SC code ensembles with respect to the expected number of cycles. Cycles have been studied extensively in related literature (see e.g., \cite{girth1,girth2,girth3}) due to their simplicity and presence in problematic objects. However, cycles alone do not always account for typical decoding failures; for example, isolated cycles are not as harmful as concentrated cycles in codes with VN degree $4$ since single cycles on their own do not lead to decoding failures (as captured by e.g., ASs \cite{hareedy2020channel}), rather concatenated cycles do. An excessive focus on the removal of isolated cycles can lead to remarkably less degrees of freedom for the removal of dominant problematic objects. In this section, we therefore extend the theory of GRADE to arbitrary subgraphs.

\subsection{Probabilistic Metric}
\label{subsec: objects_probablistic metric}

In this subsection, we generalize the results presented in \Cref{subsec: metric} to obtain closed-form representations of the expected number of objects with arbitrary topologies. The key idea is that the dependency among nodes within each object can be fully described by a minimal set of fundamental cycles (or basic cycles), which is referred to as the \textbf{cycle basis} of the object (see \Cref{defi: cycle basic}) \cite{6387307,amiri2014analysis,hareedy2020minimizing}.

\begin{defi}\label{defi: cycle basic} (\textbf{Cycle Basis}) A \textbf{cycle basis} of an object is a minimum-cardinality set of cycles using disjunctive unions of which, each cycle in the object can be obtained; we call the cycles in this set \textbf{fundamental cycles}.
\end{defi}

In the remainder of this paper, we define a \textbf{prototype} of an object, for simplicity, as an assignment of the indices of its variable nodes (VNs) and its check nodes (CNs) in the base matrix. Prototype is a natural extension of \textit{pattern}, i.e., a prototype of an object is exactly a pattern (discussed in \cite{hareedy2020channel}) when the object is a cycle. According to \cite{hareedy2020minimizing}, we call a prototype \textbf{active} if all the fundamental cycles satisfy the cycle condition simultaneously, which means the detrimental object will be created in the protograph after partitioning the base matrix. The probability of a prototype becoming active under a random partition is proved to be represented by the constant term of a multi-variate polynomial, where each variable is associated with a cycle in the cycle basis: we refer to this polynomial as the \textbf{characteristic polynomial} of the object associated with fixed prototype. The overall characteristic polynomial of the object without specifying the prototype is then obtained as an average over the characteristic polynomials associated with all possible prototypes.

We start with a motivating example.

\begin{figure}
\centering
\includegraphics[width=0.96\textwidth]{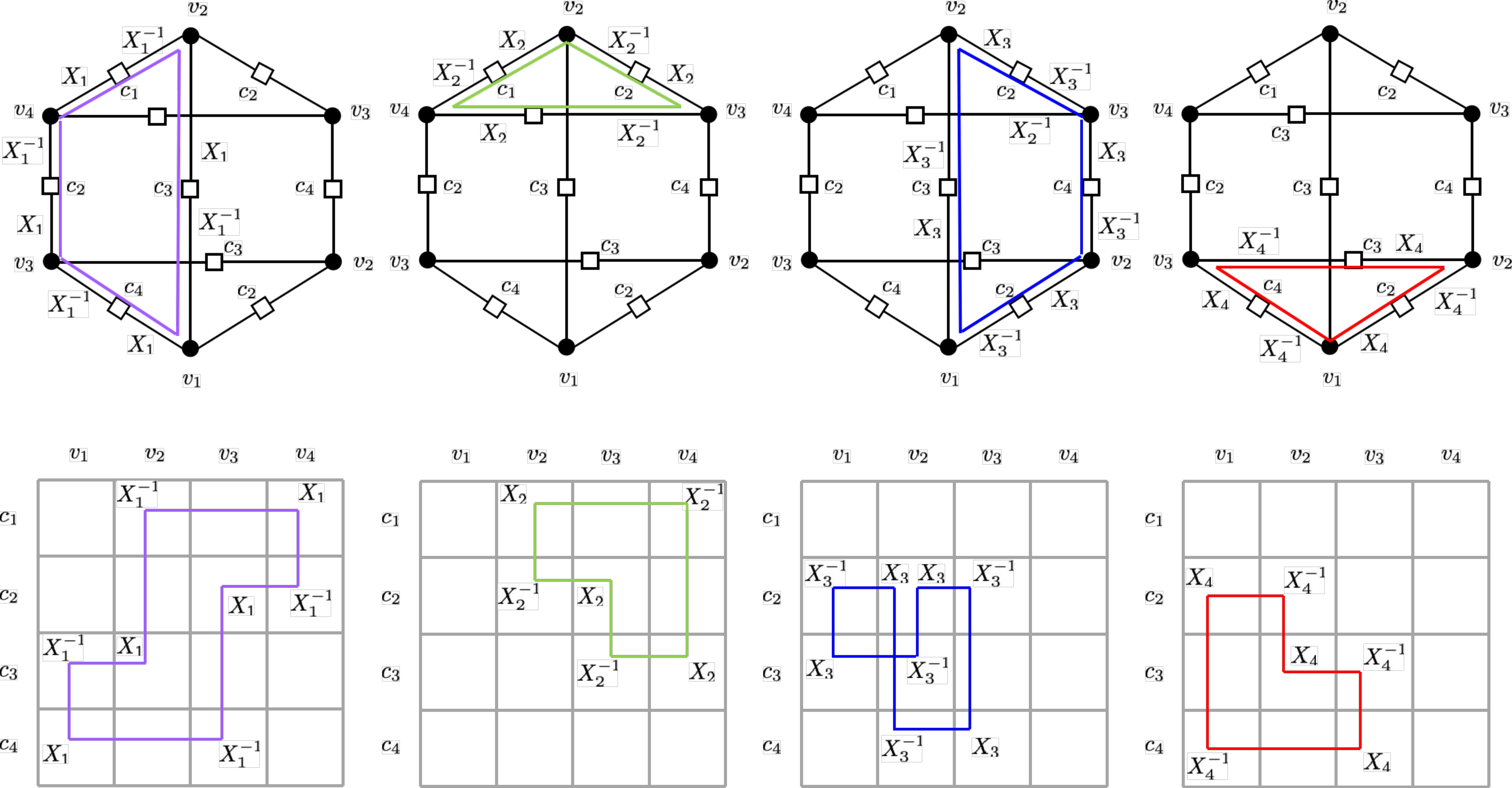}
\caption{The cycle basis of a typical $(6,0)$ ($(6,6)$)-AS in SC codes with $\gamma=3$ ($\gamma=4$) and their corresponding cycle candidates while pulled back to the base matrix. The cycle basis has $4$ fundamental cycles as shown in the top $4$ panels; each cycle decides a cycle candidate in the base matrix and an independent variable in the characteristic polynomial, as shown in the bottom panels. }
\label{fig: cycle_basis}
\end{figure}

\begin{figure}
\centering
\includegraphics[width=0.6\textwidth]{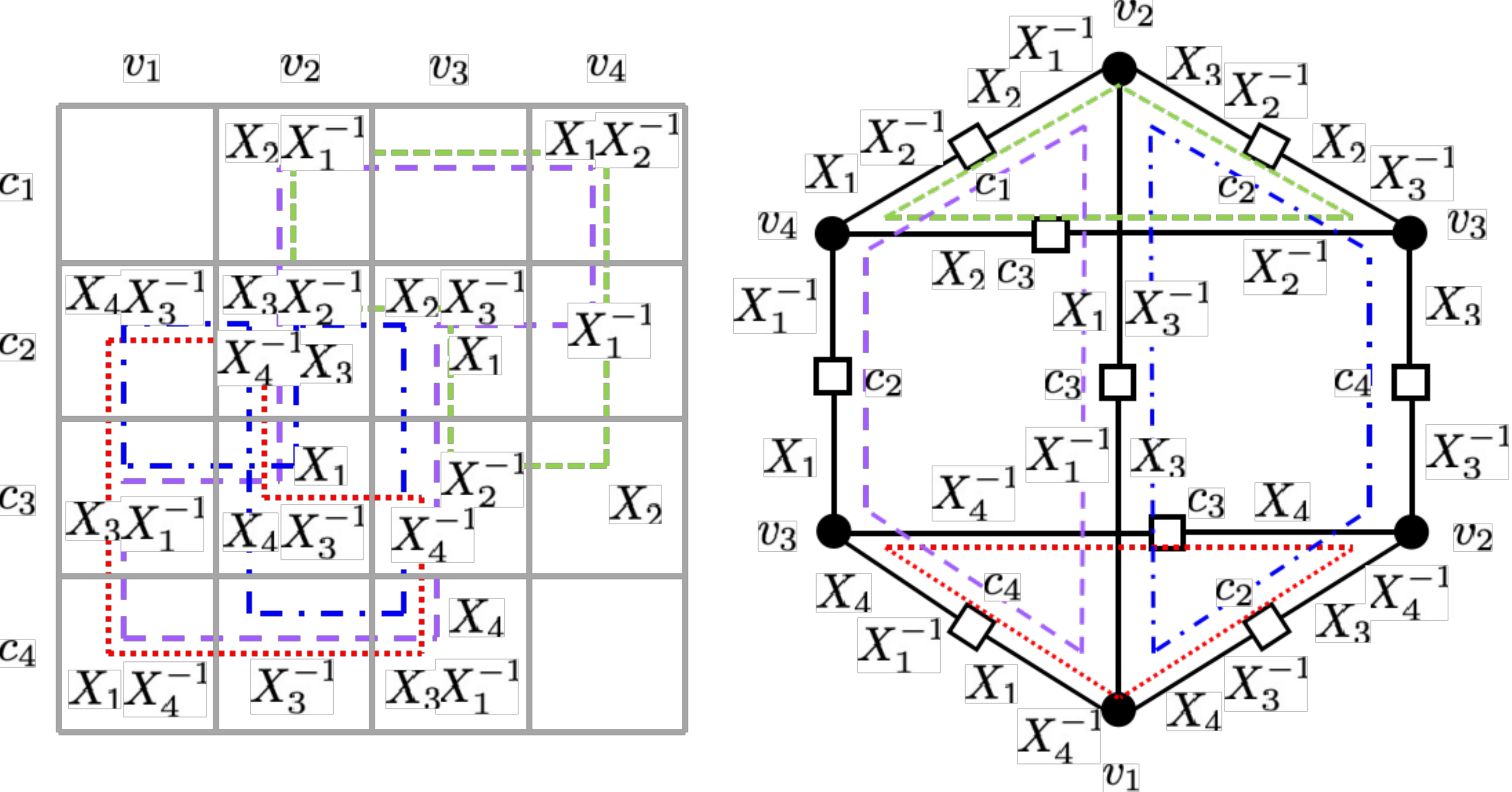}
\caption{The matrix representation of the characteristic polynomial of the AS in Fig.~\ref{fig: cycle_basis}. The monomial in each entry corresponds to a factor in the characteristic polynomial in (\ref{eqn: exam characteristic polynomial}).}
\label{fig: matrix representation}
\end{figure}

\begin{exam} \label{exam: exam general characteristic polynomial} Take the object (AS) with the node assignment shown in the top panel of Fig.~\ref{fig: cycle_basis} as an example.\footnote{Note that we only keep nodes with intrinsic connections, i.e., we ignored the degree $1$ CNs as they are not involved in any cycles and thus do not affect the probability of a prototype becoming active.} Consider the cycle basis consisting of the $4$ cycles highlighted in Fig.~\ref{fig: cycle_basis} and their associated variables $X_i$, $1\leq i\leq 4$. We refer to the cycle associated with $X_i$, $1\leq i\leq 4$, as cycle $i$. In the bottom panel of Fig.~\ref{fig: cycle_basis}, the labels $X_i$ and $X^{-1}_i$ are placed alternately on the cycle candidate corresponding to cycle $i$ in the base matrix.

We next briefly and intuitively explain how the characteristic polynomial of the object is specified as follows:
\begin{equation}
\label{eqn: exam characteristic polynomial}
\begin{split}
h(\mathbf{X})=&f(X_2^{-1}X_3^2X_4^{-1})f(X_1X_2X_3^{-1})f(X_1X_3^{-1}X_4)f(X_1^{-1}X_3X_4)\\
&f(X_1X_2^{-1})f(X_1^{-1}X_2)f^2(X_1^{-1}X_3)f(X_1X_4^{-1})f(X_2^{-1}X_4^{-1})f(X_3^{-1}X_4)\\
&f(X_1^{-1})f(X_2)f(X_3^{-1}),
\end{split}
\end{equation}
where $f(\cdot)$ is the coupling polynomial of a cycle as specified in \Cref{defi: coupling polynommial}. 

As shown in Fig.~\ref{fig: matrix representation}, we place the labels on all the cycle candidates (see Fig.~\ref{fig: cycle_basis}) altogether in the base matrix. Then, each entry of the matrix becomes associated with the product of all the labels contained in it. Take the entry at the intersection of row $c_3$ and column $v_2$ as an example. This entry is labeled with $X_1$, $X_3^{-1}$, $X_4$ on the cycle candidates for cycles $1$, $3$, and $4$, respectively. Therefore, the entry is associated with $X_1X_3^{-1}X_4$. Each product is the monomial corresponding to the matrix entry. The characteristic polynomial in (\ref{eqn: exam characteristic polynomial}) is exactly the product of all the factors obtained by replacing the variable in the coupling polynomial by the monomials corresponding to each entry.

In a way similar to the process described in the proof of \Cref{theo: cycle 6 probability}, expanding the right-hand side (RHS) of (\ref{eqn: exam characteristic polynomial}) results in terms of the form $q_iX_1^{k_1}X_2^{k_2}X_3^{k_3}X_4^{k_4}$ for each, where $q_i$ is the probability of a unique assignment to vertices, i.e., matrix entries, on the prototype of the AS such that the alternating sum of entries on cycle $i$ associated with $X_i$ in the cycle basis is $k_i$, $1\leq i\leq 4$. Therefore, the constant term is exactly the sum of the probabilities of all possible assignments (of the partitioning matrix) such that the cycle candidates of all the fundamental cycles satisfy their cycle conditions (all $k_i$'s are zeros). In other words, the constant term is exactly the probability of the prototype becoming active in the Tanner graph. 

\end{exam}

In \Cref{exam: exam general characteristic polynomial}, we have briefly introduced the idea of how we define the characteristic polynomial of an object associated with a fixed prototype. However, as shown in the case of cycle-$8$ candidates, an object is typically associated with multiple prototypes (referred to as cycle candidates when the object is a cycle). We next present an efficient method to obtain the expected number of all possible prototypes corresponding to an object.

The major idea is described as follows. Each prototype of an object leads to an equivalence relation on the CNs and VNs of the object, in which nodes with identical indices are regarded as being equivalent. The set consisting of all the prototypes describing the same equivalence relation is referred to as a \textbf{prototype class}. The characteristic polynomials of the prototypes belonging to the same prototype class are identical, and the cardinality of each prototype class is determined by their associated equivalence relation. Therefore, the key steps to obtain the characteristic polynomial of an object are: 1) to enumerate all the possible prototype classes of (or non-isomorphic equivalence relations on) a given object, and then 2) to obtain their associated characteristic polynomials and cardinalities. 

For example, consider the prototype class described by the graph in the left panel of Fig.~\ref{fig: graph representation}. Throughout this paper, we use $\MB{n}$ to represent the set $\{1,2,\dots,n\}$ for any $n\in\mathbb{N}$. Any assignment of $c_1,c_2,c_3,c_4\in\MB{\gamma}$ and $v_1,v_2,v_3,v_4\in\MB{\kappa}$ such that $c_1,c_2,c_3,c_4$ are mutually different, $v_1,v_2,v_3,v_4$ are also mutually different belongs to a unique prototype in the prototype class described by the graph. Note that the uniqueness follows from the fact that the \textbf{automorphism group} of the prototype class has only the identity element, thus the cardinality of this prototype class is $4!\binom{\gamma}{4}4!\binom{\kappa}{4}$. The automorphism group of a prototype is defined later in \Cref{defi: auto}; however, the aforementioned cardinality intuitively implies that each row/column permutation results in a unique prototype. We then move on to obtain the characteristic polynomial directly through the prototype class. The equivalence relation on CNs and VNs induces an equivalence relation on edges, in which edges with VNs and CNs all from the same equivalence class are referred to as being equivalent. As shown in Fig.~\ref{fig: graph representation}, edges from the same equivalence class are highlighted by identical markers. 

Note that each equivalence class on edges corresponds to a unique entry in the base matrix. Recall that each entry is associated with the product of all labels contained in it, which corresponds to a separate factor in the characteristic polynomial. In a similar way, if we represent each equivalence class by the product of all labels on all edges contained in it, the resultant product will be exactly the monomial associated with the entry corresponding to this equivalence class. Therefore, each factor of the characteristic polynomial in (\ref{eqn: exam characteristic polynomial}) is associated with an equivalence class on edges. For example, in Fig.~\ref{fig: graph representation}, the two edges highlighted by red triangles belong to the same equivalence class and are labeled with $X_1$ and $X_2X_3^{-1}$, respectively; and they altogether correspond to the factor $f(X_1X_2X_3^{-1})$ in the characteristic polynomial.

In the remaining text, we represent the equivalence relation by $\sim$.

\begin{defi} (\textbf{Prototype Class}) Let $\gamma,\kappa\in\mathbb{N}$. Consider an object represented by the bipartite graph $G(V,C,E)$, where $V$ and $C$ denote the set of VNs and CNs, respectively. The set 
$E$ is the set of all edges, where each edge is represented by $e_{i,j}$, for some $i\in V$, $j\in C$, connecting nodes $i$ and $j$. Let $\Se{V},\Se{C}$ represent equivalence classes on $V$ and $C$, respectively. 
A \textbf{prototype} is an assignment $P=(f,g)$, where $f: V\to \MB{\kappa}$, $g: C\to \MB{\gamma}$ such that: 
\begin{enumerate}
\item For any $c\in C$ and $v_1,v_2\in V$ such that $e_{v_1,c},e_{v_2,c}\in E$, $f(v_1)\neq f(v_2)$;
\item For any $v\in V$ and $c_1,c_2\in C$ such that $e_{v,c_1},e_{v,c_2}\in E$, $g(c_1)\neq g(c_2)$.
\end{enumerate}

The set consisting of all the prototypes $P=(f,g)$ that satisfy the following conditions is referred to as the \textbf{prototype class} associated with $(\Se{V},\Se{C})$, denoted by $\Se{P}(\Se{V},\Se{C})$:
\begin{enumerate}
\item For any $v_1,v_2\in V$, $f(v_1)=f(v_2)$ iff. $v_1\sim v_2$ in $\Se{V}$ (same column in the matrix);
\item For any $c_1,c_2\in C$, $g(c_1)=g(c_2)$ iff. $c_1\sim c_2$ in $\Se{C}$ (same row in the matrix).
\end{enumerate}

Denote the equivalence class induced by the relation: $e_{v_1,c_1}\sim e_{v_2,c_2}$ iff. $v_1\sim v_2$ and $c_1\sim c_2$, for all $v_1, v_2 \in V$ and $c_1, c_2 \in C$, by $\Se{E}(\Se{V},\Se{C})$, which is referred to as the equivalence class induced by $\Se{V}$ and $\Se{C}$.
\end{defi}

\begin{figure}
\centering
\includegraphics[width=0.6\textwidth]{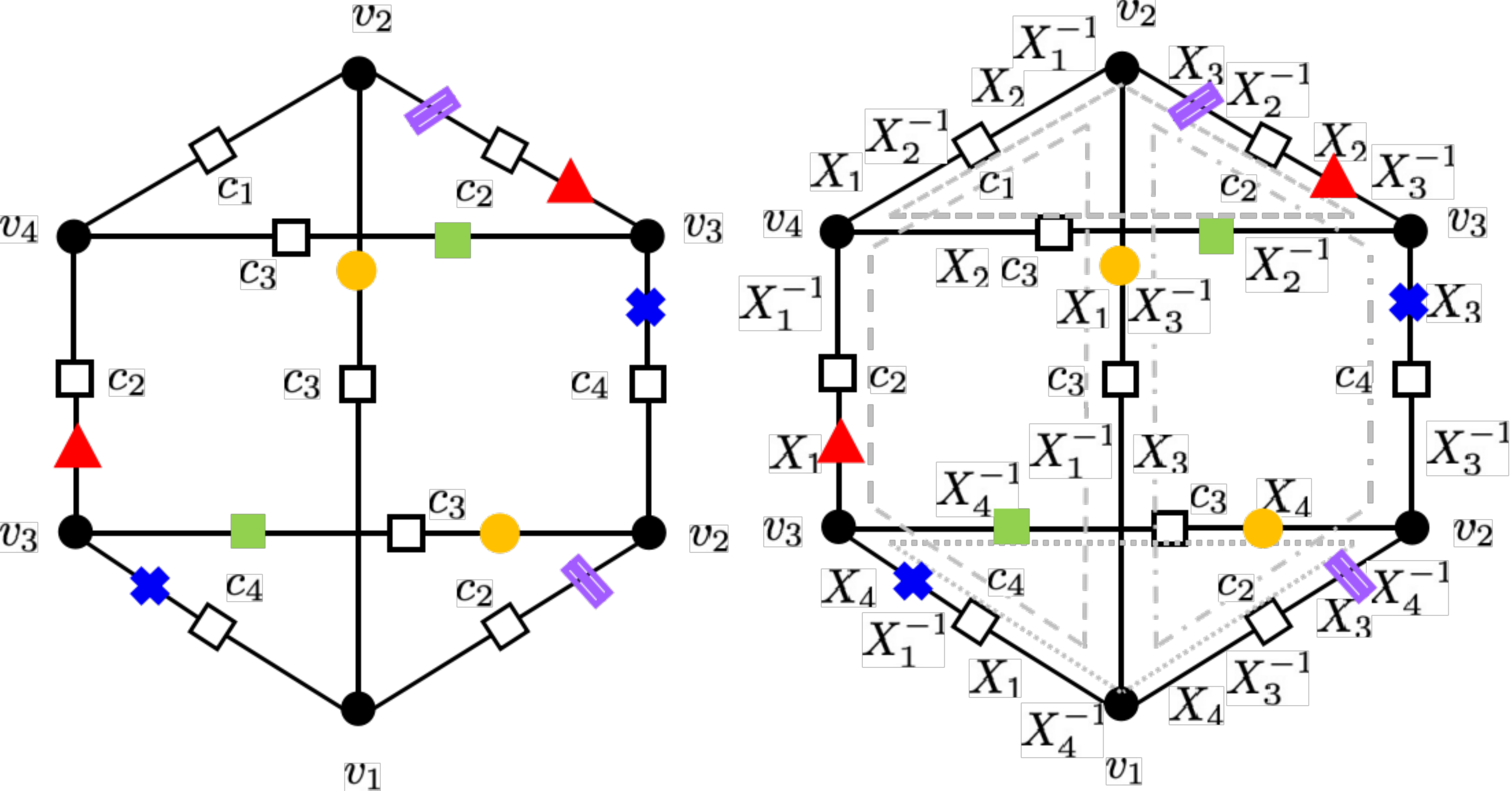}
\caption{The graph representation of the characteristic polynomial of the AS in Fig.~\ref{fig: cycle_basis}}
\label{fig: graph representation}
\end{figure}

\begin{lemma} \label{theo: Characteristic Polynomial of Prototypes}(\textbf{Characteristic Polynomial of Prototypes}) Consider the bipartite graph $G(V,C,E)$ of an object and the prototype class $\Se{P}(\Se{V},\Se{C})$. Suppose $\Se{E}(\Se{V},\Se{C})$ is the equivalence class on edges induced by $\Se{V}$ and $\Se{C}$.

Let $S$ denote the cycle basis of $G$. Define $\delta: E \times S \to \{-1,0,1\}$ as follows: for any $s\in S$, $s=(v_1,c_1,v_2,c_2,\dots,v_{g},c_g)$, and $e\in E$, $\delta_{e,s}=1$ if $e=e_{v_i,c_i}$ for some $i\in \MB{g}$, $\delta_{e,s}=-1$ if $e=e_{v_{i+1},c_i}$ for some $i\in \MB{g}$, otherwise $\delta_{e,s}=0$.

Define $h(\mathbf{X};G|\Se{V},\Se{C})$ as a polynomial of $G$ associated with $\Se{P}(\Se{V},\Se{C})$ and given by:
\begin{equation}\label{c_ch_poly}
h(\mathbf{X};G|\mathcal{V},\mathcal{C})=\prod_{\bar{e}\in\mathcal{E}(\mathcal{V},\mathcal{C})}f\left(\prod_{e\in \bar{e}}\prod_{s\in S} X_s^{\delta_{e,s}}\right).
\end{equation}
Then, the constant term of $h(\mathbf{X};G|\Se{V},\Se{C})$ is the probability that a prototype belonging to the class $\Se{P}(\Se{V},\Se{C})$ is active in the Tanner graph after partitioning.
\end{lemma}

\begin{proof} For simplicity, we write $\mathcal{E}$ instead of $\mathcal{E}(\mathcal{V},\mathcal{C})$ in the proof. Any assignment on the set of edges $E$ can be represented by $\mathbf{x}\in (x_1,x_2,\dots,x_{|E|})\in \Tx{vals}(\mathbf{a})^{|E|}$ ($\Tx{vals}(\mathbf{a})$ is defined in \Cref{theo: cycle 6 probability} as the set $\{a_0,a_1,\dots,a_{m_t}\}$), where $x_e$ denotes the assignment on edge $e$ for any $e\in E$. Consider that all the edges belonging to the same equivalence class in $\mathcal{E}$ correspond to the same entry in the base matrix (and the partitioning matrix); these edges need to be assigned with an identical number in the partitioning matrix. Therefore, the assignment on the set of edges is essentially an assignment on the equivalence classes. Let $\mathbf{i}\in\{0,1,\dots,m_t\}^{|\mathcal{E}|}$ denote an assignment on the equivalence classes $\mathcal{E}$, where each element of $\mathbf{i}$ is represented by $i_{\bar{e}}$ for some $\bar{e}\in\mathcal{E}$; all the edges in the equivalence class $\bar{e}$ are assigned with $a_{i_{\bar{e}}}$ in the partitioning matrix, for any $\bar{e}\in\mathcal{E}$. 

We know that
\begin{equation}
\begin{split}
h(\mathbf{X};G|\mathcal{V},\mathcal{C})&=\prod_{\bar{e}\in\mathcal{E}}f\left(\prod_{e\in \bar{e}}\prod_{s\in S} X_s^{\delta_{e,s}}\right)\\
&=\sum_{\mathbf{i}\in\{0,1,\dots,m_t\}^{|\mathcal{E}|}}\hspace{0.5em}\prod_{\bar{e}\in\mathcal{E}} \left[p_{i_{\bar{e}}} \prod_{e\in \bar{e}}\prod_{s\in S} X_s^{\delta_{e,s}a_{i_{\bar{e}}}}\right]\\
&=\sum_{\mathbf{i}\in\{0,1,\dots,m_t\}^{|\mathcal{E}|}}\left(\prod_{\bar{e}\in\mathcal{E}} p_{i_{\bar{e}}}\right) \prod_{\bar{e}\in\mathcal{E}}\prod_{e\in \bar{e}}\prod_{s\in S} X_s^{\delta_{e,s}a_{i_{\bar{e}}}}\\
&=\sum_{\mathbf{i}\in\{0,1,\dots,m_t\}^{|\mathcal{E}|}}\left(\prod_{\bar{e}\in\mathcal{E}} p_{i_{\bar{e}}}\right) \prod_{s\in S} X_s^{\sum_{\bar{e}\in\mathcal{E}} a_{i_{\bar{e}}}\sum_{e\in\bar{e}}\delta_{e,s}}\\
&=\sum_{\mathbf{i}\in\{0,1,\dots,m_t\}^{|\mathcal{E}|}}\left(\prod_{\bar{e}\in\mathcal{E}} p_{i_{\bar{e}}}\right) \prod_{s\in S} X_s^{l_s(\mathbf{i})},\\
\end{split}
\end{equation}
where $l_s(\mathbf{i})=\sum_{\bar{e}\in\mathcal{E}} a_{i_{\bar{e}}}\sum_{e\in\bar{e}}\delta_{e,s}=\sum_{e\in s, e\in \bar{e}} \delta_{e,s}a_{i_{\bar{e}}}$ is exactly the alternating sum of the assignment $\mathbf{i}$ on cycle $s$.

Denote by $\mathcal{Z}(G)$ the set of assignments of the prototype in the partitioning matrix such that this prototype becomes active in the Tanner graph of the code after partitioning. Then,

\begin{equation}
\begin{split}
\left[h(\mathbf{X};G|\mathcal{V},\mathcal{C})\right]_0=&\sum_{\mathbf{i}\in\{0,1,\dots,m_t\}^{|\mathcal{E}|}: l_s(\mathbf{i})=0, \forall s\in S}\hspace{0.5em}\prod_{\bar{e}\in \mathcal{E}} p_{i_{\bar{e}}}\\
=&\sum_{\mathbf{i}\in\{0,1,\dots,m_t\}^{|\mathcal{E}|}: l_s(\mathbf{i})=0, \forall s\in S}\mathbb{P}\left[x_{e}=a_{i_{\bar{e}}},\forall \bar{e}\in\mathcal{E}, e\in \bar{e}\right]\\
=&\mathbb{P}\left[\mathcal{Z}(G)\right],
\end{split}
\end{equation}
which indicates that the constant term of $h(\mathbf{X};G|\mathcal{V},\mathcal{C})$ is the probability we are seeking.
\end{proof}

In fact, the coefficients of other terms of $h(\mathbf{X};G|\mathcal{V},\mathcal{C})$ also specify the probabilities of partitioning assignments other than $\mathcal{Z}(G)$. Consequently, $h(\mathbf{X};G|\Se{V},\Se{C})$ in (\ref{c_ch_poly}) represents the characteristic polynomial of $G$ associated with $\Se{P}(\Se{V},\Se{C})$.

\begin{figure}
\centering
\includegraphics[width=0.6\textwidth]{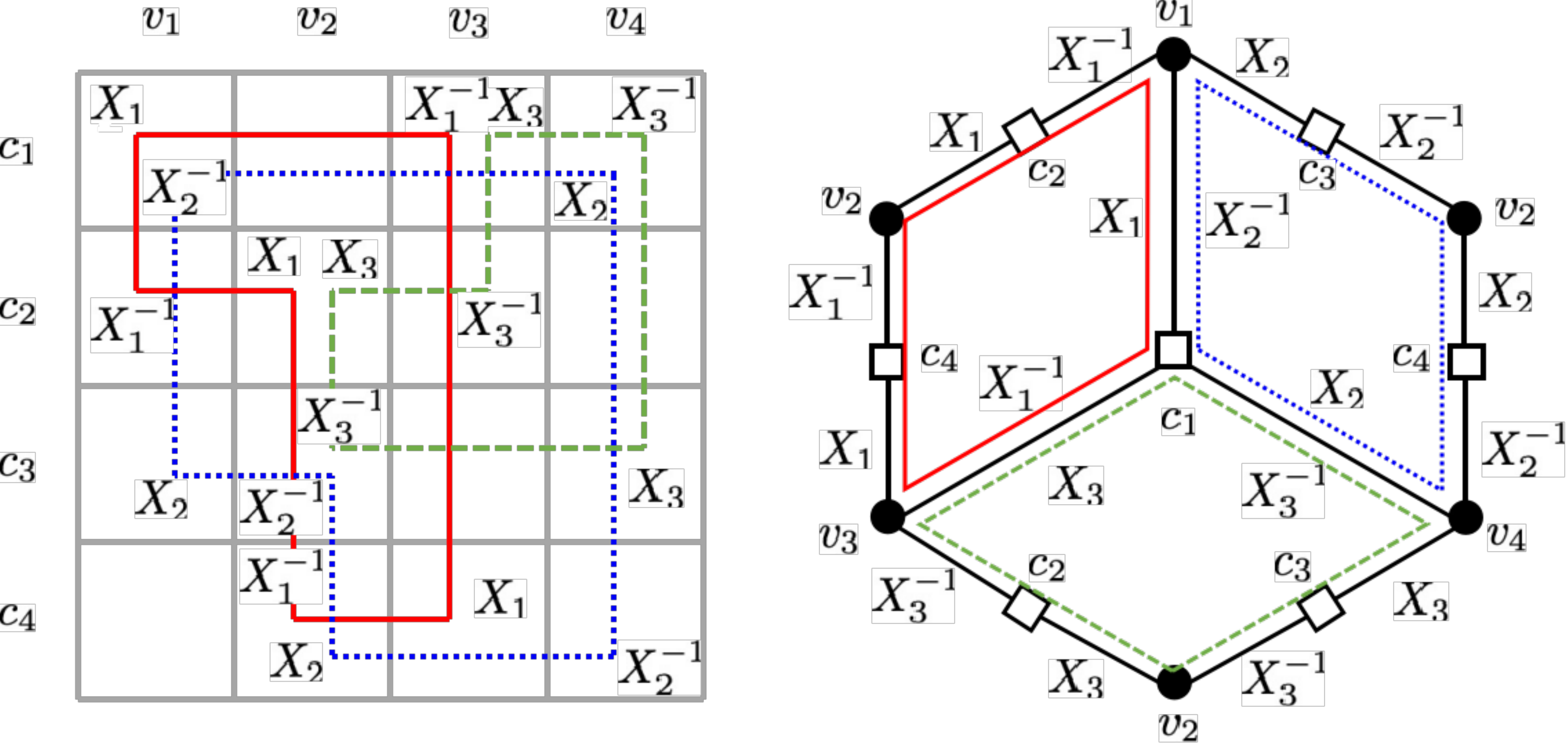}\vspace{-0.7em}
\caption{A matrix representation of the characteristic polynomial of the AS in \Cref{rem: nele}.}
\label{fig: nonelementary objects}
\end{figure}

While elementary objects (absorbing sets in particular) dominate the error floor of binary LDPC codes and NB-LDPC codes over the AWGN channel, non-elementary objects are observed to notably contribute to the error floor of NB-LDPC codes over non-canonical channels, e.g., practical magnetic recording and Flash channels \cite{hareedy2016nonbinary,hareedy2016general,hareedy2019combinatorial}. 

\begin{rem} \label{rem: nele}(\textbf{Non-Elementary Objects}) Note that \Cref{theo: Characteristic Polynomial of Prototypes} extends beyond elementary objects. Fig.~\ref{fig: nonelementary objects} shows a prototype of a non-elementary object. This prototype can still be described by a set of $3$ elementary cycles, as shown in Fig.~\ref{fig: nonelementary objects} via colors. According to \Cref{theo: Characteristic Polynomial of Prototypes}, the characteristic polynomial of the prototype is:
\begin{equation}
\begin{split}
h(\mathbf{X};G|\Se{V},\Se{C})=&f(X_1X_2^{-1})f(X_2X_3^{-1})f(X_3X_1^{-1})f(X_1X_3)f(X_1^{-1}X_2)f(X_2^{-1}X_3^{-1})\\
&f(X_1)f(X_1^{-1})f(X_2)f(X_2^{-1})f(X_3)f(X_3^{-1}).
\end{split}
\end{equation}

In combination with the WCM framework proposed in \cite{hareedy2019combinatorial} that optimizes the edge weights of NB-LDPC codes on fixed unweighted graphs, our method can open a door to systematically optimizing NB-SC codes with high memories, which have potential to be adopted in storage systems among other applications.  
\end{rem}

After obtaining the characteristic polynomial of any object associated with a fixed prototype class, we proceed to obtain the expectation of the number of active prototypes over all prototype classes. The essential step here is to calculate the cardinality of each prototype class. A natural property here is that each prototype from a specific class corresponds to assigning non-repeated elements with order from $\MB{\kappa}$ and $\MB{\gamma}$ to the equivalence classes in $\Se{V}$ and $\Se{C}$, respectively. However, specific permutations of values assigned to the nodes can lead to some other assignments that are isomorphic to each other because of the intrinsic symmetry of the prototypes. For example, in Fig.~\ref{fig: cc_example}(a), the assignment that exchanges values $v_1$ and $v_2$ while keeping values on remaining VNs as they are is equivalent to the original assignment. We call this exchange operation an \textbf{automorphism} over $G$ under $\Se{P}(\Se{V},\Se{C})$ and denote it by $(v_1v_2)$. The automorphisms over $G$ under each prototype class form a group, which is defined in \Cref{defi: auto}.

\begin{defi} \label{defi: auto} (\textbf{Automorphism Group of an Object Under a Prototype Class}) For any object represented by a bipartite graph $G(V,C,E)$, let $\mathcal{P}(\mathcal{V},\mathcal{C})$ be a prototype class of $G$. An \textbf{automorphism} over $G$ under $\mathcal{P}(\mathcal{V},\mathcal{C})$ is a pair of bijections $(\pi_V,\pi_C)$ written as $\pi_V\pi_C$, where $\pi_V: V\to V$ and $\pi_C: C\to C$ are bijections such that
\begin{enumerate}
  \item $\forall v\in V$, $c\in C$, $e_{v,c}\in E$ iff. $e_{\pi_V(v),\pi_C(c)}\in E$;
  \item $\forall v_1,v_2\in V$, $v_1\sim v_2$ iff. $\pi_V(v_1)\sim \pi_V(v_2)$;
  \item $\forall c_1,c_2\in C$, $c_1\sim c_2$ iff. $\pi_C(c_1)\sim \pi_C(c_2)$.
  \end{enumerate}
  The set containing all automorphisms over $G$ under $\mathcal{P}(\mathcal{V},\mathcal{C})$ is referred to as the \textbf{automorphism group} of $G$ under $\mathcal{P}(\mathcal{V},\mathcal{C})$.
\end{defi}

\begin{rem} From \Cref{defi: auto}, we know that any automorphism over $G$ under $\mathcal{P}(\mathcal{V},\mathcal{C})$ preserves the equivalence relation specified by $(\Se{V},\Se{C})$, i.e., $\{\pi_V(\bar{v}), \forall \bar{v}\in \Se{V}\}=\Se{V}$, $\{\pi_C(\bar{c}), \forall \bar{c}\in \Se{C}\}=\Se{C}$. Therefore, each automorphism can be simply represented as a pair of permutations over $\Se{V}$ and $\Se{C}$.
\end{rem}

\begin{lemma}{\label{lemma: general_couting}} Given the bipartite graph $G(V,C,E)$ of an object, let $\Se{B}(G)$ denote the set consisting of all prototype classes of $G$. Define the characteristic polynomial $h(\mathbf{X};G)$ of object $G$ as follows:
\begin{equation}
h(\mathbf{X};G)=\sum_{\Se{P}(\mathcal{V},\mathcal{C})\in \mathcal{B}(G)} \frac{|\mathcal{V}|!|\mathcal{C}|!}{|\textup{Aut}(G|\mathcal{V},\mathcal{C})|}\binom{\kappa}{|\mathcal{V}|}\binom{\gamma}{|\mathcal{C}|} h(\mathbf{X};G|\mathcal{V},\mathcal{C}),
\end{equation}
where $\textup{Aut}(G|\mathcal{V},\mathcal{C})$ denotes the automorphism group of the bipartite graph $G$ under prototype class $\mathcal{P}(\mathcal{V},\mathcal{C})$. Then, $\left[h(\mathbf{X};G)\right]_0$ is exactly the expected number of active prototype of object $G$.  
\end{lemma}

\begin{proof} There are $|\mathcal{V}|!|\mathcal{C}|!\binom{\kappa}{|\mathcal{V}|}\binom{\gamma}{|\mathcal{C}|}$ assignments on indices of nodes in $G$ in the base matrix that satisfy the equivalence relation specified by $(\Se{V},\Se{C})$. Among these assignments, each one has been counted exactly $|\textup{Aut}(G|\mathcal{V},\mathcal{C})|$ times. Therefore, the cardinality of the prototype class $\mathcal{P}(\mathcal{V},\mathcal{C})$ is exactly $\frac{|\mathcal{V}|!|\mathcal{C}|!}{|\textup{Aut}(G|\mathcal{V},\mathcal{C})|}\binom{\kappa}{|\mathcal{V}|}\binom{\gamma}{|\mathcal{C}|} $. 

For any prototype $P$ of $G$, define a Bernoulli random variable $X_P$, where $\mathbb{P}\left[X_P=1\right]=\mathbb{P}\left[P\text{ is active}\right]$, $\mathbb{P}\left[X_P=0\right]=\mathbb{P}\left[P\text{ is not active}\right]$. Let $X=\sum_{P}X_P$ denotes the summation of $X_P$ over all possible prototypes of $G$. Then,

\begin{equation}
\begin{split}
\mathbb{E}\left[X\right]&=\sum_P\mathbb{E}\left[X_P\right]\\
&=\sum_{\Se{P}(\mathcal{V},\mathcal{C})\in \mathcal{B}(G)}\sum_{P\in \Se{P}(\mathcal{V},\mathcal{C})}\mathbb{E}\left[X_P\right]\\
&=\sum_{\Se{P}(\mathcal{V},\mathcal{C})\in \mathcal{B}(G)}\sum_{P\in \Se{P}(\mathcal{V},\mathcal{C})}\mathbb{P}\left[X_P=1\right]\\
&=\sum_{\Se{P}(\mathcal{V},\mathcal{C})\in \mathcal{B}(G)}\sum_{P\in \Se{P}(\mathcal{V},\mathcal{C})}\left[h(\mathbf{X};G|\mathcal{V},\mathcal{C})\right]_0\\
&=\sum_{\Se{P}(\mathcal{V},\mathcal{C})\in \mathcal{B}(G)}\frac{|\mathcal{V}|!|\mathcal{C}|!}{|\textup{Aut}(G|\mathcal{V},\mathcal{C})|}\binom{\kappa}{|\mathcal{V}|}\binom{\gamma}{|\mathcal{C}|} \left[h(\mathbf{X};G|\mathcal{V},\mathcal{C})\right]_0\\
&=\left[h(\mathbf{X};G)\right]_0.
\end{split}
\end{equation}
Thus, the lemma is proved.
\end{proof}

\begin{figure}
\centering
  \subfigure[Prototype class $1$.]{
       \includegraphics[width=0.23\textwidth]{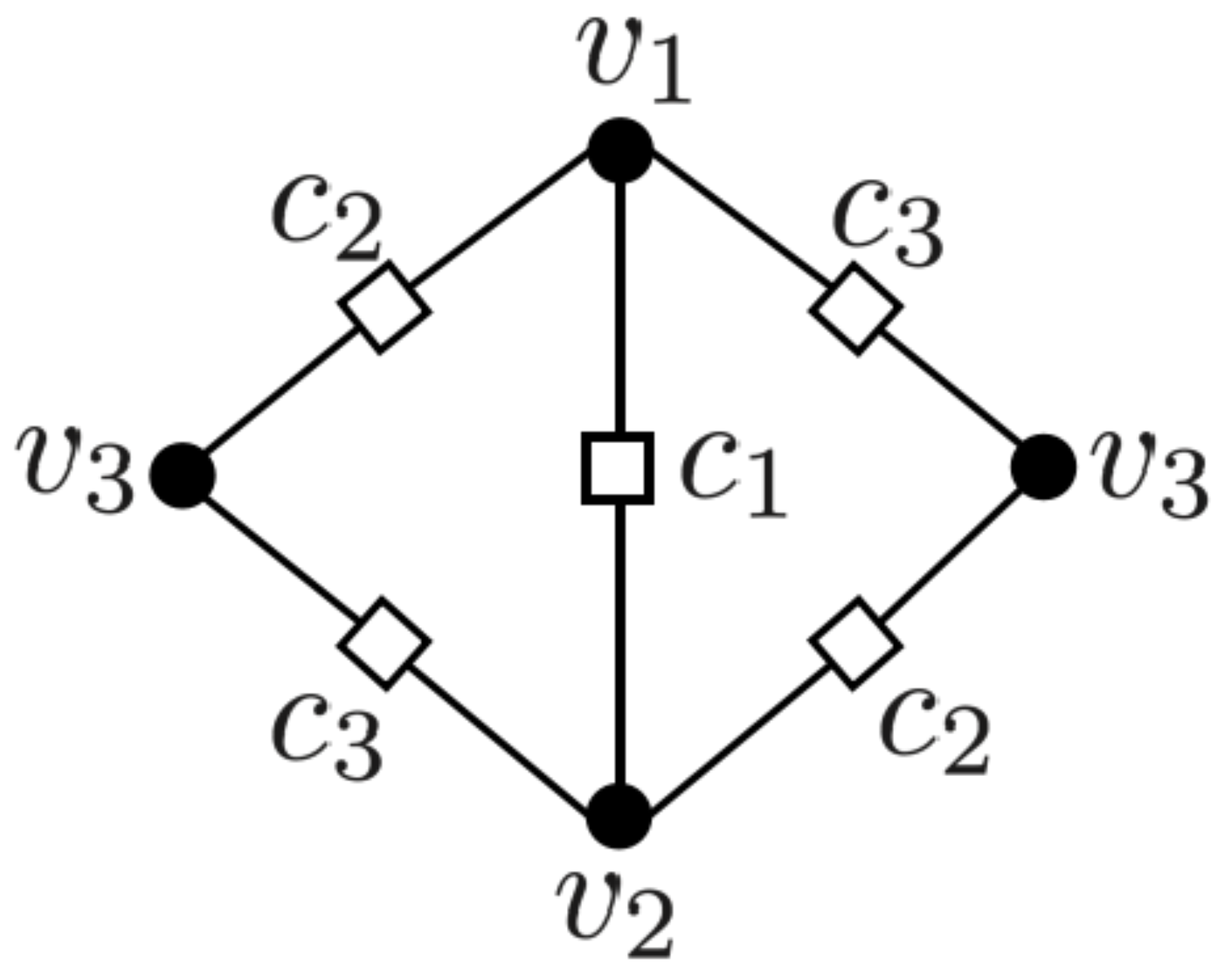}
       \label{1a}
       }
  \subfigure[Prototype class $2$.]{
        \includegraphics[width=0.23\textwidth]{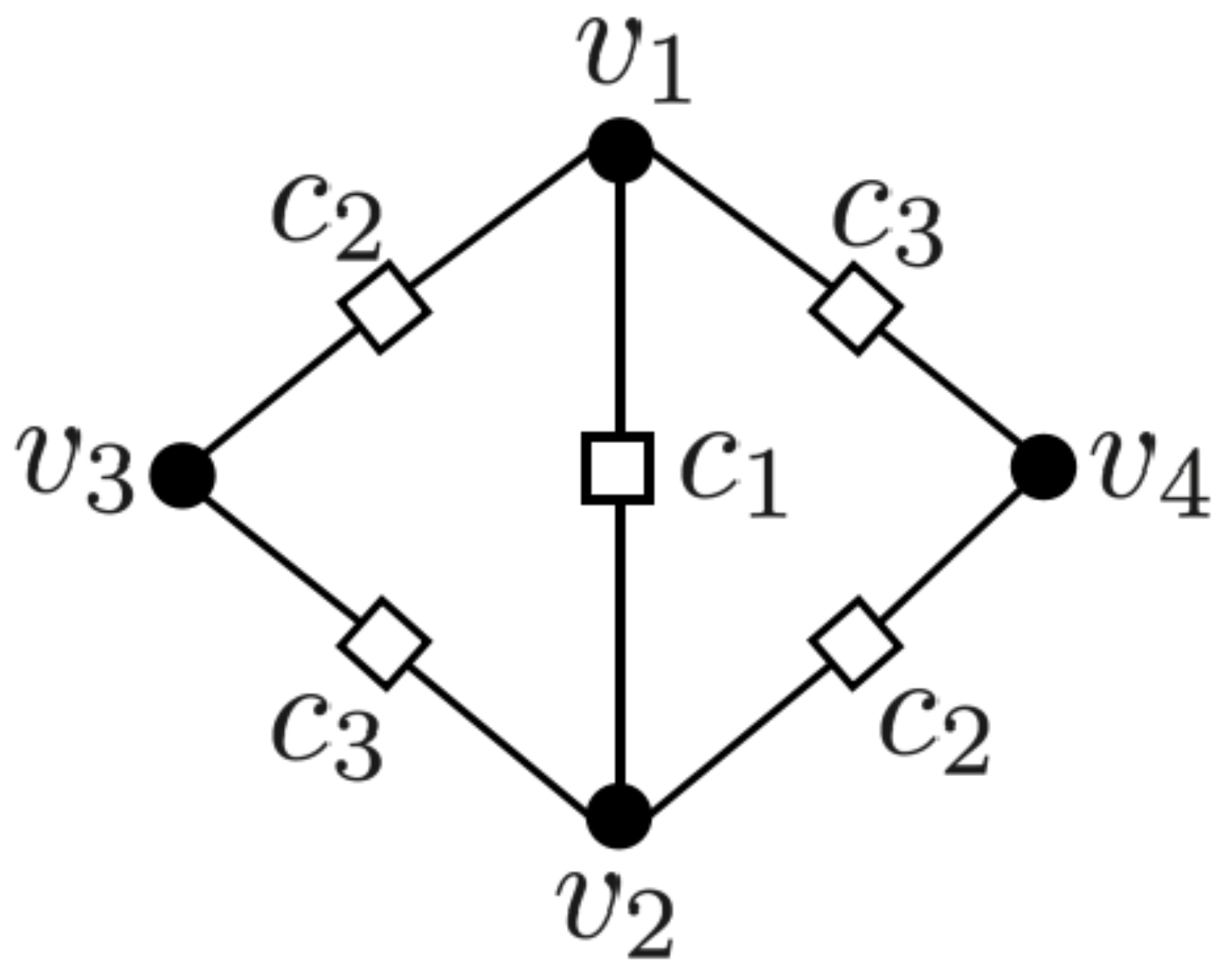}
        \label{1b}
        }
  \subfigure[Prototype class $3$.]{
       \includegraphics[width=0.23\textwidth]{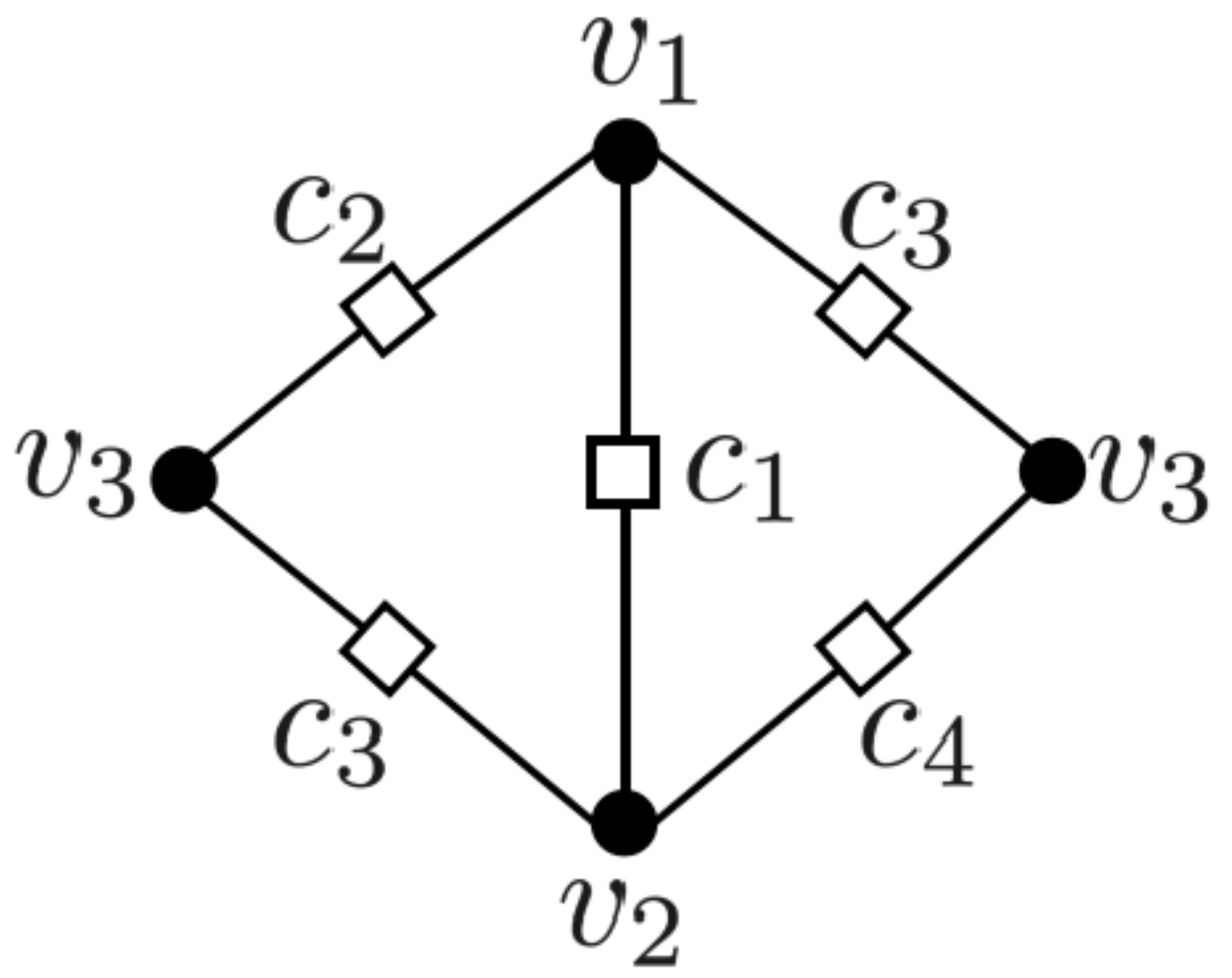}
       \label{1a}
       }
  \subfigure[Prototype class $4$.]{
        \includegraphics[width=0.23\textwidth]{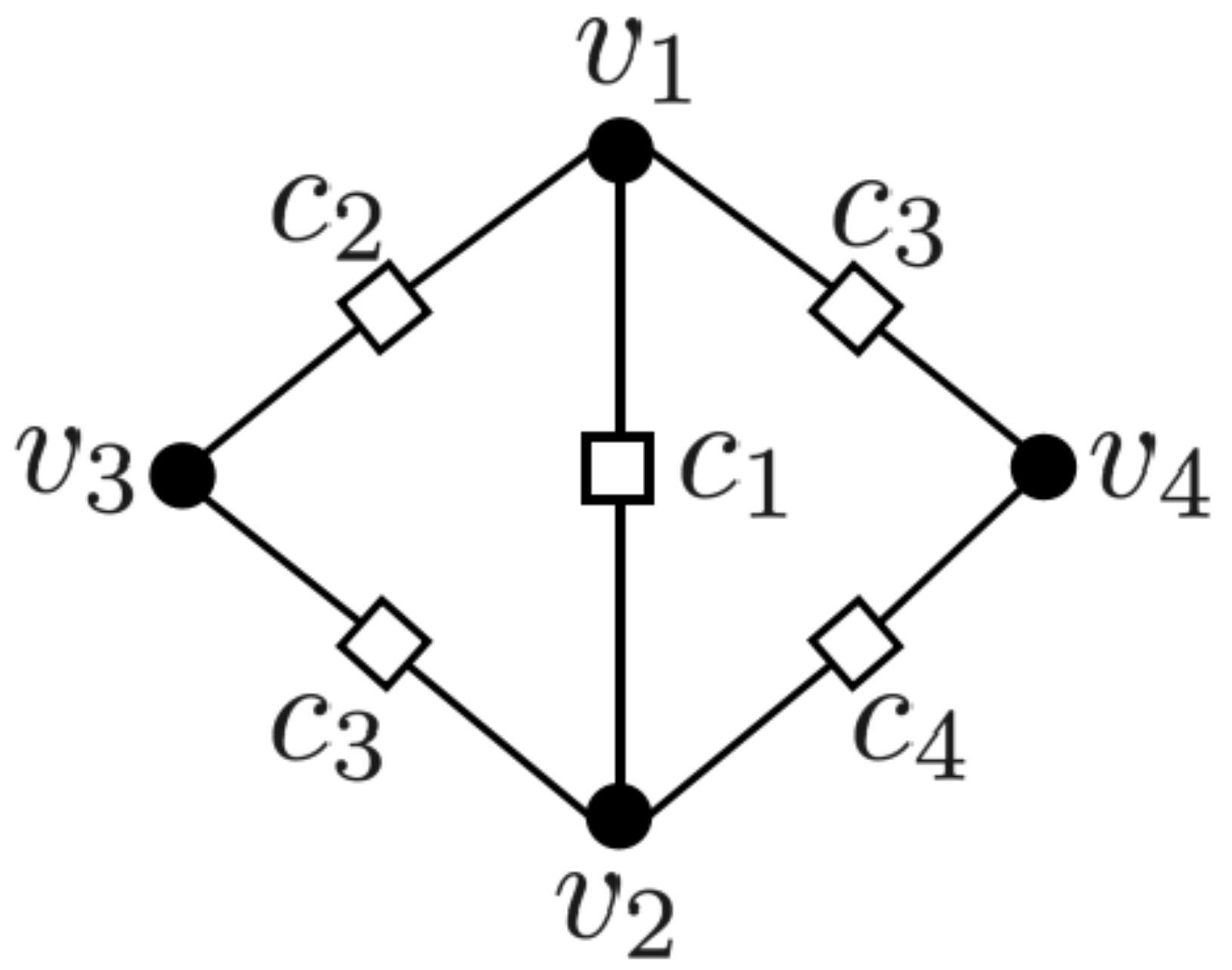}
        \label{1b}
        }
  \caption{The $4$ prototype classes of $2$ concatenated cycles-$6$ and their corresponding topologies.}
  \label{fig: cc_example} 
\end{figure}

\begin{exam}\label{exam: cc_6} Suppose $\gamma\in\{3,4\}$. Take the graph $G$ of two concatenated cycles-$6$ as an example; there are $4$ different possible prototype classes $(\mathcal{V}_i,\mathcal{C}_i)$, $1\leq i\leq 4$, as shown in Fig.~\ref{fig: cc_example}. The automorphism groups corresponding to the $4$ prototype classes, denote by $|\textup{Aut}(G|\mathcal{V}_i,\mathcal{C}_i)|$, $1\leq i\leq 4$, respectively, are: 
\begin{equation}
\begin{split}
\textup{Aut}(G|\mathcal{V}_1,\mathcal{C}_1)&=\{e,(v_1v_2),(c_2c_3),(v_1v_2)(c_1c_3)\},\\
\textup{Aut}(G|\mathcal{V}_2,\mathcal{C}_2)&=\{e,(v_1v_2)(c_2c_3),(v_3v_4)(c_2c_3),(v_1v_2)(v_3v_4)\},\\
\textup{Aut}(G|\mathcal{V}_3,\mathcal{C}_3)&=\{e,(v_1v_2)(c_2c_4)\},\\
\textup{Aut}(G|\mathcal{V}_4,\mathcal{C}_4)&=\{e,(v_1v_2)(v_3v_4)(c_2c_4)\},\\
\end{split}
\end{equation}
where the element $e$ in these groups is the identity element (no permutations). Therefore, the cardinality of the automorphism groups corresponding to the $4$ prototype classes are $|\textup{Aut}(G|\mathcal{V}_1,\mathcal{C}_1)|=4$, $|\textup{Aut}(G|\mathcal{V}_2,\mathcal{C}_2)|=4$, $|\textup{Aut}(G|\mathcal{V}_3,\mathcal{C}_3)|=2$, and $|\textup{Aut}(G|\mathcal{V}_4,\mathcal{C}_4)|=2$, respectively. Moreover, the characteristic polynomials for $G$ corresponding to each prototype class are:
\begin{equation}
\begin{split}
h(\mathbf{X};G|\mathcal{V}_1,\mathcal{C}_1)=&f(X_1X_2)f(X_1^{-1}X_2^{-1})f(X_1X_2^{-1})f(X_1^{-1}X_2)f(X_1)f(X_1^{-1})f(X_2)f(X_2^{-1}),\\
h(\mathbf{X};G|\mathcal{V}_3,\mathcal{C}_3)=&f(X_1X_2)f(X_1^{-1}X_2^{-1})f(X_1^{-1}X_2)f^2(X_1)f(X_1^{-1})f(X_2)f^2(X_2^{-1}),\\
h(\mathbf{X};G|\mathcal{V}_2,\mathcal{C}_2)=&h(\mathbf{X};G|\mathcal{V}_4,\mathcal{C}_4)=f(X_1X_2)f(X_1^{-1}X_2^{-1})f^2(X_1)f^2(X_1^{-1})f^2(X_2)f^2(X_2^{-1}).\\
\end{split}
\end{equation}

According to \Cref{lemma: general_couting}, when $\gamma=3$, the characteristic polynomial $h(\mathbf{X};G)$ is derived as follows:
\begin{equation}
\begin{split}
h(\mathbf{X};G)=&\frac{\kappa!\gamma!}{4(\kappa-3)!(\gamma-3)!}f(X_1X_2)f(X_1^{-1}X_2^{-1})f(X_1X_2^{-1})f(X_1^{-1}X_2)f(X_1)f(X_1^{-1})f(X_2)f(X_2^{-1})\\
&+\frac{\kappa!\gamma!}{4(\kappa-4)!(\gamma-3)!}f(X_1X_2)f(X_1^{-1}X_2^{-1})f^2(X_1)f^2(X_1^{-1})f^2(X_2)f^2(X_2^{-1})\\
=&\frac{3\kappa!}{2(\kappa-4)!}\bigg(f(X_1X_2)f(X_1^{-1}X_2^{-1})f^2(X_1)f^2(X_1^{-1})f^2(X_2)f^2(X_2^{-1})\\
&+\frac{1}{\kappa-3}f(X_1X_2)f(X_1^{-1}X_2^{-1})f(X_1X_2^{-1})f(X_1^{-1}X_2)f(X_1)f(X_1^{-1})f(X_2)f(X_2^{-1})\bigg).
\end{split}
\end{equation}
When $\gamma=4$, the characteristic polynomial $h(\mathbf{X};G)$ is derived as follows:
\begin{equation}
\begin{split}
h(\mathbf{X};G)=&\frac{\kappa!\gamma!}{4(\kappa-3)!(\gamma-3)!}f(X_1X_2)f(X_1^{-1}X_2^{-1})f(X_1X_2^{-1})f(X_1^{-1}X_2)f(X_1)f(X_1^{-1})f(X_2)f(X_2^{-1})\\
&+\frac{\kappa!\gamma!}{2(\kappa-3)!(\gamma-4)!}f(X_1X_2)f(X_1^{-1}X_2^{-1})f(X_1^{-1}X_2)f^2(X_1)f(X_1^{-1})f(X_2)f^2(X_2^{-1})\\
&+\left(\frac{\kappa!\gamma!}{4(\kappa-4)!(\gamma-3)!}+\frac{\kappa!\gamma!}{2(\kappa-4)!(\gamma-4)!}\right)f(X_1X_2)f(X_1^{-1}X_2^{-1})f^2(X_1)f^2(X_1^{-1})f^2(X_2)f^2(X_2^{-1})\\
=&\frac{18\kappa!}{(\kappa-4)!}\bigg(f(X_1X_2)f(X_1^{-1}X_2^{-1})f^2(X_1)f^2(X_1^{-1})f^2(X_2)f^2(X_2^{-1})\\
&+\frac{2}{3(\kappa-3)}f(X_1X_2)f(X_1^{-1}X_2^{-1})f(X_1^{-1}X_2)f^2(X_1)f(X_1^{-1})f(X_2)f^2(X_2^{-1})\\
&+\frac{1}{3(\kappa-3)}f(X_1X_2)f(X_1^{-1}X_2^{-1})f(X_1X_2^{-1})f(X_1^{-1}X_2)f(X_1)f(X_1^{-1})f(X_2)f(X_2^{-1})\bigg).
\end{split}
\end{equation}
\end{exam}

\begin{rem} Observe that in \Cref{exam: cc_6}, the number of assignments such that all edges are distinct dominates among all the cases, especially when $\kappa$ is large enough; we refer to such dominant assignments as the \textbf{typical assignments}. Therefore, it is normally sufficiently accurate to optimize over the characteristic polynomial corresponding to the typical assignments only. Specifically, for $2$ concatenated cycles of length $2i$ and $2j$, suppose the number of edges in common is $2k$. Then, the characteristic polynomial can be well approximated by $\tilde{h}(\mathbf{X};G)=Cf^k(X_1X_2)f^k(X_1^{-1}X_2^{-1})f^{i-k}(X_1)f^{i-k}(X_1^{-1})f^{j-k}(X_2)f^{j-k}(X_2^{-1})$ for some constant $C\in\mathbb{N}$.
\end{rem}

\subsection{Gradient-Descent Distributor}
\label{subsec: objects_gradient descent distributor}

\begin{theo} Given the bipartite graph $G(V,C,E)$ of an object and the prototype class $\Se{P}(\Se{V},\Se{C})$. Suppose $\Se{E}(\Se{V},\Se{C})$ is the equivalence class induced by $\Se{V}$ and $\Se{C}$. Let $(\mathbf{v})_t$ be the $t$-th entry of the vector $\mathbf{v}$. Following the notation in \Cref{theo: Characteristic Polynomial of Prototypes}, the gradient of $\left[h(\mathbf{X};G|\Se{V},\Se{C})\right]_0$ with respect to $\mathbf{p}$ is given by:

\begin{equation}
\left(\nabla_{\mathbf{p}} \left[h(\mathbf{X};G|\Se{V},\Se{C})\right]_0\right)_{t}=\left[\sum_{\bar{e}\in\mathcal{E}} \prod_{s\in S}X_s^{a_{t}\sum_{e\in \bar{e}} \delta_{e,s}}\prod_{\bar{e}'\in\mathcal{E}\setminus \{\bar{e}\}}f\left(\prod_{e\in \bar{e}'}\prod_{s\in S} X_s^{\delta_{e,s}}\right)\right]_0.
\end{equation}

\end{theo}

\begin{proof} We obtain the gradient with respect to $\mathbf{p}$ as follows:
\begin{equation}
\begin{split}
&\left(\nabla_{\mathbf{p}} h(\mathbf{X};G|\mathcal{V},\mathcal{C})\right)_{t}\\
=&\sum_{\bar{e}\in\mathcal{E}} \left(\nabla_{\mathbf{p}} f\left(\prod_{e\in \bar{e}}\prod_{s\in S} X_s^{\delta_{e,s}}\right)\right)_{\hspace{-0.3em}t} \hspace{0.5em}\prod_{\bar{e}'\in\mathcal{E}\setminus \{\bar{e}\}}f\left(\prod_{e\in \bar{e}'}\prod_{s\in S} X_s^{\delta_{e,s}}\right)\\
=&\sum_{\bar{e}\in\mathcal{E}} \left(\nabla_{\mathbf{p}} f\left(\prod_{s\in S}\prod_{e\in \bar{e}} X_s^{\delta_{e,s}}\right)\right)_{\hspace{-0.3em}t} \hspace{0.5em}\prod_{\bar{e}'\in\mathcal{E}\setminus \{\bar{e}\}}f\left(\prod_{e\in \bar{e}'}\prod_{s\in S} X_s^{\delta_{e,s}}\right)\\
=&\sum_{\bar{e}\in\mathcal{E}} \left(\nabla_{\mathbf{p}} f\left(\prod_{s\in S}X_s^{\sum_{e\in \bar{e}} \delta_{e,s}}\right) \right)_{\hspace{-0.3em}t} \hspace{0.5em}\prod_{\bar{e}'\in\mathcal{E}\setminus \{\bar{e}\}}f\left(\prod_{e\in \bar{e}'}\prod_{s\in S} X_s^{\delta_{e,s}}\right)\\
=&\sum_{\bar{e}\in\mathcal{E}} \prod_{s\in S}X_s^{a_{t}\sum_{e\in \bar{e}} \delta_{e,s}}\prod_{\bar{e}'\in\mathcal{E}\setminus \{\bar{e}\}}f\left(\prod_{e\in \bar{e}'}\prod_{s\in S} X_s^{\delta_{e,s}}\right).\\
\end{split}
\end{equation}
Therefore, 
\begin{equation}
\begin{split}
&\left(\nabla_{\mathbf{p}} \left[h(\mathbf{X};G|\mathcal{V},\mathcal{C})\right]_0\right)_{t}=\left[\left(\nabla_{\mathbf{p}} h(\mathbf{X};G|\mathcal{V},\mathcal{C})\right)_{t}\right]_0\\
=&\left[\sum_{\bar{e}\in\mathcal{E}} \prod_{s\in S}X_s^{a_{t}\sum_{e\in \bar{e}} \delta_{e,s}}\prod_{\bar{e}'\in\mathcal{E}\setminus \{\bar{e}\}}f\left(\prod_{e\in \bar{e}'}\prod_{s\in S} X_s^{\delta_{e,s}}\right)\right]_0.\\
\end{split}
\end{equation}
\end{proof}

\begin{figure}
\centering
\includegraphics[width=0.25\textwidth]{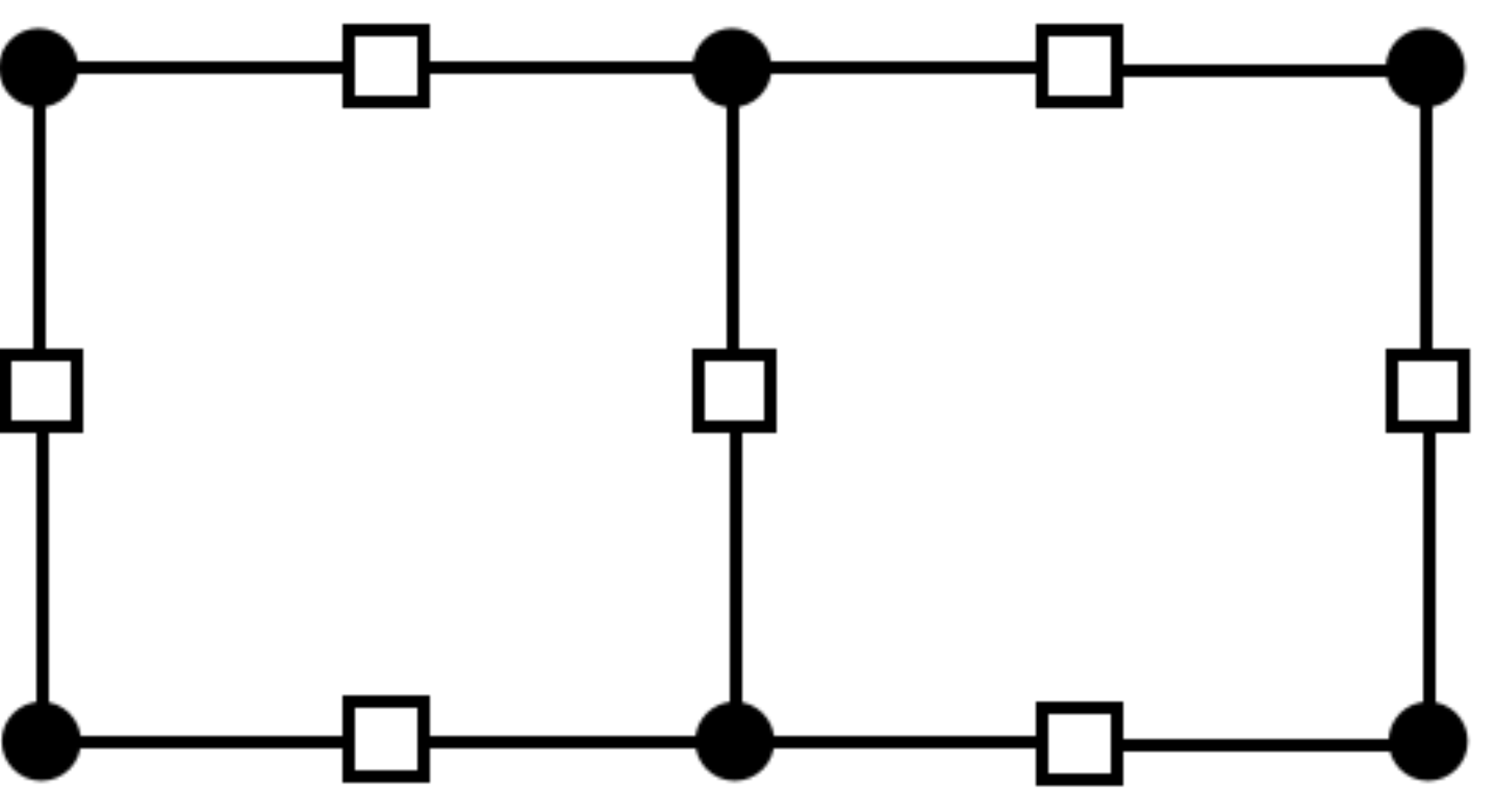}
  \caption{The targeted object consisting of two concatenated cycle-$8$ in \Cref{exam: GD_TC_6} and \Cref{exam: GD_TC_9}.}
  \label{fig: exam56} 
\end{figure}

Provided the explicit expression of $h(\mathbf{X};G)$ and its gradient, one can easily apply the gradient-descent algorithm to obtain an edge distribution that locally minimizes the expected number of active prototypes of the object specified by $G$. In \Cref{exam: GD_TC_6} and \Cref{exam: GD_TC_9}, we apply GRADE to two concatenated cycles-$8$ as shown in Fig.~\ref{fig: exam56}, under any prototype class $\Se{P}(\Se{V},\Se{C})$ such that $\Se{V}$ and $\Se{C}$ induce no equivalent edges in $E$. Denote by $P_{8-8}(\mathbf{a},\mathbf{p}|\Se{V},\Se{C})$ the probability that a prototype of this object becomes active after partitioning in an SC ensemble with coupling pattern $\mathbf{a}$ and edge distribution $\mathbf{p}$.

\begin{exam} \label{exam: GD_TC_6} Consider the following three cases of SC ensembles with $m=6$: 
\begin{enumerate}
  \item Full-memory codes with uniform edge distribution: $m_t=m=6$, $\mathbf{a}=(0,1,\dots,6)$ and $\mathbf{p}=\frac{1}{7} \mathbf{1}_{7}$. Then, $P_{8-8}(\mathbf{a},\mathbf{p}|\Se{V},\Se{C})=0.0049$;
  \item Full-memory codes with distribution obtained from GRADE: $m_t=m=6$, $\mathbf{a}=(0,1,\dots,6)$ and $\mathbf{p}=(0.2991,0.0899,0.0749,0.0733,0.0749,0.0896,0.2984)$. Then, $P_{8-8}(\mathbf{a},\mathbf{p}|\Se{V},\Se{C})=0.0032$;
  \item Non-full-memory codes with distribution obtained from GRADE: $m_t=3$, $\mathbf{a}=(0,1,4,6)$ and $\mathbf{p}=(0.2604,0.2063,0.2219,0.3114)$. Then, $P_{8-8}(\mathbf{a},\mathbf{p}|\Se{V},\Se{C})=0.0035$.
\end{enumerate}
\end{exam}

\begin{exam}\label{exam: GD_TC_9} Consider the following three cases of SC ensembles with $m=9$: 
\begin{enumerate}
  \item Full-memory codes with uniform edge distribution: $m_t=m=9$, $\mathbf{a}=(0,1,\dots,9)$ and $\mathbf{p}=\frac{1}{10} \mathbf{1}_{10}$. Then, $P_{8-8}(\mathbf{a},\mathbf{p}|\Se{V},\Se{C})=0.0024$;
  \item Full-memory codes with distribution obtained from GRADE: $m_t=m=9$, $\mathbf{a}=(0,1,\dots,9)$ and $\mathbf{p}=(0.2648,0.0803,0.0509,0.0526,0.0519,0.0519,0.0525,0.0508,0.0801,0.2644)$. Then, $P_{8-8}(\mathbf{a},\mathbf{p}|\Se{V},\Se{C})=0.0015$;
  \item Non-full-memory codes with distribution obtained from GRADE: $m_t=4$, $\mathbf{a}=(0,1,4,7,9)$ and $\mathbf{p}=(0.2479,0.1799,0.1262,0.1645,0.2814)$. Then, $P_{8-8}(\mathbf{a},\mathbf{p}|\Se{V},\Se{C})=0.0016$.
\end{enumerate}
\end{exam}

\begin{rem} In contrast to what we have shown regarding applying GRADE to single cycles, the gains obtained from applying GRADE to concatenated cycles are much more evident. As shown in \Cref{exam: GD_TC_6} and \Cref{exam: GD_TC_9}, for $m=6$ and $m=9$, the local minima obtained for full memory codes are quite close to the gains obtained for codes with coupling patterns $(0,1,4,6)$ and $(0,1,4,7,9)$, respectively (referred to as \textbf{topologically-coupled (TC) codes} later on). We show next in \Cref{section: construction} and \Cref{section: simulation} that TC codes have close performance to GD codes with full-memories, where both are obtained from applying GRADE-AO followed by CPO to concatenated cycles. 
\end{rem}

\begin{rem} Note that although we focus on non-tail-biting SC codes throughout this paper, this condition is by no means necessary. To extend our method to tail-biting codes, one just needs to change the cycle condition in (\ref{eqn: partition cycle}) from ``$\sum\nolimits_{k=1}^{g}\mathbf{P}(i_{k},j_{k})=\sum\nolimits_{k=1}^{g}\mathbf{P}(i_{k},j_{k+1})$'' to ``$\sum\nolimits_{k=1}^{g}\mathbf{P}(i_{k},j_{k})\equiv\sum\nolimits_{k=1}^{g}\mathbf{P}(i_{k},j_{k+1}) \mod L$''. If $L > m+1$, which is the typical case, the resultant optimal distribution is still very likely to be nonuniform because of the asymmetry among components indexed by $\{0,1,\dots,m\}$. This fact is important since while constructing non-tail-biting codes with large $\kappa$ and $m$, a large $L$ is typically desired due to the notable rate loss resulting from a small $L$. However, in certain practical applications, codes are typically of moderate length, which implies that a moderate $L$ is desirable. In such situations where $\kappa$ and $m$ are large, tail-biting codes do not suffer the same rate loss, and they can offer high error floor performance despite limiting the gain obtained from threshold saturation.
\end{rem}


\section{Algorithmic Optimization}
\label{section: construction}

We have developed the theory and the algorithm to obtain edge distributions that locally minimize the number of short cycles in \Cref{subsec: gradient descent} and generalized the results from cycles to arbitrary objects in \Cref{subsec: objects_gradient descent distributor}. In this section, we investigate algorithmic optimizers (AO) that search for excellent partitioning matrices under the guidance of GRADE. In particular, the edge distribution $\mathbf{p}_{opt}$ obtained through GRADE confines the search space to only contain matrices that have edge distributions near $\mathbf{p}_{opt}$.

We discuss both heuristic AOs based on semi-greedy algorithms and globally-optimal AOs based on variations of the OO technique proposed in \cite{esfahanizadeh2018finite,hareedy2020channel}. The heuristic AOs require low computational complexity and are applicable to arbitrary objects and any code parameters, but are only locally optimal. The OO-based AOs obtain the globally-optimal solutions, but currently only work on short cycles in SC codes with small pseudo-memories; we refer to these codes as \textbf{topologically-coupled (TC) codes}. The reason behind the nomenclature ``TC codes'' is the topological degrees of freedom they offer the code designer via the selection of the non-zero component matrices.

\subsection{Heuristic AO}
\label{subsec: heuristic AO}

In this subsection, we consider AOs that are based on heuristic methods. In this case, our proposed GRADE algorithm obtains an edge distribution to guide the AO. Starting from a random partitioning matrix $\mathbf{P}$ with the derived distribution, one can perform a semi-greedy algorithm that searches for partitioning matrix near the initial $\mathbf{P}$ that locally minimizes the number of targeted objects. Constraining the search space to contain $\mathbf{P}$'s that have distributions within small $L_1$ and $L_\infty$ distances from that of the original $\mathbf{P}$, and adopting the CPO next, significantly reduces the computational complexity to find a strong high-memory code. GRADE-guided heuristic AO has advantages in two aspects: 1) low complexity by reduced search space, and 2) higher probability of arriving at superior solution by providing a good enough initialization to AO that can avoid undesirable local minima.

\subsubsection{Cycle-Based Optimization}

\begin{algorithm}
\caption{Cycle-Based GRADE-A Optimizer (AO)} \label{algo: AO_cycle}
\begin{algorithmic}[1]
\Require 
\Statex $\gamma,\kappa,m,m_t,\mathbf{a}$: parameters of an SC ensemble;
\Statex $\mathbf{p}$: edge distribution obtained from \Cref{algo: GRADE};
\Statex $w$: weight of each cycle-$6$ assuming that of a cycle-$8$ is $1$;
\Statex $d_1$, $d_2$: parameters indicating the size of the search space;
\Ensure
\Statex $\mathbf{P}$: a locally optimal partitioning matrix;
\State Obtain the lists $\mathcal{L}_6(i,j)$, $\mathcal{L}_8(i,j)$ of cycles-$6$ candidates and cycle-$8$ candidates in the base matrix that contain node $(i,j)$, $1\leq i\leq \gamma$, $1\leq j\leq\kappa$;
\State Obtain $\mathbf{u}=\Tx{arg}\min\nolimits_{\mathbf{x}\in \{0,1,\dots,\gamma\kappa\}^{m+1}, ||\mathbf{x}||_1=\gamma\kappa}||\frac{1}{\gamma\kappa}\mathbf{x}-\mathbf{p}||_2$;
\For{$i\in\{0,1,\dots,m\}$} 
\State Place $\mathbf{u}\left[i+1\right]$ $i$'s into $\mathbf{P}$ randomly;
\EndFor
\State $\mathbf{d}\gets \mathbf{0}_{m+1}$;
\State $\Tx{noptimal}\gets \Tx{False}$;
\For{$i\in\{1,2,\dots,\gamma\}$, $j\in\{1,2,\dots,\kappa\}$}
\State $n_6\gets \Nm{\mathcal{L}_6(i,j)}$, $n_8\gets \Nm{\mathcal{L}_8(i,j)}$, $n\gets wn_6+n_8$;
\For{$v\in\{0,1,\dots,m\}$}
\State $\mathbf{d}'=\mathbf{d}$, $\mathbf{d}'\left[v+1\right]\gets \mathbf{d}'\left[v+1\right]+1$, $p\gets \mathbf{P}(i,j)$;
\If{$||\mathbf{d}'||_1\leq d_1$ and $||\mathbf{d}'||_{\infty}\leq d_2$}
\State $\mathbf{P}(i,j)\gets v$;
\State $t_6\gets \Nm{\mathcal{L}_6(i,j)}$, $t_8\gets \Nm{\mathcal{L}_8(i,j)}$, $t\gets wt_6+t_8$;
\If{$t<n$}
\State $\Tx{noptimal}\gets\Tx{True}$, $n\gets t$, $\mathbf{d}\gets\mathbf{d}'$, $\mathbf{P}(i,j)\gets v$;
\EndIf
\EndIf
\EndFor
\EndFor
\If{$\Tx{noptimal}$} 
\State \textbf{goto} step 6;
\EndIf
\State \textbf{return} $\mathbf{P}$;
\end{algorithmic}
\end{algorithm}

Based on the GRADE specified in \Cref{algo: GRADE}, we first present in \Cref{algo: AO_cycle} a corresponding AO that focuses on minimizing the weighted sum of the number of cycles-$6$ and cycles-$8$. We refer to codes obtained from GRADE-AO as \textbf{gradient-descent (GD) codes}. By replacing the initial distribution $\mathbf{p}$ with the uniform distribution, we obtain the so-called \textbf{uniform (UNF) codes}. In the special cases where the SC codes are not of full memory, i.e., the pseudo-memory is not identical with the memory, we refer to the codes obtained from GRADE-AO as \textbf{topologically-coupled (TC) codes}. We show in \Cref{section: simulation} by simulation that the distribution obtained by GRADE results in constructions that are better than those adopting uniform distribution and in existing literature. 

\subsubsection{Finer-Grained Optimization}

We next develop a finer-grained optimizer in \Cref{algo: AO_objects}, in which the targeted objects are two concatenated cycles where each of them is a cycle-$6$ or a cycle-$8$, as discussed in the examples of \Cref{subsec: objects_gradient descent distributor}. The critical part of the algorithm is enumerating all the objects of interest efficiently. Since we focus on \textit{concatenated} cycles of length $6$ or $8$, the key idea is to characterize concatenated cycles by the positions of the two degree $3$ VNs and the three paths connecting them. Here, we only consider objects that are elementary ASs. We call a path $\mathcal{P}=v_1$-$c_1$-$v_2$-$\cdots$-$c_{l}$-$v_{l+1}$ a type-$l$ path connecting $(c_1,v_1)$ and $(c_l,v_{l+1})$ and denote it by $L(\mathcal{P})=l$. Paths of type-$1$, type-$2$, and type-$3$ are shown in Fig.~\ref{fig: Paths}. Each concatenation of two cycles can be referred to as an $i$-$j$-$k$ object, where $i,j,k$ are the types of the three paths connecting the degree $3$ nodes in this object. Our targeted objects, where each of the two cycles is either a cycle-$6$ or a cycle-$8$, can only be $2$-$1$-$2$, $2$-$1$-$3$, $2$-$2$-$2$, or $3$-$1$-$3$ objects. Steps $1$-$5$ in \Cref{algo: AO_objects} are aimed at listing the paths of type $1$--$3$, and all the possible combinations of indices of the beginning and the ending CNs on each path.

\begin{figure}
\centering
\includegraphics[width=0.6\textwidth]{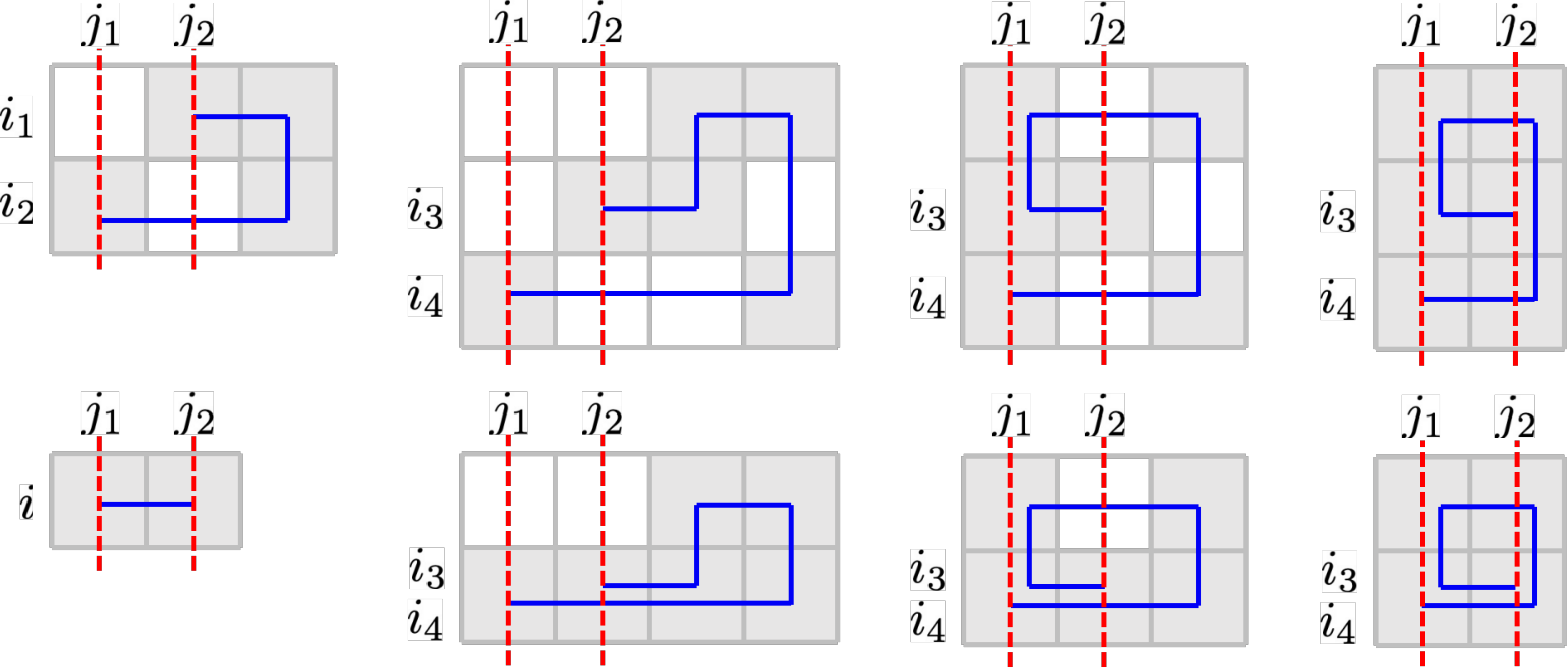}
  \caption{Type $1$-$3$ paths mentioned in \Cref{algo: AO_objects}: $\mathcal{L}_1(i,j_1,j_2)$, $\mathcal{L}_2(i_1,i_2,j_1,j_2)$, and $\mathcal{L}_3(i_3,i_4,j_1,j_2)$, $1\leq j_1<j_2 \leq\kappa$, $1\leq i,i_1,i_2,i_3,i_4\leq \gamma$, $i_1\neq i_2$.}
  \label{fig: Paths} 
\end{figure}

\begin{algorithm}
\caption{Fine-Grained GRADE-A Optimizer (AO)} \label{algo: AO_objects}
\begin{algorithmic}[1]
\Require 
\Statex $\gamma,\kappa,m,m_t,\mathbf{a}$: parameters of an SC code with full memory;
\Statex $\mathbf{p}$: edge distribution obtained from \Cref{algo: GRADE};
\Statex $\mathbf{w}=\left(w_1,w_2,w_3,w_4\right)$:  the weights of $2$-$1$-$2$, $2$-$1$-$3$, $2$-$2$-$2$, $3$-$1$-$3$ objects, respectively.
\Statex $d_1$, $d_2$: parameters indicating the size of the search space;
\Ensure
\Statex $\mathbf{P}$: a locally optimal partitioning matrix;
\State Obtain the lists $\mathcal{L}_1(i,j_1,j_2)$, $\mathcal{L}_2(i_1,i_2,j_1,j_2)$, and $\mathcal{L}_3(i_3,i_4,j_1,j_2)$, $1\leq j_1<j_2 \leq\kappa$, $1\leq i,i_1,i_2,i_3,i_4\leq \gamma$, $i_1\neq i_2$, where the lists are specified as follows:
\Statex a) $\mathcal{L}_1(i,j_1,j_2)$: all type-$1$ paths connecting $(i,j_1)$ and $(i,j_2)$ in the base matrix;
\Statex b) $\mathcal{L}_2(i_1,i_2,j_1,j_2)$: all type-$2$ paths connecting $(i_1,j_1)$ and $(i_2,j_2)$ in the base matrix;
\Statex c) $\mathcal{L}_3(i_3,i_4,j_1,j_2)$: all type-$3$ paths connecting $(i_3,j_1)$ and $(i_4,j_2)$ in the base matrix;
\State $\mathcal{I}_{212}\gets\{(i,i_1,i_2,i_3,i_4):1\leq i,i_1,i_2,i_3,i_4\leq \gamma, i_1<i_3, i_1\neq i_2, i_3\neq i_4\}$;
\State $\mathcal{I}_{213}\gets\{(i,i_1,i_2,i_3,i_4):1\leq i,i_1,i_2,i_3,i_4\leq \gamma, i_1\neq i_2\}$;
\State $\mathcal{I}_{222}\gets\{(i_1,i_2,i_3,i_4,i_5,i_6):1\leq i_1,i_2,i_3,i_4,i_5,i_6\leq \gamma, i_1<i_3<i_5, i_1\neq i_2, i_3\neq i_4, i_5\neq i_6\}$;
\State $\mathcal{I}_{313}\gets\{(i,i_1,i_2,i_3,i_4):1\leq i,i_1,i_2,i_3,i_4\leq \gamma, i_1<i_3\}$;
\State Obtain $\mathbf{u}=\Tx{arg}\min\nolimits_{\mathbf{x}\in \{0,1,2,\dots,\gamma\kappa\}^{m+1}, ||\mathbf{x}||_1=\gamma\kappa}||\frac{1}{\gamma\kappa}\mathbf{x}-\mathbf{p}||_2$;
\For{$i\in\{0,1,\dots,m\}$} 
\State Place $\mathbf{u}\left[i+1\right]$ $i$'s into $\mathbf{P}$ randomly;
\EndFor
\State $\mathbf{d}\gets \mathbf{0}_{m+1}$;
\State $n\gets M$; //$M$ is some very large constant
\State $\Tx{noptimal}\gets \Tx{False}$;
\For{$i\in\{1,2,\dots,\gamma\}$, $j\in\{1,2,\dots,\kappa\}$}
\For{$v\in\{0,1,\dots,m\}$}
\State $\mathbf{d}'=\mathbf{d}$, $\mathbf{d}'\left[v+1\right]\gets \mathbf{d}'\left[v+1\right]+1$, $p\gets \mathbf{P}(i,j)$;
\If{$||\mathbf{d}'||_1\leq d_1$ and $||\mathbf{d}'||_{\infty}\leq d_2$}
\State $\mathbf{P}(i,j)\gets v$;
\For{$1\leq j_1<j_2 \leq\kappa$, $1\leq i_0,i_1,i_2,i_3,i_4\leq \gamma$, $i_1\neq i_2$}
\State $v_{1,l}(i_0,j_1,j_2)\gets \lvert \{\mathcal{P}| L(\mathcal{P};\mathbf{P})=l, \mathcal{P}\in  \mathcal{L}_1(i_0,j_1,j_2)\}\rvert$, $-m\leq l\leq m$;
\State $v_{2,l}(i_1,i_2,j_1,j_2)\gets \lvert \{\mathcal{P}| L(\mathcal{P};\mathbf{P})=l, \mathcal{P}\in  \mathcal{L}_2(i_1,i_2,j_1,j_2)\}\rvert$, $-2m\leq l\leq 2m$;
\State $v_{3,l}(i_3,i_4,j_1,j_2)\gets \lvert \{\mathcal{P}| L(\mathcal{P};\mathbf{P})=l, \mathcal{P}\in  \mathcal{L}_3(i_3,i_4,j_1,j_2)\}\rvert$, $-m\leq l\leq m$;
\EndFor
\State $t_{212}\gets \sum\nolimits_{1\leq j_1<j_2\leq \kappa} \sum\nolimits_{(i_0,i_1,i_2,i_3,i_4)\in \mathcal{I}_{212}} \sum\nolimits_{-m\leq l\leq m} v_{1,l}(i_0,j_1,j_2)v_{2,l}(i_1,i_2,j_1,j_2)v_{2,l}(i_3,i_4,j_1,j_2)$; 
\State $t_{213}\gets \sum\nolimits_{1\leq j_1<j_2\leq \kappa} \sum\nolimits_{(i_0,i_1,i_2,i_3,i_4)\in \mathcal{I}_{213}} \sum\nolimits_{-m\leq l\leq m} v_{1,l}(i_0,j_1,j_2)v_{2,l}(i_1,i_2,j_1,j_2)v_{3,l}(i_3,i_4,j_1,j_2)$; 
\State $t_{222}\gets \sum\nolimits_{1\leq j_1<j_2\leq \kappa} \sum\nolimits_{(i_1,i_2,i_3,i_4,i_5,i_6)\in \mathcal{I}_{222}} \sum\nolimits_{-2m\leq l\leq 2m} v_{2,l}(i_1,i_2,j_1,j_2)v_{2,l}(i_3,i_4,j_1,j_2)v_{2,l}(i_5,i_6,j_1,j_2)$; 
\State $t_{313}\gets \sum\nolimits_{1\leq j_1<j_2\leq \kappa} \sum\nolimits_{(i_0,i_1,i_2,i_3,i_4)\in \mathcal{I}_{313}} \sum\nolimits_{-m\leq l\leq m} v_{1,l}(i_0,j_1,j_2)v_{3,l}(i_1,i_2,j_1,j_2)v_{3,l}(i_3,i_4,j_1,j_2)$; 
\State $t\gets w_1t_{212}+w_2t_{213}+w_3t_{222}+w_4t_{313}$;
\If{$t<n$}
\State $\Tx{noptimal}\gets\Tx{True}$, $n\gets t$, $\mathbf{d}\gets\mathbf{d}'$, $\mathbf{P}(i,j)\gets v$;
\EndIf
\EndIf
\EndFor
\EndFor
\If{$\Tx{noptimal}$} 
\State \textbf{goto} step 11;
\EndIf
\State \textbf{return} $\mathbf{P}$;
\end{algorithmic}
\end{algorithm}

\begin{rem} Note that in the finer-grained optimization, the condition of a concatenated-cycle pair in the protograph becoming a pair of concatenated cycles in the Tanner graph after lifting is that the two cycle candidates contained in this prototype all satisfy the cycle condition on lifting parameters specified in \Cref{lemma: cycle condition}. Therefore, after applying AO to minimize the number of concatenated-cycle pairs, instead of using the original CPO designed for cycle optimization in \cite{hareedy2020minimizing}, we adopt a modified version that is tailored for optimization over the number of pairs of concatenated cycles accordingly. 
\end{rem}

In a way similar to what we have done with approaches eliminating cycles, we define GD codes, UNF codes, and TC codes here. Moreover, as shown in \Cref{subsec: objects_gradient descent distributor}, the expected number of concatenated-cycle pairs in GD codes (with full memories) is close to that of GD-TC codes with carefully chosen pseudo-memories and coupling patterns.\footnote{We refer to the prototype of a pair of concatenated cycles as a concatenated-cycle pair.} In particular, memory $6$ GD ensembles can be approximated by GD-TC ensembles with pseudo-memory $3$ and coupling pattern $(0,1,4,6)$; memory $9$ GD ensembles can be approximated by GD-TC ensembles with pseudo-memory $4$ and coupling pattern $(0,1,4,7,9)$. This leads to the conjecture that the performance of GD-TC codes and the performance of GD codes are quite close, which is somewhat surprising provided that they differ a lot in edge distribution, and the impact of concatenated cycles on waterfall performance is not strictly characterized. In \Cref{subsec: simu_concat_cycle}, Monte-Carlo simulations support our conjecture, which enables more possibilities in TC codes. For example, TC codes can be globally-optimized given that their pseudo-memories are low; details will be discussed later on in \Cref{subsec: globally-optimal AO}.

\subsection{Globally-Optimal AO}
\label{subsec: globally-optimal AO}


In this subsection, we explore globally-optimal constructions of TC codes with small pseudo-memories. The motivation behind this task is to construct an SC code with memory $m$ under the same computational complexity needed to construct a full memory $m_t$ code, where $m_t<m$. Given $m_t$ and $m$, we first find the optimal $\mathbf{a}$, in terms of the minimum number of prototypes of interest, with length $m_t+1$ in a brute-force manner. Taking $m=4$ and $m_t=2$ as an example, the optimal coupling pattern with respect to the number of cycles is $\mathbf{a}=(0,1,4)$ and the corresponding optimal distribution is almost uniform. Moreover, we already know from \Cref{subsec: objects_gradient descent distributor} that regarding the optimal coupling pattern with respect to the number of concatenated-cycle pairs, $\mathbf{a}=(0,1,4,6)$ and $\mathbf{a}=(0,1,4,7,9)$ are not only the optimal coupling patterns for $(m,m_t)=(6,3)$ and $(m,m_t)=(9,4)$, respectively, but also approximate the optimal full memory GD ensembles quite closely in terms of performance.

Given the optimal coupling pattern $\mathbf{a}$, we then obtain an optimal partitioning matrix by the OO method proposed in \cite{esfahanizadeh2018finite} and \cite{hareedy2020channel}. We extend the OO method for memory $m_0$ SC codes to any TC code with pseudo-memory $m_t=m_0$, which does not increase the complexity of the approach. Note that despite the current OO works only on cycles, future steps can be taken towards the extension of OO into concatenated cycles, which has potential to lead to TC codes with excellent performance that is even better than the GD codes with full memories, provided that it is much harder to obtain globally-optimal solutions for GD codes with full memories.

Optimal TC codes with pseudo-memory $m_t$ have strictly fewer cycle candidates in their protographs than optimal SC codes with full memory $m=m_t$. Take $m=4$ and $m_t=2$ as an example. Suppose the optimal SC code has the partition $\bm{\Pi}=\mathbf{H}_0^{\Tx{P}}+\mathbf{H}_1^{\Tx{P}}+\mathbf{H}_2^{\Tx{P}}$. Consider the TC code with partition $\bm{\Pi}=\mathbf{H}_0^{\Tx{P}}+\mathbf{H}_1^{\Tx{P}}+\mathbf{H}_4^{\Tx{P}}$ such that $\mathbf{H}_2^{\Tx{P}}=\mathbf{H}_4^{\Tx{P}}$. Then, any cycle-$6$ candidate resulting from a cycle candidate in the base matrix assigned with $0$-$1$-$0$-$1$-$2$-$0$, $1$-$2$-$1$-$2$-$2$-$0$, or $0$-$1$-$2$-$1$-$x$-$x$, $x\in\{0,1,2\}$, in $\mathbf{P}$ no longer has a counterpart in the TC code, since by replacing $2$'s with $4$'s, assignments $0$-$1$-$0$-$1$-$4$-$0$, $1$-$4$-$1$-$4$-$4$-$0$, and $0$-$1$-$4$-$1$-$x$-$x$, $x\in\{0,1,4\}$, no longer satisfy the cycle condition in \Cref{lemma: cycle condition}. Moreover, there exists a bijection between the remaining candidates in the SC code and all candidates in the TC code through the replacement of $2$'s with $4$'s. 

Fig.~\ref{fig: TC_cycles} presents part of the protograph of a TC code with coupling pattern $(0,1,4)$ and that of its corresponding SC code with full memory $2$. The cycle-$6$ candidate colored by blue is assigned with $0$-$1$-$2$-$1$-$1$-$1$ in the SC code, which satisfies the cycle condition cycle candidates are generated in the protograph. However, the assignment becomes $0$-$1$-$4$-$1$-$1$-$1$ in the TC code, which no longer satisfies the cycle condition and no cycle candidates are generated in the protograph. The cycle-$6$ candidate colored by green corresponds to one that results in cycle-$6$ candidates in both the SC and the TC codes shown in the figure. We also marked out a cycle-$8$ candidate (colored by red) that only leads to cycle candidates in the SC code.

According to the aforementioned discussion, TC codes are better (have less cycles) than SC codes with the same circulant size and $m=m_t$. In \Cref{section: simulation}, we present simulation results of such codes and show that they can also outperform SC codes with the same constraint length (larger circulant size) and $m=m_t$.

\begin{figure}
\centering
\includegraphics[width=0.9\textwidth]{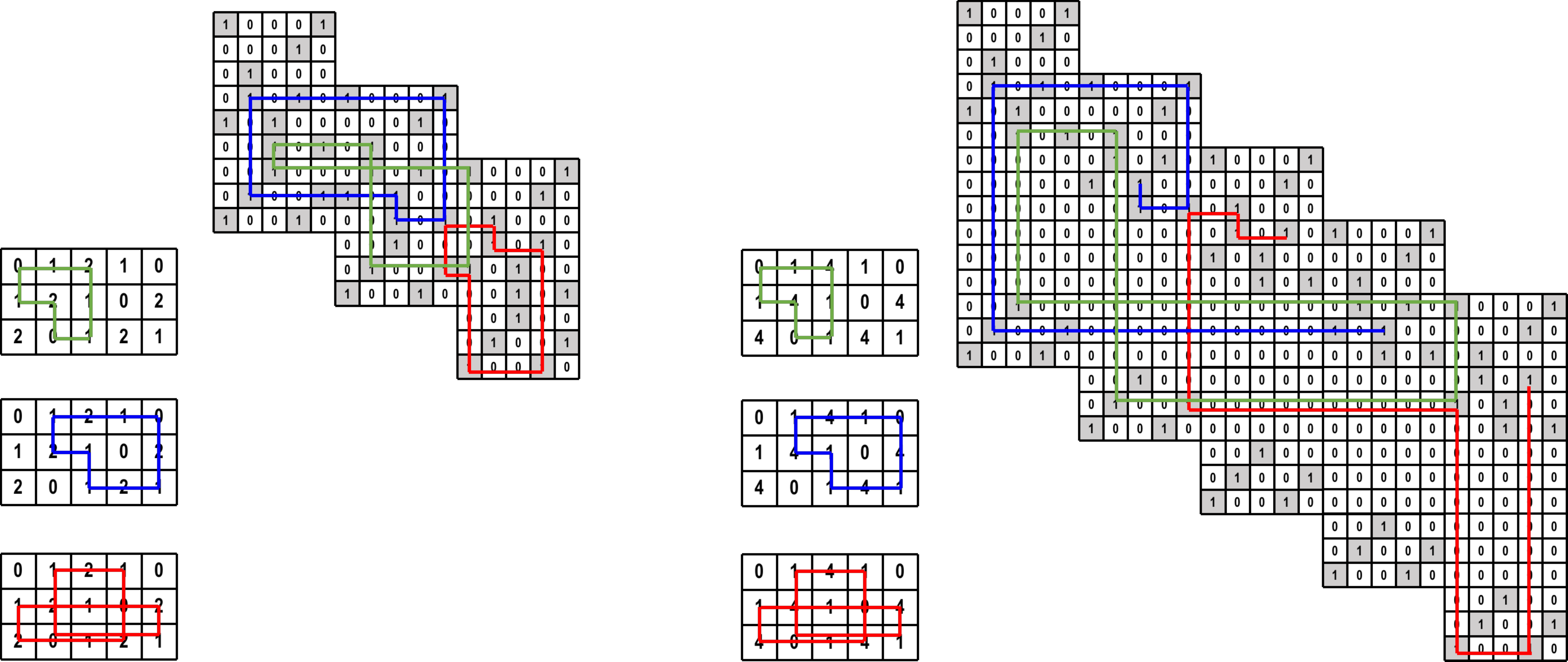}
  \caption{The first $5$ replicas of the protograph of a TC code with coupling pattern $(0,1,4)$ (the right panel), and the first $3$ replicas of the protograph of its corresponding SC code with memomry $2$ (the left panel). Three cycle candidates (colored by green, blue, and red, respectively) in the base matrix and their corresponding paths in the two protographs are marked out.}
  \label{fig: TC_cycles} 
\end{figure}

\section{Simulation Results}
\label{section: simulation}

In this section, we show the frame error rate/uncorrectable bit error rate (FER/UBER) curves of seven groups of SC codes designed by the GRADE-AO methods presented in \Cref{section: construction}. We demonstrate that codes constructed by the GRADE-AO methods offer significant performance gains compared with codes with uniform edge distributions and codes constructed through purely algorithmic methods.

\subsection{Optimization over Cycles}
\label{subsec: simu_cycle}

In this subsection, we simulate codes constructed based on optimizing the number of cycles using GRADE-AO specified in \Cref{section: framework} on the AWGN channel. Out of these three plots, Fig.~\ref{fig: FER_GD_UNF_3} and Fig.~\ref{fig: FER_GD_UNF_4_29} compare GD codes with UNF codes designed as in \Cref{subsec: heuristic AO}. Fig.~\ref{fig: FER_TC_SC_4_17} compares a TC code designed as in \Cref{subsec: globally-optimal AO} with optimal SC codes constructed through the OO-CPO method proposed in \cite{esfahanizadeh2018finite}. The GD/UNF codes have parameters $(\gamma,\kappa,m,z,L)=(3,7,5,13,100)$, $(3,17,9,7,100)$, and $(4,29,19,29,20)$, respectively. The TC code has parameters $(\gamma,\kappa,m_t,z,L)=(4,17,2,17,50)$ with the coupling pattern $\mathbf{a}=(0,1,4)$. For a fair comparison, we have selected two SC codes: one with a similar constraint length $(m+1)z$ and the other with an identical circulant power $z$. To ensure that the SC codes and the TC code have close rates and codelengths, the two SC codes have parameters $(\gamma,\kappa,m,z,L)=(4,17,2,28,30)$ and $(4,17,2,17,50)$, respectively. The statistics regarding the number of cycles of each code are presented in \Cref{table: cycle statistics}.

\begin{table}
\label{fig: partition_gamma4}
\centering
\caption{Statistics of the Number of Cycles}
\begin{tabular}{c|c|r|r}
\toprule
 $(\gamma,\kappa)$ & Code & Cycles-$6$ & Cycles-$8$ \\
\hline
\multirow{2}{*}{$(3,7)$} & GD & $0$ & $0$\\
 & UNF & $0$ & $6{,}292$\\
 \hline
\multirow{3}{*}{$(3,17)$} & GD & $0$ & $397{,}880$\\
 & UNF & $0$ & $559{,}902$ \\
 & Battaglioni \textit{et al.}\cite{battaglioni2017design}& $0$ & $451{,}337$\\
 \hline
\multirow{2}{*}{$(4,29)$} & GD & $0$ & $528{,}090$\\
 & UNF & $0$ & $1{,}087{,}268$\\
 \hline
\multirow{3}{*}{$(4,17)$} & TC & $15{,}436$ & -\\
 & SC (matched constraint length) & $19{,}180$ & -\\
 & SC (matched circulant size)& $74{,}579$ & -\\
\bottomrule
\end{tabular}
\label{table: cycle statistics}
\end{table}

Fig.~\ref{fig: FER_GD_UNF_3} shows FER curves of our GD/UNF comparisons with $(\gamma,\kappa)=(3,7)$ and $(3,17)$. The partitioning matrices and the lifting matrices of the codes are specified in Appendix \ref{append: AWGN_3_7} and Appendix \ref{append: AWGN_3_17}. When $\gamma=3$, cycles-$6$ are easily removed by the CPO. Therefore, we perform joint optimization on the number of cycles-$6$ and cycles-$8$ candidates by assigning different weights to cycle candidates in \Cref{algo: AO_cycle}. We observe a performance gain for the GD code with respect to the UNF code in both the waterfall region and the error floor region. Moreover, the number of cycles-$8$ in the $(3,17)$ GD code is reduced by $29\%$ and $12\%$ compared with the UNF code and the code constructed by Battaglioni \textit{et al.} in \cite{battaglioni2017design}, respectively. In addition, the $(3,17)$ GD code has no weight-$6$ absorbing sets (ASs) and $133$ weight-$7$ ASs, whereas the UNF code has $6$ weight-$6$ ASs and $361$ weight-$7$ ASs. As for the $(3,7)$ codes, all cycles-$6$ and cycles-$8$ are removed. Thus, the gain of the GD code compared with the UNF code exceeds the gain observed in the $(3,17)$ codes.

\begin{figure}
\centering
\includegraphics[width=0.45\textwidth]{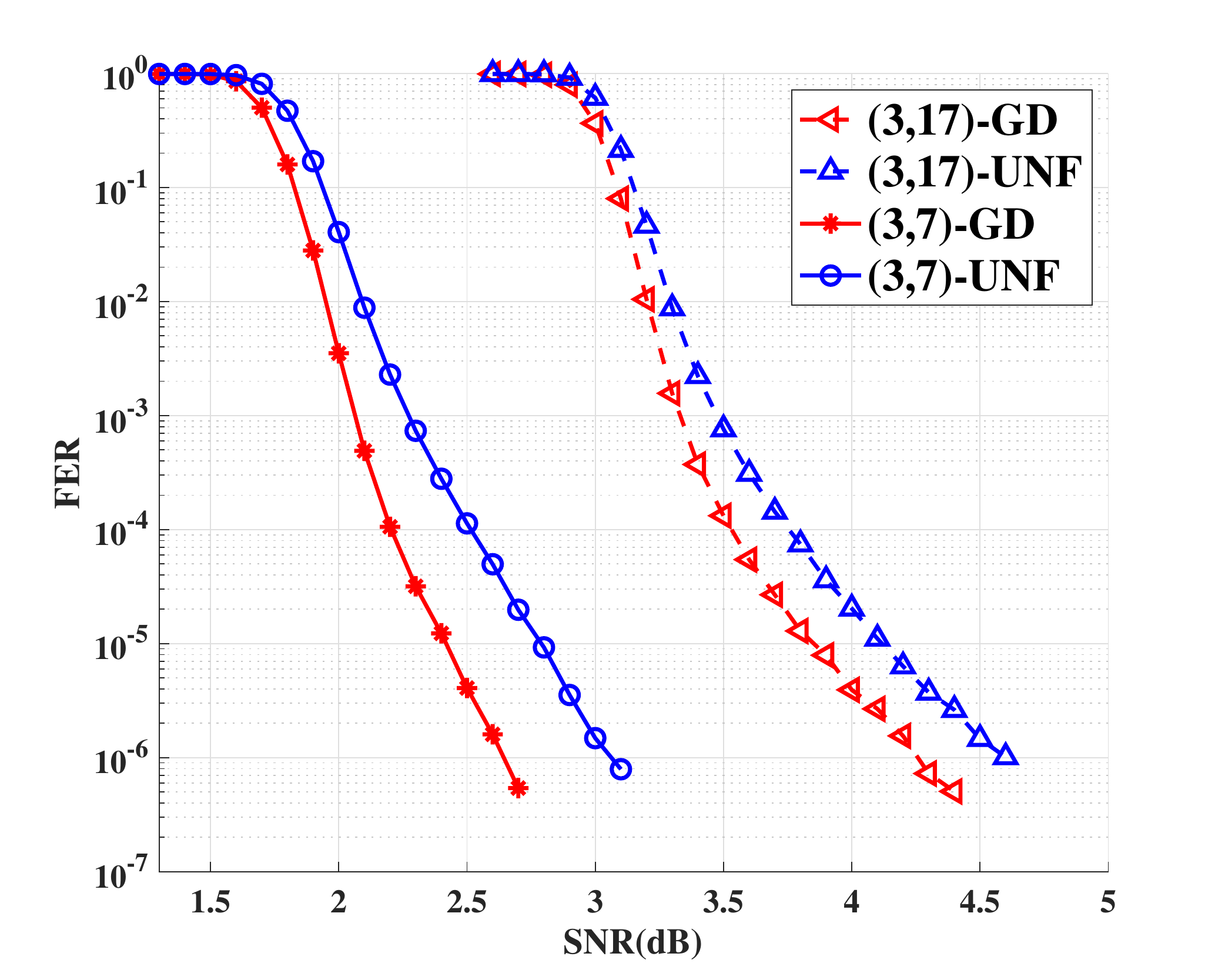}
\caption{FER curves of GD/UNF codes with $\gamma=3$ in the AWGN channel.}
\label{fig: FER_GD_UNF_3}
\end{figure}

Fig.~\ref{fig: FER_GD_UNF_4_29} shows FER curves of the GD/UNF comparison with $(\gamma,\kappa)=(4,29)$. The partitioning matrices and the lifting matrices of the codes are specified in Appendix \ref{append: AWGN_4_29}. Cycles-$6$ in the GD code and the UNF code are both removed, and the number of cycles-$8$ in the GD code demonstrates a $51.4\%$ reduction from the count observed in the UNF code. It is worth mentioning that both codes have no ASs of weights up to $8$, which is reflected in their FER curves via the sharp waterfall regions and the non-existing error floor regions. The FER of the GD/UNF codes decreases with a rate exceeding $12$ orders of magnitude per $0.5$ dB signal-to-noise ratio (SNR) increase. Moreover, the GD code has a significant gain of about $0.25$ dB over the UNF code. 

\begin{figure}
\centering
\includegraphics[width=0.45\textwidth]{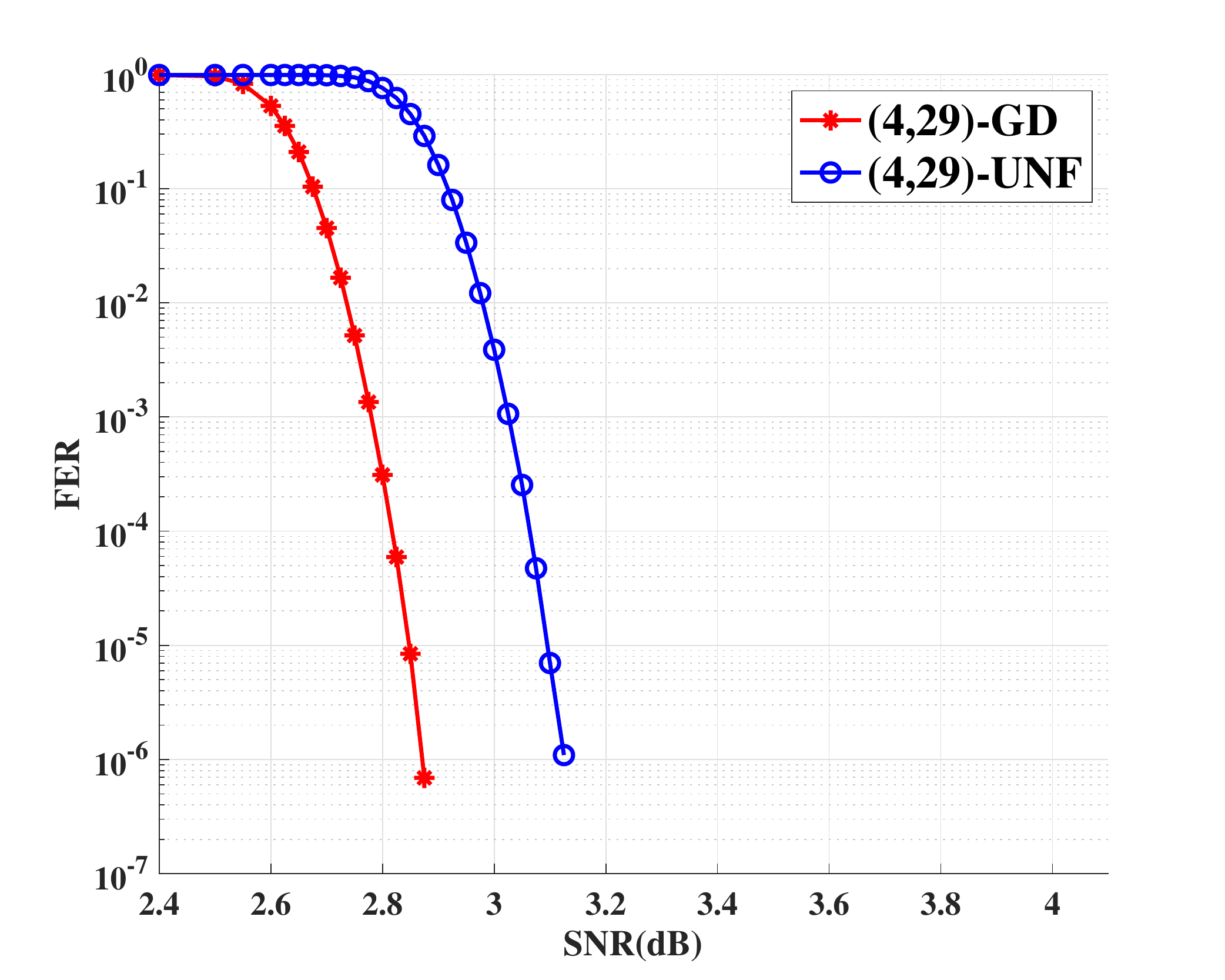}
\caption{FER curves of GD/UNF codes with $(\gamma,\kappa)=(4,29)$ in the AWGN channel.}
\label{fig: FER_GD_UNF_4_29}
\end{figure}

Fig.~\ref{fig: FER_TC_SC_4_17} shows the FER curves of the TC/SC codes with $(\gamma,\kappa)=(4,17)$. The partitioning matrices and the lifting matrices of the codes are specified in Appendix \ref{append: AWGN_4_17}. The number of cycles-$6$ in the $(4,17)$ TC code demonstrates a $79\%$ and a $20\%$ reduction from the counts observed in the SC codes with a matched constraint length and a matched circulant size, respectively. Moreover, the TC code has no weight-$6$ nor weight-$8$ ASs. It is shown that the TC code outperforms the optimal SC code with a matched constraint length, and that the gain is of greater magnitude when compared with the SC code of identical circulant size. 

\begin{rem} Note that although TC codes have higher memories and thus larger constraint lengths than SC codes of matched circulant sizes, they possess the same number of nonzero component matrices, and thus the same degrees of freedom in construction. This fact makes TC codes even more promising if we can devise for them windowed decoding algorithms with window sizes that are comparable to the corresponding SC codes of matched circulant sizes.
\end{rem}

\begin{figure}
\centering
\includegraphics[width=0.45\textwidth]{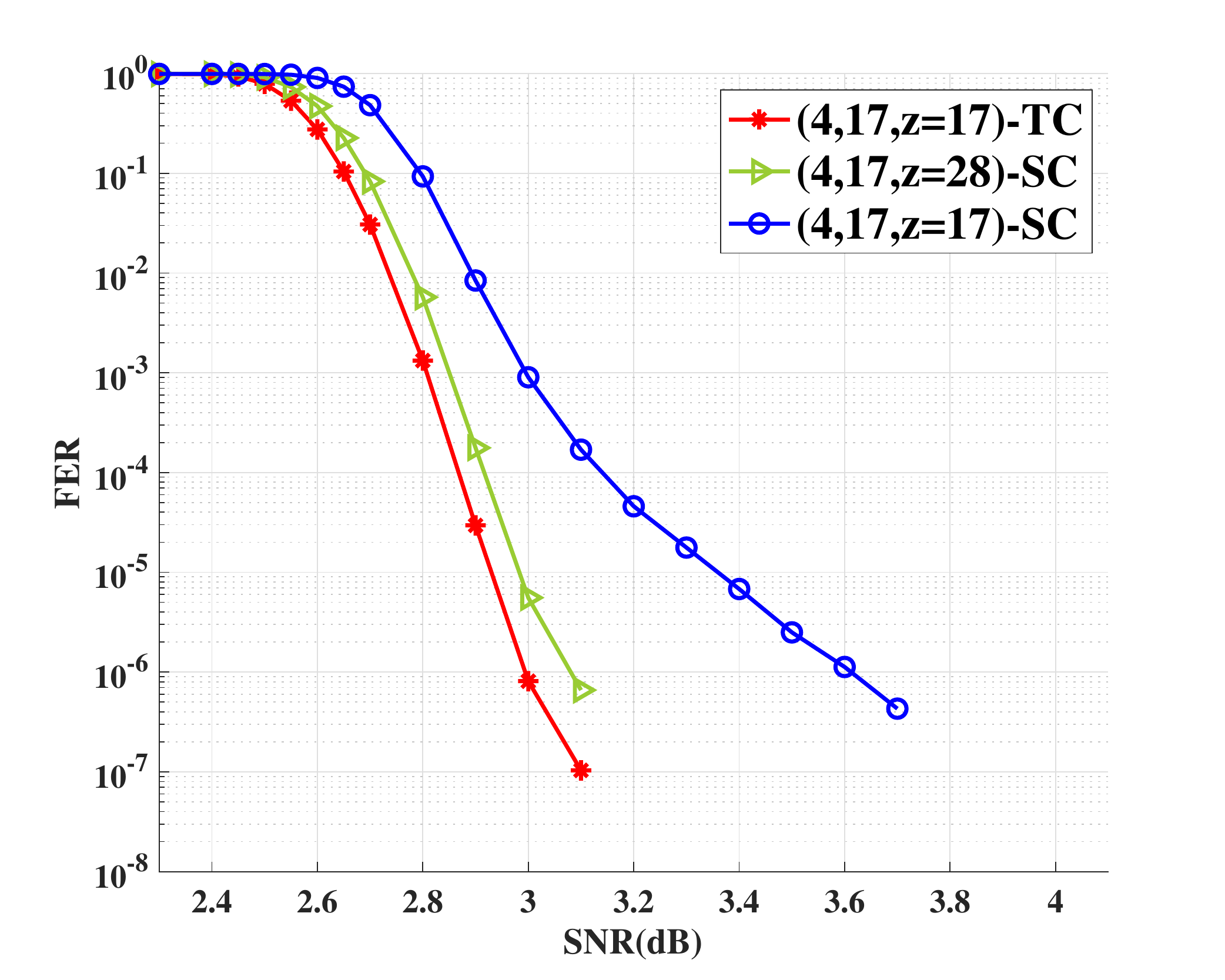}
\caption{FER curves of TC/SC codes with $(\gamma,\kappa)=(4,17)$ in the AWGN channel.}
\label{fig: FER_TC_SC_4_17}
\end{figure}

\subsection{Optimization over Concatenated Cycles}
\label{subsec: simu_concat_cycle}

In this subsection, we simulate codes constructed based on minimizing the number of concatenated-cycle pairs using the GRADE-AO specified in \Cref{section: generalization of GRADE} and \Cref{subsec: heuristic AO}. We compare GD/UNF codes with $m=6$ in addition to TC codes with $m_t=3$ and $\mathbf{a}=(0,1,4,6)$ on the binary symmetric channel (BSC), Flash channel, and magnetic recording channel. There are two groups of GD/TC/UNF codes that have parameters $(\gamma,\kappa,m,z,L)=(4,24,6,17,40)$ and $(4,20,6,13,20)$, respectively. The statistics regarding the number of cycles of each code are presented in \Cref{table: concat cycle statistics}. 

\begin{table}
\label{fig: partition_gamma4}
\centering
\caption{Statistics of the Number of Concatenated Cycles}
\begin{tabular}{c|c|r|r|r|r|r}
\toprule
 $(\gamma,\kappa)$ & Code & Cycles-$6$ & $2$-$1$-$2$ Objects & $2$-$2$-$2$ Objects & $2$-$1$-$3$ Objects & $3$-$1$-$3$ Objects \\
\hline
\multirow{3}{*}{$(4,24)$ (NLM)} & GD & $4{,}794$ & $1{,}751$ & $807{,}534$ & $1{,}162{,}035$ &$125{,}869{,}717$ \\
 & TC & $4{,}352$ & $4{,}301$ & $816{,}816$  &$1{,}125{,}400$  &$126{,}903{,}436$\\
 & UNF & $11{,}713$ & $7{,}293$ & $1{,}308{,}762$ &$2{,}615{,}297$ & $208{,}425{,}933$\\
 \hline
\multirow{3}{*}{$(4,24)$ (BSC)} & GD & $4{,}794$ & $1{,}802$ & $822{,}120$  &$1{,}212{,}100$ &$126{,}514{,}918$\\
 & TC & $4{,}403$ & $4{,}097$ & $807{,}534$ &$1{,}139{,}000$ &$126{,}434{,}525$\\ 
 & UNF & $11{,}713$ & $7{,}293$ & $1{,}308{,}762$ &$2{,}615{,}297$ & $208{,}425{,}933$\\
 \hline
\multirow{3}{*}{$(4,20)$} & GD & $2{,}171$ & $338$ & $178{,}334$ &$238{,}160$    &$19{,}638{,}372$\\
 & TC & $2{,}665$ & $1{,}404$ & $178{,}828$ &$262{,}509$ &$20{,}571{,}499$ \\
 & UNF & $2{,}444$ & $2{,}860$ & $287{,}014$ &$540{,}865$ &$34{,}362{,}393$\\
\bottomrule
\end{tabular}
\label{table: concat cycle statistics}
\end{table}

\begin{figure}
\centering
\includegraphics[width=0.45\textwidth]{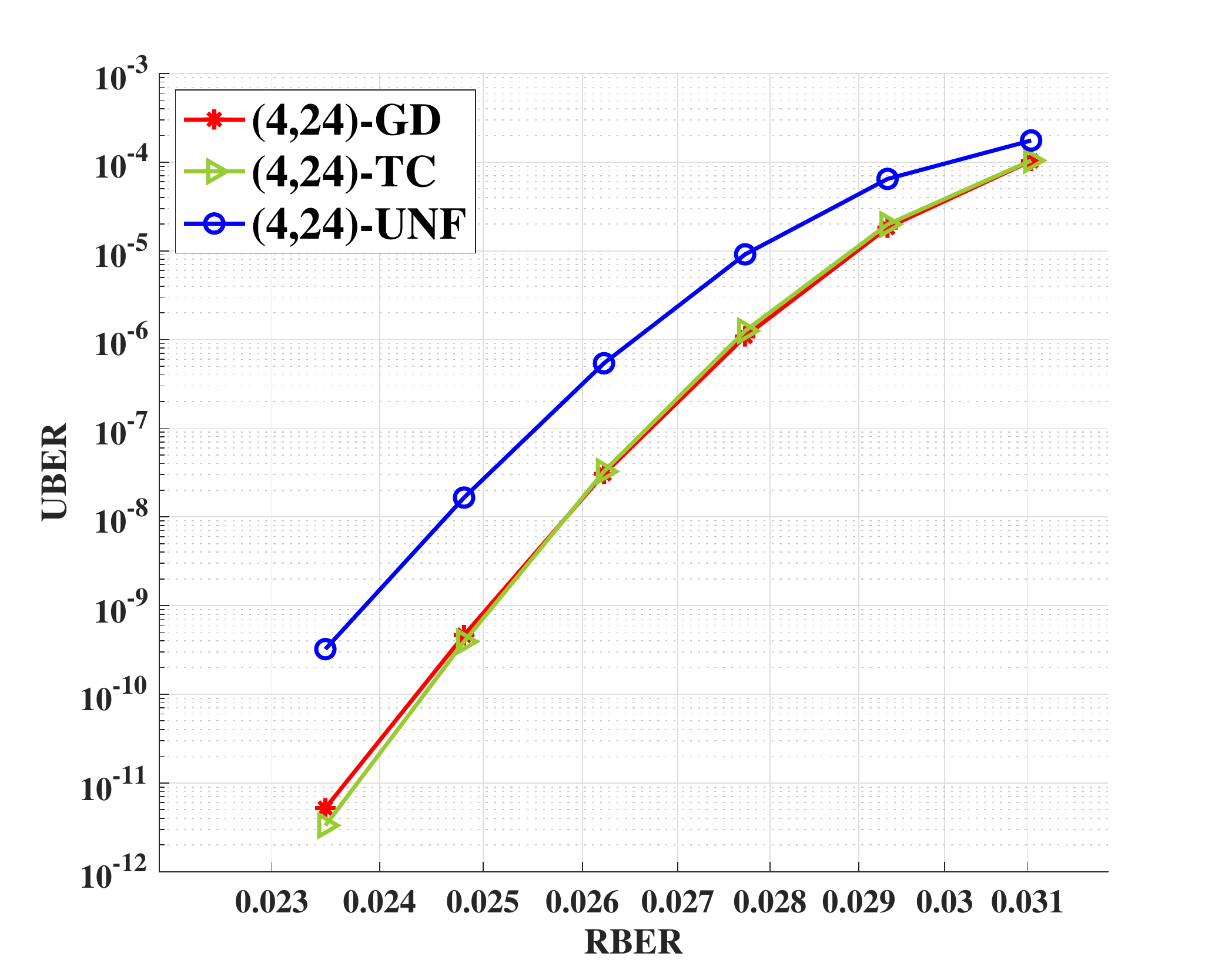}
\caption{UBER curves of GD/TC/UNF codes with $(\gamma,\kappa)=(4,24)$ in the NLM channel.}
\label{fig: FER_NLM_4_24}
\end{figure}

Fig.~\ref{fig: FER_NLM_4_24} shows UBER curves of the GD/TC/UNF codes with $(\gamma,\kappa)=(4,24)$ on a Flash channel.\footnote{Note that although we only provide simulation results on some typical channels for brevity in this paper, our approach is generally applicable to many other channels, such as the ones underlying three-dimensional cross point (3D XPoint) \cite{3DXpoint} and two-dimensional magnetic recording (TDMR) \cite{TDMR} systems.} The Flash channel used in this section is a practical, asymmetric Flash channel, which is the normal-Laplace mixture (NLM) Flash channel \cite{parnell2014modelling}. In the NLM channel, the threshold voltage distribution of sub-$20$nm multi-level cell (MLC) Flash memories is carefully modeled. The four levels are modeled as different NLM distributions, incorporating several sources of error due to wear-out effects, e.g., programming/erasing problems, thereby resulting in significant asymmetry. Furthermore, the authors of \cite{parnell2014modelling} provided accurate fitting results of their model for program/erase (P/E) cycles up to $10$ times the manufacturer's endurance specification (up to $30{,}000$ P/E cycles). We implemented the NLM channel based on the parameters described in \cite{parnell2014modelling}. Here, we use $3$ threshold voltage reads, and the sector size is $512$ bytes. For decoding, we use a finite-precision (FP) fast Fourier transform based $q$-ary sum-product algorithm (FFT-QSPA) LDPC decoder \cite{declercq2007decoding}. The decoder performs a maximum of $50$ iterations, and it stops if a codeword is reached sooner.

The partitioning matrices and the lifting matrices of the codes are specified in Appendix \ref{append: NLM_4_24}. The non-binary edge weights are set as in \cite{divsalar2014non} and \cite{bazarsky2013design}. The codes can be further optimized by applying the more advanced WCM framework presented in \cite{hareedy2016general} and \cite{hareedy2019combinatorial}. The first row of \Cref{table: concat cycle statistics} shows the statistics of unlabeled cycles and concatenated cycles in each code. The number of objects in GD/TC codes are reduced by around $40\%$ compared with the UNF code. No error floors are observed in any one of the UBER curves. The UBER of the GD/TC codes decreases with a rate exceeding $14$ orders of magnitude per $0.01$ RBER decrease. Moreover, the GD/TC codes have a significant gain of about $2$ orders of magnitude over the UNF code at RBER $0.0235$. It worths mentioning that the UBER curves of the GD/TC codes nearly overlap, which is in accordance with the closeness of the statistics of objects in them.

While NB codes are adopted in the simulations over the NLM channel, independent coding are more widely applied in practical Flash solutions in order to preserve high access speed. Therefore, we next present in Fig.~\ref{fig: FER_BSC_4_24} the UBER curves of the GD/TC/UNF codes with $(\gamma,\kappa)=(4,24)$ on the BSC, as a simplified model of single-level cell (SLC) channel with $1$ threshold voltage read. 

\begin{figure}
\centering
\includegraphics[width=0.45\textwidth]{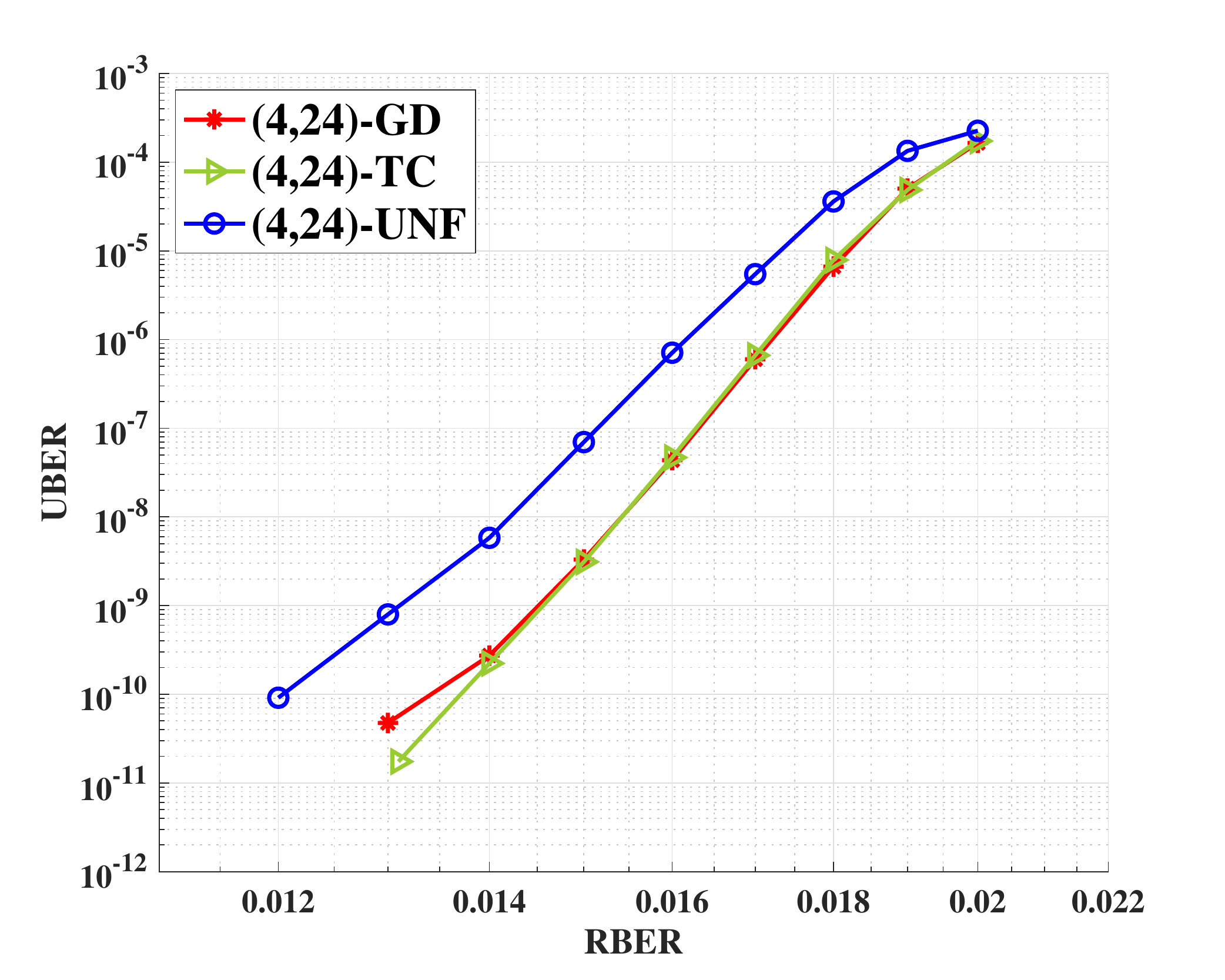}
\caption{UBER curves of GD/TC/UNF codes with $(\gamma,\kappa)=(4,24)$ in the BSC.}
\label{fig: FER_BSC_4_24}
\end{figure}

The partitioning matrices and the lifting matrices of the codes are specified in Appendix \ref{append: BSC_4_24}. While partitioning matrices of the $(4,24)$ NB codes constructed for the NLM channels are adopted as they are here, we have modified the lifting parameters slightly in order to remove all unlabeled ASs with weights less than or equal to $7$ in the GD/TC codes: this is achieved by changing one entry in each lifting matrix. According to \Cref{table: concat cycle statistics}, the number of objects in the GD/TC codes has not changed dramatically, and they still demonstrate a $40\%$ reduction compared with the count observed in the UNF code. As shown in Fig.~\ref{fig: FER_BSC_4_24}, the UBER curves of GD/TC codes are still close in the early waterfall region like they are in the NLM channel simulations; however, they start to deviate at RBER less than $0.014$. The TC code has no observed error floor in its performance curve as expected, and it has a $2$ orders of magnitude gain over the UNF code at RBER $0.013$. The GD code curve surprisingly floors despite that there are no ASs with weight less than $8$ in it: error profile analysis shows that the error floor at RBER $0.013$ results from only $2$ different large weight errors (one of weight $78$ and another of weight $168$) instead of structured small weight errors. 

\begin{rem} While the reason why the large weight errors observed in the error profile of the GD code are detrimental in the BSC simulations remains unexplored and is left for future investigation, the TC code is observed to be robust against these errors, which is specifically intriguing. Moreover, the fact that the waterfall performance of GD/TC codes is remarkably superior to that of UNF codes calls for an asymptotic analysis that takes edge distribution into consideration. The significant gain achieved by TC codes substantiates the potential of TC codes in Flash memories. 
\end{rem}

\begin{figure}
\centering
\includegraphics[width=0.45\textwidth]{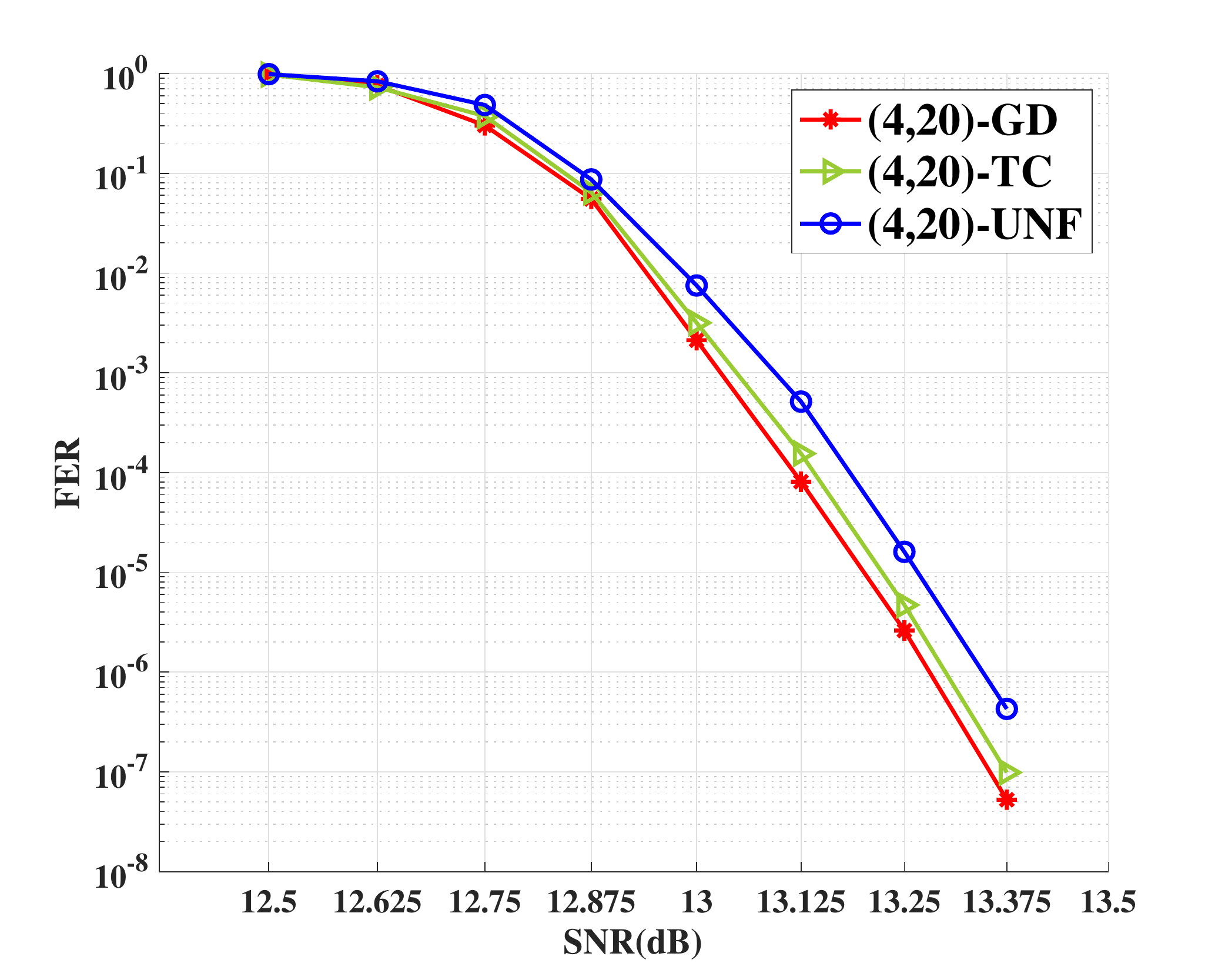}
\caption{FER curves of GD/TC/UNF codes with $(\gamma,\kappa)=(4,20)$ in the MR channel.}
\label{fig: FER_MR_4_20}
\end{figure}

Fig.~\ref{fig: FER_MR_4_20} shows FER curves of the GD/TC/UNF codes with $(\gamma,\kappa)=(4,20)$ on the MR channel. The MR system adopts the partial-response (PR) channel presented in \cite{hareedy2016nonbinary} and sequence detection. This PR channel incorporates the MR channel effects: inter-symbol interference (intrinsic memory), jitter, and electronic noise. The normalized channel density \cite{srinivasa2014communication,hareedy2016nonbinary} is $1.4$, and the PR equalization target is $\left(8,14,2\right)$. 
The filtering units are followed by a Bahl-Cocke-Jelinek-Raviv (BCJR) detector \cite{bcjr1974}, which is based on pattern-dependent noise prediction (PDNP) \cite{jaekyun2001pattern}, and again an FP FFT-QSPA ($q = 2$) LDPC decoder \cite{fossorier2004quasicyclic}. The number of global (detect-decoder) iterations is $10$, and the number of local (decoder only) iterations is $20$. Unless a codeword is reached, the decoder performs its prescribed number of local iterations for each global iteration. More details can be found in \cite{hareedy2016nonbinary}.

The partitioning matrices and the lifting matrices of the codes are specified in Appendix \ref{append: MR_4_20}. The number of targeted objects (concatenated-cycle pairs) observed in the GD code demonstrates an approximate $40\%$ reduction from the count observed in the UNF code. The FER of the GD/TC codes decreases with a rate that is approximately $13$ orders of magnitude per $1$ dB SNR increase. Moreover, the GD code has a significant gain of about $1$ dB over the UNF code at SNR $13.375$ dB. These results substantiate the remarkable impact of the GRADE-AO method in constructing SC codes with superior performance for storage devices, with potential usage in further applications including wireless communication systems.

\section{Conclusion}
\label{section: conclusion}
Discrete optimization of the constructions of spatially-coupled (SC) codes with high memories is known to be computationally expensive. Heuristic algorithms are efficient, but can hardly guarantee the performance because of the lack of theoretical guidance. In this paper, we proposed the so-called GRADE-AO method, a probabilistic framework that efficiently searches for locally optimal QC-SC codes with arbitrary memories. We obtained a locally optimal edge distribution that minimizes the expected number of the most detrimental objects via gradient descent. Starting from a random partitioning matrix with the derived edge distribution, we then applied a semi-greedy algorithm to find a locally optimal partitioning matrix near it. While the application of GRADE-AO in optimizing the number of short cycles has shown noticeable gains, we focused in this paper on minimizing the number of more detrimental objects, the concatenated cycles. This finer-grained optimization avoids unnecessary attention on individual cycles which are typically not problematic on their own, especially in codes with high VN degrees and irregular codes. Simulation results show that our proposed constructions have a significant performance gain over state-of-the-art codes; this gain is shown to be universal in both waterfall and error floor regions, as well as on channels underlying various practical systems. Future work includes extending the framework to other classes of underlying block codes.

\section*{Acknowledgment} 
The authors would like to thank Shyam Venkatasubramaian for his assistance in carrying out part of the simulations in this research, and would also like to thank Christopher Cannella and Arnab Kar for useful discussions regarding probabilistic optimization on SC code constructions while being at Duke University. 
\bibliography{ref}

\begin{thebibliography}{10}
\providecommand{\url}[1]{#1}
\csname url@samestyle\endcsname
\providecommand{\newblock}{\relax}
\providecommand{\bibinfo}[2]{#2}
\providecommand{\BIBentrySTDinterwordspacing}{\spaceskip=0pt\relax}
\providecommand{\BIBentryALTinterwordstretchfactor}{4}
\providecommand{\BIBentryALTinterwordspacing}{\spaceskip=\fontdimen2\font plus
\BIBentryALTinterwordstretchfactor\fontdimen3\font minus
  \fontdimen4\font\relax}
\providecommand{\BIBforeignlanguage}[2]{{%
\expandafter\ifx\csname l@#1\endcsname\relax
\typeout{** WARNING: IEEEtran.bst: No hyphenation pattern has been}%
\typeout{** loaded for the language `#1'. Using the pattern for}%
\typeout{** the default language instead.}%
\else
\language=\csname l@#1\endcsname
\fi
#2}}
\providecommand{\BIBdecl}{\relax}
\BIBdecl

\bibitem{Yang2020GRADE}
S.~Yang, A.~Hareedy, S.~Venkatasubramanian, R.~Calderbank, and L.~Dolecek,
  ``{G}{R}{A}{D}{E}-{A}{O}: Towards near-optimal spatially-coupled codes with
  high memories,'' in \emph{2021 IEEE International Symposium on Information
  Theory (ISIT)}, Jul. 2021, pp. 587--592.

\bibitem{5695130}
S.~{Kudekar}, T.~J. {Richardson}, and R.~L. {Urbanke}, ``Threshold saturation
  via spatial coupling: Why convolutional {L}{D}{P}{C} ensembles perform so
  well over the {B}{E}{C},'' \emph{IEEE {T}rans. {I}nformation {T}heory},
  vol.~57, no.~2, pp. 803--834, Feb. 2011.

\bibitem{kumar2014threshold}
S.~{Kumar}, A.~J. {Young}, N.~{Macris}, and H.~D. {Pfister}, ``Threshold
  saturation for spatially coupled {L}{D}{P}{C} and {L}{D}{G}{M} codes on
  {B}{M}{S} channels,'' \emph{IEEE {T}rans. {I}nformation {T}heory}, vol.~60,
  no.~12, pp. 7389--7415, Dec. 2014.

\bibitem{olmos2015scaling}
P.~M. Olmos and R.~L. Urbanke, ``A scaling law to predict the finite-length
  performance of spatially-coupled {L}{D}{P}{C} codes,'' \emph{IEEE {T}rans.
  {I}nformation {T}heory}, vol.~61, no.~6, pp. 3164--3184, 2015.

\bibitem{hareedy2017high}
A.~{Hareedy}, H.~{Esfahanizadeh}, and L.~{Dolecek}, ``High performance
  non-binary spatially-coupled codes for flash memories,'' in \emph{2017 IEEE
  Information Theory Workshop (ITW)}, Nov. 2017, pp. 229--233.

\bibitem{lentmaier2010iterative}
M.~{Lentmaier}, A.~{Sridharan}, D.~J. {Costello}, and K.~S. {Zigangirov},
  ``Iterative decoding threshold analysis for {L}{D}{P}{C} convolutional
  codes,'' \emph{IEEE {T}rans. {I}nformation {T}heory}, vol.~56, no.~10, pp.
  5274--5289, Oct. 2010.

\bibitem{Iyengar2013windowed}
A.~R. {Iyengar}, P.~H. {Siegel}, R.~L. {Urbanke}, and J.~K. {Wolf}, ``Windowed
  decoding of spatially coupled codes,'' \emph{IEEE {T}rans. {I}nformation
  {T}heory}, vol.~59, no.~4, pp. 2277--2292, Apr. 2013.

\bibitem{mitchell2015spatially}
D.~G.~M. {Mitchell}, M.~{Lentmaier}, and D.~J. {Costello}, ``Spatially coupled
  {L}{D}{P}{C} codes constructed from protographs,'' \emph{IEEE {T}rans.
  {I}nformation {T}heory}, vol.~61, no.~9, pp. 4866--4889, Sep. 2015.

\bibitem{esfahanizadeh2018finite}
H.~Esfahanizadeh, A.~Hareedy, and L.~Dolecek, ``Finite-length construction of
  high performance spatially-coupled codes via optimized partitioning and
  lifting,'' \emph{IEEE {T}rans. {C}ommunications}, vol.~67, no.~1, pp. 3--16,
  Jan. 2018.

\bibitem{hareedy2020channel}
A.~Hareedy, R.~Wu, and L.~Dolecek, ``A channel-aware combinatorial approach to
  design high performance spatially-coupled codes,'' \emph{IEEE {T}rans.
  {I}nformation {T}heory}, vol.~66, no.~8, pp. 4834--4852, Aug. 2020.

\bibitem{pusane2011deriving}
A.~E. {Pusane}, R.~{Smarandache}, P.~O. {Vontobel}, and D.~J. {Costello},
  ``Deriving good {L}{D}{P}{C} convolutional codes from {L}{D}{P}{C} block
  codes,'' \emph{IEEE {T}rans. {I}nformation {T}heory}, vol.~57, no.~2, pp.
  835--857, Feb. 2011.

\bibitem{dolecek2010analysis}
L.~{Dolecek}, Z.~{Zhang}, V.~{Anantharam}, M.~J. {Wainwright}, and
  B.~{Nikolic}, ``Analysis of absorbing sets and fully absorbing sets of
  array-based {L}{D}{P}{C} codes,'' \emph{IEEE {T}rans. {I}nformation
  {T}heory}, vol.~56, no.~1, pp. 181--201, Jan. 2010.

\bibitem{naseri2020spatially}
S.~{Naseri} and A.~H. {Banihashemi}, ``Spatially coupled {L}{D}{P}{C} codes
  with small constraint length and low error floor,'' \emph{IEEE Communications
  Letters}, vol.~24, no.~2, pp. 254--258, Feb. 2020.

\bibitem{naseri2021construction}
S.~Naseri and A.~H. Banihashemi, ``Construction of time invariant spatially
  coupled ldpc codes free of small trapping sets,'' \emph{IEEE Transactions on
  Communications}, vol.~69, no.~6, pp. 3485--3501, Jun. 2021.

\bibitem{battaglioni2017design}
M.~{Battaglioni}, A.~{Tasdighi}, G.~{Cancellieri}, F.~{Chiaraluce}, and
  M.~{Baldi}, ``Design and analysis of time-invariant {S}{C}-{L}{D}{P}{C}
  convolutional codes with small constraint length,'' \emph{IEEE {T}rans.
  {C}ommunications}, vol.~66, no.~3, pp. 918--931, Mar. 2018.

\bibitem{9112247}
S.~Mo, L.~Chen, D.~J. Costello, D.~G.~M. Mitchell, R.~Smarandache, and J.~Qiu,
  ``Designing protograph-based quasi-cyclic spatially coupled {L}{D}{P}{C}
  codes with large girth,'' \emph{IEEE Transactions on Communications},
  vol.~68, no.~9, pp. 5326--5337, 2020.

\bibitem{beemer2017generalized}
A.~Beemer, S.~Habib, C.~A. Kelley, and J.~Kliewer, ``A generalized algebraic
  approach to optimizing {S}{C}-{L}{D}{P}{C} codes,'' in \emph{2017 55th Annual
  Allerton Conference on Communication, Control, and Computing (Allerton)},
  2017, pp. 672--679.

\bibitem{hareedy2016general}
A.~Hareedy, C.~Lanka, and L.~Dolecek, ``A general non-binary {L}{D}{P}{C} code
  optimization framework suitable for dense flash memory and magnetic
  storage,'' \emph{IEEE Journal on Selected Areas in Communications}, vol.~34,
  no.~9, pp. 2402--2415, Sep. 2016.

\bibitem{hareedy2019combinatorial}
A.~Hareedy, C.~Lanka, N.~Guo, and L.~Dolecek, ``A combinatorial methodology for
  optimizing non-binary graph-based codes: Theoretical analysis and
  applications in data storage,'' \emph{IEEE Transactions on Information
  Theory}, vol.~65, no.~4, pp. 2128--2154, Apr. 2019.

\bibitem{fossorier2004quasicyclic}
M.~P.~C. {Fossorier}, ``Quasicyclic low-density parity-check codes from
  circulant permutation matrices,'' \emph{IEEE {T}rans. {I}nformation
  {T}heory}, vol.~50, no.~8, pp. 1788--1793, Aug. 2004.

\bibitem{Schmalen}
L.~{Schmalen}, V.~{Aref}, and F.~{Jardel}, ``Non-uniformly coupled {L}{D}{P}{C}
  codes: Better thresholds, smaller rate-loss, and less complexity,'' in
  \emph{2017 IEEE International Symposium on Information Theory (ISIT)}, Jun.
  2017, pp. 376--380.

\bibitem{girth1}
Y.~Wang, J.~Yedidia, and S.~Draper, ``Construction of high-girth
  {Q}{C}-{L}{D}{P}{C} codes,'' in \emph{2008 5th International Symposium on
  Turbo Codes and Related Topics}, Sep. 2008, pp. 180--185.

\bibitem{girth2}
I.~E. Bocharova, F.~Hug, R.~Johannesson, B.~D. Kudryashov, and R.~V. Satyukov,
  ``Searching for voltage graph-based {L}{D}{P}{C} tailbiting codes with large
  girth,'' \emph{IEEE Transactions on Information Theory}, vol.~58, no.~4, pp.
  2265--2279, Apr. 2012.

\bibitem{girth3}
A.~Tasdighi, A.~H. Banihashemi, and M.-R. Sadeghi, ``Efficient search of
  girth-optimal {Q}{C}-{L}{D}{P}{C} codes,'' \emph{IEEE Transactions on
  Information Theory}, vol.~62, no.~4, pp. 1552--1564, Apr. 2016.

\bibitem{6387307}
J.~Wang, L.~Dolecek, and R.~D. Wesel, ``The cycle consistency matrix approach
  to absorbing sets in separable circulant-based {L}{D}{P}{C} codes,''
  \emph{IEEE Transactions on Information Theory}, vol.~59, no.~4, pp.
  2293--2314, Apr. 2013.

\bibitem{amiri2014analysis}
B.~Amiri, J.~Kliewer, and L.~Dolecek, ``Analysis and enumeration of absorbing
  sets for non-binary graph-based codes,'' \emph{IEEE transactions on
  communications}, vol.~62, no.~2, pp. 398--409, Feb. 2014.

\bibitem{hareedy2020minimizing}
A.~Hareedy, R.~Kuditipudi, and R.~Calderbank, ``Minimizing the number of
  detrimental objects in multi-dimensional graph-based codes,'' \emph{IEEE
  Transactions on Communications}, vol.~68, no.~9, pp. 5299--5312, Sep. 2020.

\bibitem{hareedy2016nonbinary}
A.~Hareedy, B.~Amiri, R.~Galbraith, and L.~Dolecek, ``Non-binary {L}{D}{P}{C}
  codes for magnetic recording channels: Error floor analysis and optimized
  code design,'' \emph{IEEE Transactions on Communications}, vol.~64, no.~8,
  pp. 3194--3207, Aug. 2016.

\bibitem{3DXpoint}
F.~T. Hady, A.~Foong, B.~Veal, and D.~Williams, ``Platform storage performance
  with 3{D} {X}{P}oint technology,'' \emph{Proceedings of the IEEE}, vol. 105,
  no.~9, pp. 1822--1833, Sep. 2017.

\bibitem{TDMR}
R.~Wood, M.~Williams, A.~Kavcic, and J.~Miles, ``The feasibility of magnetic
  recording at 10 terabits per square inch on conventional media,'' \emph{IEEE
  Transactions on Magnetics}, vol.~45, no.~2, pp. 917--923, Feb. 2009.

\bibitem{parnell2014modelling}
T.~Parnell, N.~Papandreou, T.~Mittelholzer, and H.~Pozidis, ``Modelling of the
  threshold voltage distributions of sub-20nm {N}{A}{N}{D} flash memory,'' in
  \emph{2014 IEEE Global Communications Conference}, 2014, pp. 2351--2356.

\bibitem{declercq2007decoding}
D.~Declercq and M.~Fossorier, ``Decoding algorithms for nonbinary {L}{D}{P}{C}
  codes over {G}{F}$(q)$,'' \emph{IEEE Transactions on Communications},
  vol.~55, no.~4, pp. 633--643, Apr. 2007.

\bibitem{divsalar2014non}
L.~Dolecek, D.~Divsalar, Y.~Sun, and B.~Amiri, ``Non-binary protograph-based
  ldpc codes: Enumerators, analysis, and designs,'' \emph{IEEE transactions on
  information theory}, vol.~60, no.~7, pp. 3913--3941, Jul. 2014.

\bibitem{bazarsky2013design}
A.~Bazarsky, N.~Presman, and S.~Litsyn, ``Design of non-binary quasi-cyclic
  {L}{D}{P}{C} codes by {A}{C}{E} optimization,'' in \emph{2013 IEEE
  Information Theory Workshop (ITW)}, Aug. 2013, pp. 1--5.

\bibitem{srinivasa2014communication}
S.~G. Srinivasa, Y.~Chen, and S.~Dahandeh, ``A communication-theoretic
  framework for $2$-{D} {M}{R} channel modeling: Performance evaluation of
  coding and signal processing methods,'' \emph{IEEE Transactions on
  Magnetics}, vol.~50, no.~3, pp. 6--12, Mar. 2014.

\bibitem{bcjr1974}
L.~Bahl, J.~Cocke, F.~Jelinek, and J.~Raviv, ``Optimal decoding of linear codes
  for minimizing symbol error rate (corresp.),'' \emph{IEEE Transactions on
  Information Theory}, vol.~20, no.~2, pp. 284--287, Feb. 1974.

\bibitem{jaekyun2001pattern}
J.~Moon and J.~Park, ``Pattern-dependent noise prediction in signal-dependent
  noise,'' \emph{IEEE Journal on Selected Areas in Communications}, vol.~19,
  no.~4, pp. 730--743, Apr. 2001.

\end{thebibliography}
\bibliographystyle{IEEEtran}

\newpage
\appendices

\section{Partitioning Matrices and Lifting Matrices for $(3,7)$ Codes on AWGN Channel}
\label{append: AWGN_3_7}

\begin{figure}[H]
\centering
  \subfigure[GD code.]{%
       \includegraphics[width=0.35\textwidth]{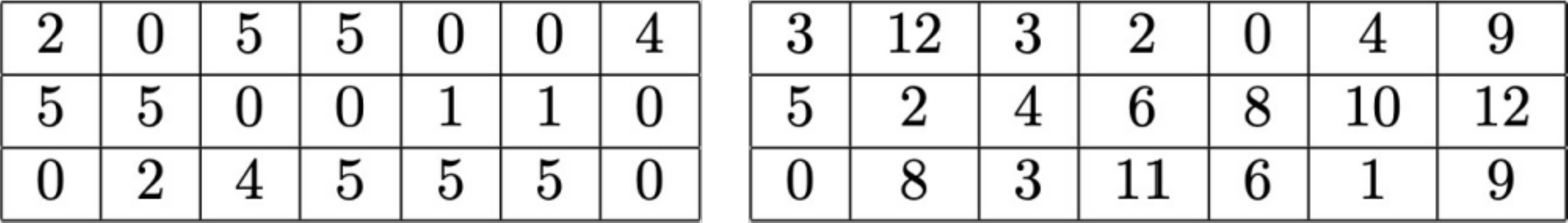}
       \label{1a}
       }
  \subfigure[UNF code.]{%
        \includegraphics[width=0.35\textwidth]{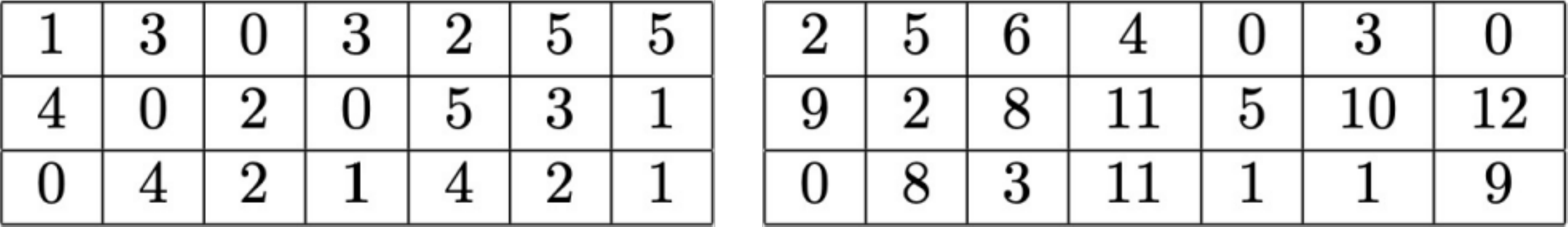}
        \label{1b}
        }
  \caption{Partitioning matrices (left) and lifting matrices (right) of GD/UNF codes with $(\gamma,\kappa,m,z,L)=(3,7,5,13,100)$.}
  \label{fig: code_GD_UNF_3_7} 
\end{figure}

\section{Partitioning Matrices and Lifting Matrices for $(3,17)$ Codes on AWGN Channel}
\label{append: AWGN_3_17}

\begin{figure}[H]
\centering
  \subfigure[GD code.]{%
       \includegraphics[width=0.38\textwidth]{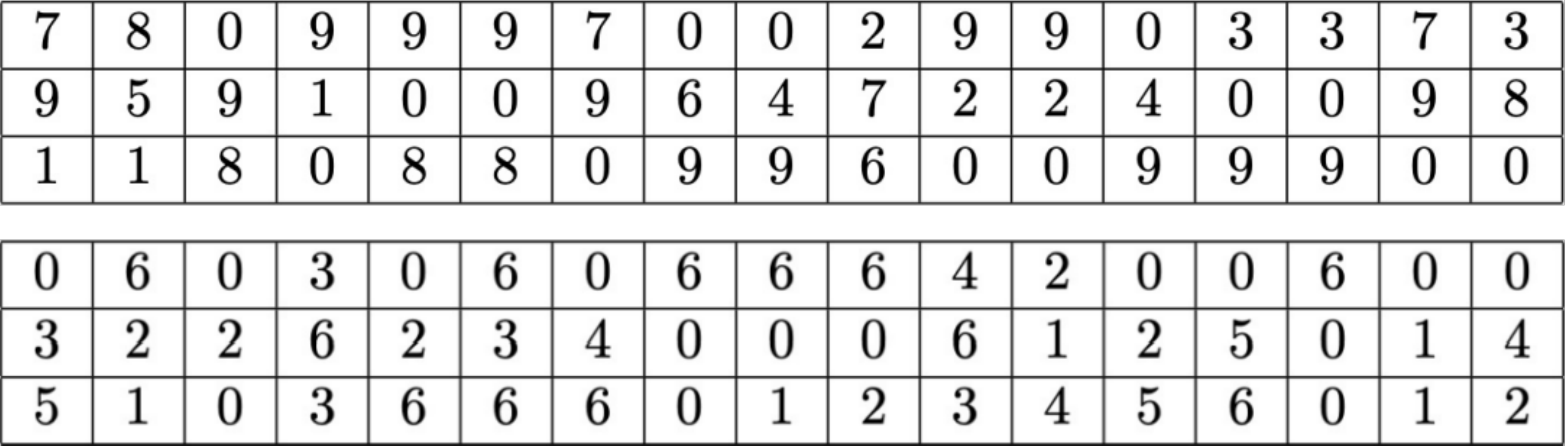}
       \label{1a}
       }
  \subfigure[UNF code.]{%
        \includegraphics[width=0.38\textwidth]{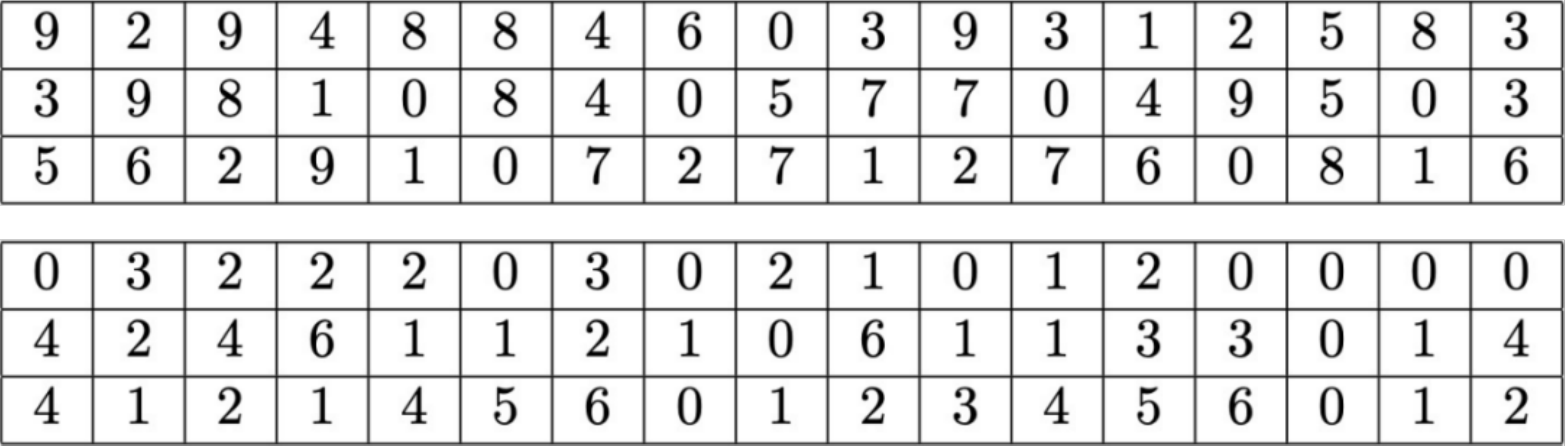}
        \label{1b}
        }
  \caption{Partitioning matrices (top) and lifting matrices (bottom) of GD/UNF codes with $(\gamma,\kappa,m,z,L)=(3,17,9,7,100)$.}
  \label{fig: code_GD_UNF_3_17} 
\end{figure}

\section{Partitioning Matrices and Lifting Matrices for $(4,29)$ Codes  on AWGN Channel}
\label{append: AWGN_4_29}

\begin{figure}[H]
\centering
\resizebox{0.87\textwidth}{!}{\begin{tabular}{|c|c|c|c|c|c|c|c|c|c|c|c|c|c|c|c|c|c|c|c|c|c|c|c|c|c|c|c|c|}
\hline
0&0&0&19&17&17&1&11&13&5&18&10&19&13&1&6&19&8&19&0&19&0&0&0&2&17&6&19&4\\
\hline
19&18&18&2&0&19&19&5&3&19&9&2&9&9&3&17&6&0&2&16&12&13&8&18&16&0&17&10&0\\
\hline
1&14&3&16&7&1&4&19&5&0&0&16&0&0&7&19&10&19&16&18&3&18&15&3&19&8&19&1&15\\
\hline
16&0&14&1&11&2&15&2&19&16&18&0&19&19&19&0&0&5&1&0&9&4&19&14&7&12&0&19&1\\
\hline
\end{tabular}}
\vspace{5pt}

\resizebox{0.87\textwidth}{!}{\begin{tabular}{|c|c|c|c|c|c|c|c|c|c|c|c|c|c|c|c|c|c|c|c|c|c|c|c|c|c|c|c|c|}
\hline
7&1&18&21&5&4&17&0&6&16&26&8&13&7&5&6&9&2&0&0&0&0&0&5&0&4&19&0&0\\
\hline
3&15&4&1&5&3&12&19&10&21&19&3&4&19&28&1&3&5&12&18&11&10&15&17&19&21&18&25&27\\
\hline
0&8&16&24&3&11&20&20&6&14&22&1&9&2&25&21&12&7&28&6&15&23&19&10&0&1&4&13&21\\
\hline
0&18&7&25&1&3&21&10&28&17&6&24&13&2&20&9&27&16&5&23&12&1&19&8&26&15&4&22&11\\
\hline
\end{tabular}}

\caption{Partitioning matrix (top) and lifting matrix (bottom) for GD code with $(\gamma,\kappa,m,z,L)=(4,29,19,29,20)$.}
  \label{fig: code_GD_4_29} 
\end{figure}

\begin{figure}[H]
\centering
\resizebox{0.87\textwidth}{!}{\begin{tabular}{|c|c|c|c|c|c|c|c|c|c|c|c|c|c|c|c|c|c|c|c|c|c|c|c|c|c|c|c|c|}
\hline
0&17&3&13&1&14&4&12&4&15&10&2&17&2&18&11&17&15&11&3&13&12&13&6&2&5&14&13&14\\
\hline
8&0&19&19&18&8&5&18&13&6&11&3&2&4&11&3&9&15&16&7&7&12&19&16&4&9&0&13&3\\
\hline
14&6&12&10&12&1&17&9&7&5&16&19&1&15&5&19&6&5&15&7&0&2&3&10&15&9&6&7&11\\
\hline
17&14&0&2&9&18&12&1&8&11&4&4&7&10&1&8&14&8&0&16&17&16&1&0&10&18&18&10&8\\
\hline
\end{tabular}}
\vspace{5pt}

\resizebox{0.87\textwidth}{!}{\begin{tabular}{|c|c|c|c|c|c|c|c|c|c|c|c|c|c|c|c|c|c|c|c|c|c|c|c|c|c|c|c|c|}
\hline
12&1&1&7&14&27&4&26&25&2&0&6&15&7&24&1&1&6&17&5&13&19&2&0&11&0&0&0&1\\
\hline
5&2&4&22&8&5&23&1&4&18&28&1&19&17&22&6&3&3&14&9&11&13&15&3&0&2&3&25&27\\
\hline
23&8&2&24&3&7&1&27&6&14&21&12&9&17&5&4&12&20&28&7&7&13&2&25&18&26&5&13&21\\
\hline
28&18&7&4&14&3&21&10&28&17&6&24&13&2&9&8&1&26&5&23&12&1&19&8&26&15&4&22&11\\
\hline
\end{tabular}}

\caption{Partitioning matrix (top) and lifting matrix (bottom) for UNF code with $(\gamma,\kappa,m,z,L)=(4,29,19,29,20)$.}
  \label{fig: code_UNF_4_29} 
\end{figure}

\section{Partitioning Matrices and Lifting Matrices for $(4,17)$ Codes on AWGN Channel}
\label{append: AWGN_4_17}

\begin{figure}[H]
\centering
  \subfigure[TC code with $(z,L)=(4,17,2,17,50)$ and $\mathbf{a}=(0,1,4)$.]{%
        \includegraphics[width=0.46\textwidth]{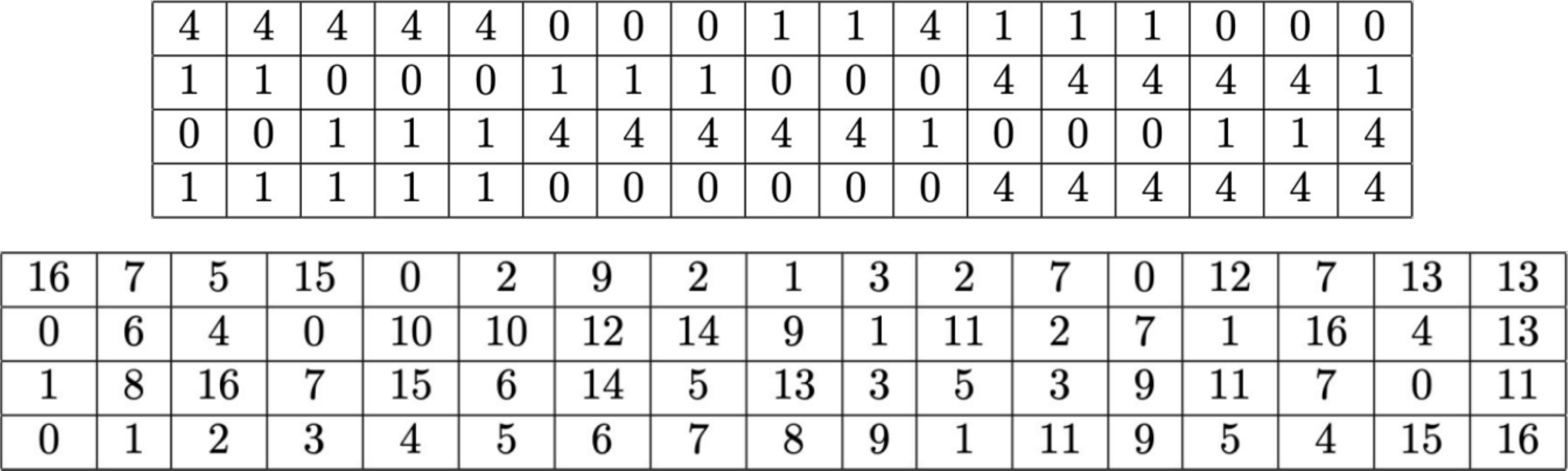}
        \label{1b}
        }
  \subfigure[SC codes with $(z,L)=(17,50)$ (middle) and $(z,L)=(28,30)$ (bottom).]{%
        \includegraphics[width=0.5\textwidth]{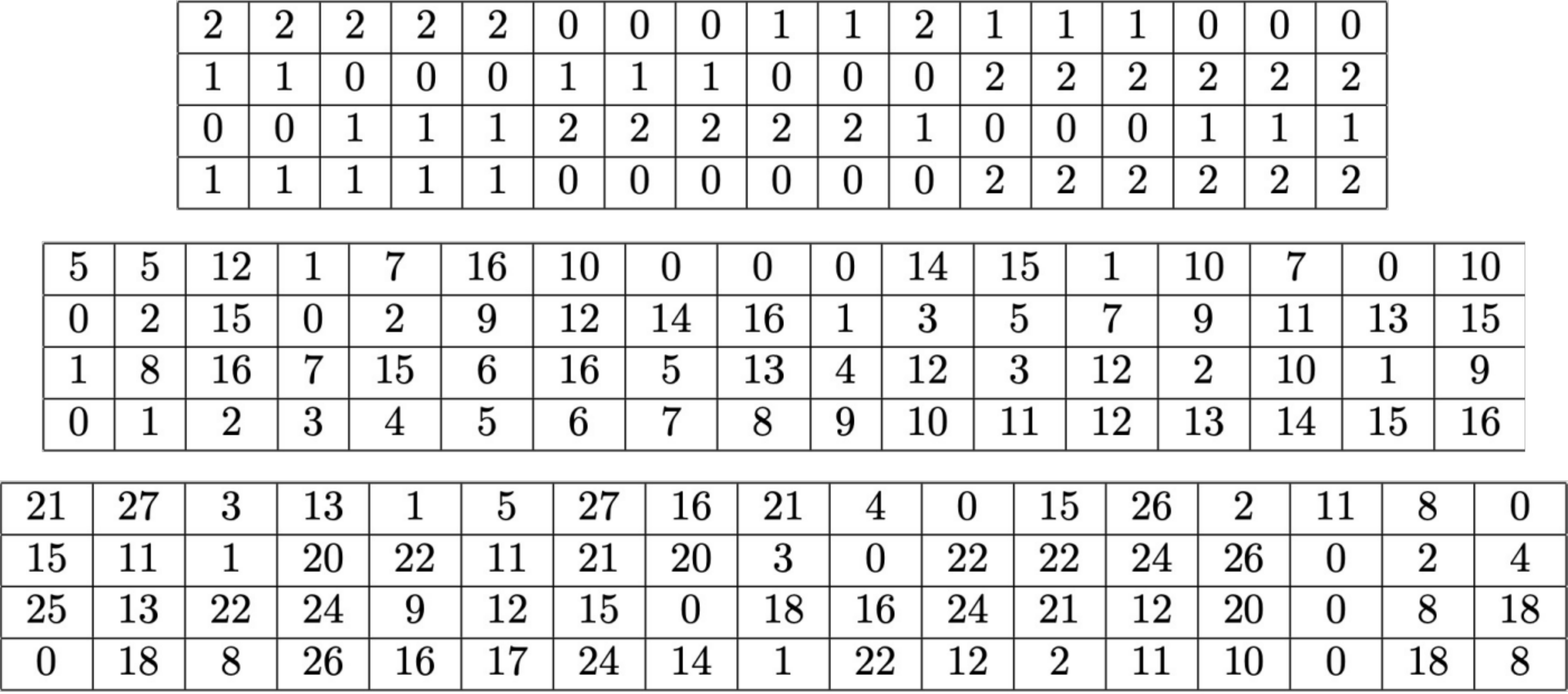}
        \label{1d}
        }
  \caption{Partitioning matrices (top) and lifting matrices (bottom) for TC/SC codes with $(\gamma,\kappa,m_t)=(4,17,2)$.}
  \label{fig: code_TC_SC_4_17} 
\end{figure}

\section{Partitioning Matrices and Lifting Matrices for $(4,24)$ Codes in Simulations on BSC}
\label{append: BSC_4_24}

\begin{figure}[H]
\centering
\resizebox{0.63\textwidth}{!}{\begin{tabular}{|c|c|c|c|c|c|c|c|c|c|c|c|c|c|c|c|c|c|c|c|c|c|c|c|c|c|c|c|c|}
\hline
0&6&2&0&6&1&0&1&6&6&6&6&0&0&0&4&5&4&5&0&0&5&0&1\\
\hline
0&0&6&2&6&6&6&0&1&1&5&6&4&6&0&0&1&6&6&0&6&1&4&0\\
\hline
3&3&6&6&0&6&1&5&3&0&0&1&2&6&6&6&2&0&0&6&4&6&0&6\\
\hline
6&5&1&5&2&0&3&5&0&3&0&0&6&0&6&3&6&2&0&6&0&0&6&4\\
\hline
\end{tabular}}
\vspace{5pt}

\resizebox{0.75\textwidth}{!}{\begin{tabular}{|c|c|c|c|c|c|c|c|c|c|c|c|c|c|c|c|c|c|c|c|c|c|c|c|c|c|c|c|c|}
\hline
11&15&1&7&8&11&14&1&5&6&16&9&12&0&5&13&1&0&3&5&15&0&0&1\\
\hline
9&2&8&6&8&10&10&3&5&8&5&5&7&12&9&14&15&0&2&4&6&8&10&12\\
\hline
0&2&6&7&15&12&4&5&13&4&7&8&11&2&1&1&8&0&8&16&9&15&11&5\\
\hline
0&1&2&3&3&5&6&7&14&8&14&11&12&13&14&15&16&0&1&2&3&4&11&6\\
\hline
\end{tabular}}

\caption{Partitioning matrix (top) and lifting matrix (bottom) for GD Code with $(\gamma,\kappa,m,z,L)=(4,24,6,17,40)$.}
  \label{fig: BSC_GD_4_24} 
\end{figure}

\begin{figure}[H]
\centering
\resizebox{0.63\textwidth}{!}{\begin{tabular}{|c|c|c|c|c|c|c|c|c|c|c|c|c|c|c|c|c|c|c|c|c|c|c|c|c|c|c|c|c|}
\hline
0&1&4&6&6&1&0&1&4&6&4&6&6&6&0&6&1&1&1&0&4&0&1&0\\
\hline
1&6&0&6&4&6&6&0&1&6&1&6&0&0&6&0&0&6&6&4&0&1&0&4\\
\hline
6&6&6&0&0&0&1&4&0&0&6&1&1&4&6&4&6&4&0&4&6&6&4&0\\
\hline
4&0&1&1&1&6&4&6&6&0&0&1&4&4&0&1&6&0&4&6&4&4&6&6\\
\hline
\end{tabular}}
\vspace{5pt}

\resizebox{0.75\textwidth}{!}{\begin{tabular}{|c|c|c|c|c|c|c|c|c|c|c|c|c|c|c|c|c|c|c|c|c|c|c|c|c|c|c|c|c|}
\hline
13&4&1&7&1&11&3&6&6&6&1&12&15&14&5&4&0&12&0&0&0&0&0&0\\
\hline
5&11&4&9&13&15&0&14&11&1&7&8&7&9&9&12&2&0&2&4&6&14&4&16\\
\hline
0&8&16&1&15&6&14&5&14&4&9&3&11&2&10&1&9&0&8&13&12&15&6&11\\
\hline
0&9&2&3&4&5&6&2&8&9&8&4&12&13&14&15&0&9&15&13&3&2&8&4\\
\hline
\end{tabular}}

\caption{Partitioning matrix (top) and lifting matrix (bottom) for TC code with $(\gamma,\kappa,m,z,L)=(4,24,6,17,40)$.}
  \label{fig: BSC_TC_4_24} 
\end{figure}

\begin{figure}[H]
\centering
\resizebox{0.63\textwidth}{!}{\begin{tabular}{|c|c|c|c|c|c|c|c|c|c|c|c|c|c|c|c|c|c|c|c|c|c|c|c|c|c|c|c|c|}
\hline
0&5&6&5&0&6&2&2&5&1&2&6&2&0&6&2&2&3&2&3&0&3&1&1\\
\hline
5&4&0&5&0&2&3&1&0&0&4&0&1&4&3&6&2&4&6&4&6&6&4&3\\
\hline
1&1&5&0&6&4&6&5&1&5&1&4&2&5&0&3&1&3&2&1&5&3&4&3\\
\hline
6&1&2&2&4&4&1&5&6&6&3&1&5&0&0&4&5&1&4&3&0&0&3&6\\
\hline
\end{tabular}}
\vspace{5pt}

\resizebox{0.75\textwidth}{!}{\begin{tabular}{|c|c|c|c|c|c|c|c|c|c|c|c|c|c|c|c|c|c|c|c|c|c|c|c|c|c|c|c|c|}
\hline
8&16&7&7&0&9&6&14&6&14&1&13&14&8&0&0&4&6&0&0&0&0&3&16\\
\hline
3&2&6&6&13&10&12&14&10&3&3&11&7&9&11&7&15&0&13&4&4&8&10&12\\
\hline
2&8&16&1&15&6&14&5&13&2&12&3&11&2&13&1&9&0&8&16&7&15&6&14\\
\hline
1&1&2&3&13&5&11&7&8&9&10&11&12&13&14&15&16&0&1&2&3&4&5&6\\
\hline
\end{tabular}}

\caption{Partitioning matrix (top) and lifting matrix (bottom) for UNF code with $(\gamma,\kappa,m,z,L)=(4,24,6,17,40)$.}
  \label{fig: BSC_UNF_4_24} 
\end{figure}

\section{Partitioning Matrices and Lifting Matrices for $(4,24)$ Codes in Simulations on NLM Channel}
\label{append: NLM_4_24}

\begin{figure}[H]
\centering
\resizebox{0.63\textwidth}{!}{\begin{tabular}{|c|c|c|c|c|c|c|c|c|c|c|c|c|c|c|c|c|c|c|c|c|c|c|c|c|c|c|c|c|}
\hline
0&6&2&0&6&1&0&1&6&6&6&6&0&0&0&4&5&4&5&0&0&5&0&1\\
\hline
0&0&6&2&6&6&6&0&1&1&5&6&4&6&0&0&1&6&6&0&6&1&4&0\\
\hline
3&3&6&6&0&6&1&5&3&0&0&1&2&6&6&6&2&0&0&6&4&6&0&6\\
\hline
6&5&1&5&2&0&3&5&0&3&0&0&6&0&6&3&6&2&0&6&0&0&6&4\\
\hline
\end{tabular}}
\vspace{5pt}

\resizebox{0.75\textwidth}{!}{\begin{tabular}{|c|c|c|c|c|c|c|c|c|c|c|c|c|c|c|c|c|c|c|c|c|c|c|c|c|c|c|c|c|}
\hline
16&15&1&7&8&11&14&1&5&6&16&9&12&0&5&13&1&0&3&5&15&0&0&1\\
\hline
9&2&8&6&8&10&10&3&5&8&5&5&7&12&9&14&15&0&2&4&6&8&10&12\\
\hline
0&2&6&7&15&12&4&5&13&4&7&8&11&2&1&1&8&0&8&16&9&15&11&5\\
\hline
0&1&2&3&3&5&6&7&14&8&14&11&12&13&14&15&16&0&1&2&3&4&11&6\\
\hline
\end{tabular}}

\caption{Partitioning matrix (top) and lifting matrix (bottom) for GD code with $(\gamma,\kappa,m,z,L)=(4,24,6,17,40)$.}
  \label{fig: NLM_GD_4_24} 
\end{figure}

\begin{figure}[H]
\centering
\resizebox{0.63\textwidth}{!}{\begin{tabular}{|c|c|c|c|c|c|c|c|c|c|c|c|c|c|c|c|c|c|c|c|c|c|c|c|c|c|c|c|c|}
\hline
0&1&4&6&6&1&0&1&4&6&4&6&6&6&0&6&1&1&1&0&4&0&1&0\\
\hline
1&6&0&6&4&6&6&0&1&6&1&6&0&0&6&0&0&6&6&4&0&1&0&4\\
\hline
6&6&6&0&0&0&1&4&0&0&6&1&1&4&6&4&6&4&0&4&6&6&4&0\\
\hline
4&0&1&1&1&6&4&6&6&0&0&1&4&4&0&1&6&0&4&6&4&4&6&6\\
\hline
\end{tabular}}
\vspace{5pt}

\resizebox{0.75\textwidth}{!}{\begin{tabular}{|c|c|c|c|c|c|c|c|c|c|c|c|c|c|c|c|c|c|c|c|c|c|c|c|c|c|c|c|c|}
\hline
13&4&1&7&1&11&3&6&6&6&1&12&15&14&5&4&0&12&0&0&0&0&0&0\\
\hline
5&11&4&9&13&15&0&14&11&1&7&8&7&9&9&12&2&0&2&4&6&14&4&16\\
\hline
0&8&16&1&15&6&14&5&14&4&9&3&11&2&10&1&9&0&8&13&12&15&6&11\\
\hline
0&5&2&3&4&5&6&2&8&9&8&4&12&13&14&15&0&9&15&13&3&2&8&4\\
\hline
\end{tabular}}

\caption{Partitioning matrix (top) and lifting matrix (bottom) for TC code with $(\gamma,\kappa,m,z,L)=(4,24,6,17,40)$.}
  \label{fig: NLM_TC_4_24} 
\end{figure}

\begin{figure}[H]
\centering
\resizebox{0.63\textwidth}{!}{\begin{tabular}{|c|c|c|c|c|c|c|c|c|c|c|c|c|c|c|c|c|c|c|c|c|c|c|c|c|c|c|c|c|}
\hline
0&5&6&5&0&6&2&2&5&1&2&6&2&0&6&2&2&3&2&3&0&3&1&1\\
\hline
5&4&0&5&0&2&3&1&0&0&4&0&1&4&3&6&2&4&6&4&6&6&4&3\\
\hline
1&1&5&0&6&4&6&5&1&5&1&4&2&5&0&3&1&3&2&1&5&3&4&3\\
\hline
6&1&2&2&4&4&1&5&6&6&3&1&5&0&0&4&5&1&4&3&0&0&3&6\\
\hline
\end{tabular}}
\vspace{5pt}

\resizebox{0.75\textwidth}{!}{\begin{tabular}{|c|c|c|c|c|c|c|c|c|c|c|c|c|c|c|c|c|c|c|c|c|c|c|c|c|c|c|c|c|}
\hline
8&16&7&7&0&9&6&14&6&14&1&13&14&8&0&0&4&6&0&0&0&0&3&16\\
\hline
3&2&6&6&13&10&12&14&10&3&3&11&7&9&11&7&15&0&13&4&4&8&10&12\\
\hline
2&8&16&1&15&6&14&5&13&2&12&3&11&2&13&1&9&0&8&16&7&15&6&14\\
\hline
1&1&2&3&13&5&11&7&8&9&10&11&12&13&14&15&16&0&1&2&3&4&5&6\\
\hline
\end{tabular}}

\caption{Partitioning matrix (top) and lifting matrix (bottom) for UNF code with $(\gamma,\kappa,m,z,L)=(4,24,6,17,40)$.}
  \label{fig: NLM_UNF_4_24} 
\end{figure}

\section{Partitioning Matrices and Lifting Matrices for $(4,20)$ Codes on MR Channel}
\label{append: MR_4_20}

\begin{figure}[H]
\centering
\resizebox{0.52\textwidth}{!}{\begin{tabular}{|c|c|c|c|c|c|c|c|c|c|c|c|c|c|c|c|c|c|c|c|c|c|c|c|c|c|c|c|c|}
\hline
0&0&3&6&0&1&2&2&0&6&6&0&6&5&6&0&6&0&5&6\\
\hline
3&0&1&0&5&0&6&3&2&0&6&3&1&6&6&6&0&6&6&5\\
\hline
6&6&0&1&1&6&1&6&5&0&0&4&5&0&0&4&0&5&2&3\\
\hline
1&6&6&4&6&4&6&0&6&4&0&6&0&1&0&0&6&2&0&0\\
\hline
\end{tabular}}
\vspace{5pt}

\resizebox{0.6\textwidth}{!}{\begin{tabular}{|c|c|c|c|c|c|c|c|c|c|c|c|c|c|c|c|c|c|c|c|c|c|c|c|c|c|c|c|c|}
\hline
3&5&4&8&3&11&0&2&0&1&0&9&6&9&0&0&2&2&1&1\\
\hline
7&2&6&6&8&4&8&1&2&6&6&4&6&7&7&7&6&8&5&10\\
\hline
0&7&3&11&0&7&9&1&12&0&2&10&5&0&8&3&11&7&1&9\\
\hline
0&5&10&2&7&12&4&11&1&6&11&5&8&7&5&10&2&7&12&4\\
\hline
\end{tabular}}

\caption{Partitioning matrix (top) and lifting matrix (bottom) for GD code with $(\gamma,\kappa,m,z,L)=(4,20,6,13,20)$.}
  \label{fig: MR_GD_4_20} 
\end{figure}

\begin{figure}[H]
\centering
\resizebox{0.52\textwidth}{!}{\begin{tabular}{|c|c|c|c|c|c|c|c|c|c|c|c|c|c|c|c|c|c|c|c|c|c|c|c|c|c|c|c|c|}
\hline
6&0&0&6&0&0&4&6&0&1&6&6&6&1&6&6&0&0&1&1\\
\hline
1&6&4&4&4&1&1&6&4&6&0&0&0&6&0&4&6&0&4&6\\
\hline
4&0&6&4&1&6&1&1&4&4&4&6&0&0&1&0&1&6&6&4\\
\hline
0&6&4&0&6&4&6&0&6&0&1&1&4&4&6&1&6&6&0&0\\
\hline
\end{tabular}}
\vspace{5pt}

\resizebox{0.63\textwidth}{!}{\begin{tabular}{|c|c|c|c|c|c|c|c|c|c|c|c|c|c|c|c|c|c|c|c|c|c|c|c|c|c|c|c|c|}
\hline
1&2&8&5&11&3&0&1&10&4&1&12&10&0&10&3&0&6&0&0\\
\hline
8&7&4&12&5&10&12&1&3&5&10&10&4&8&2&7&6&12&10&12\\
\hline
0&11&3&11&0&1&9&4&12&7&8&4&9&10&8&3&11&2&1&10\\
\hline
0&12&10&2&7&12&4&9&1&6&4&10&11&0&5&9&2&7&12&4\\
\hline
\end{tabular}}

\caption{Partitioning matrix (top) and lifting matrix (bottom) for TC code with $(\gamma,\kappa,m,z,L)=(4,20,6,13,20)$.}
  \label{fig: MR_TC_4_20} 
\end{figure}

\begin{figure}[H]
\centering
\resizebox{0.52\textwidth}{!}{\begin{tabular}{|c|c|c|c|c|c|c|c|c|c|c|c|c|c|c|c|c|c|c|c|c|c|c|c|c|c|c|c|c|}
\hline
0&0&6&5&5&1&5&1&6&3&3&2&2&3&2&6&0&3&0&2\\
\hline
5&5&2&1&1&2&5&2&0&3&5&3&5&5&1&6&6&1&2&2\\
\hline
6&0&1&5&2&4&1&4&4&0&4&4&6&0&2&0&6&3&4&4\\
\hline
0&4&0&3&4&5&1&5&1&6&1&0&1&4&6&3&0&6&3&6\\
\hline
\end{tabular}}
\vspace{5pt}

\resizebox{0.63\textwidth}{!}{\begin{tabular}{|c|c|c|c|c|c|c|c|c|c|c|c|c|c|c|c|c|c|c|c|c|c|c|c|c|c|c|c|c|}
\hline
6&11&1&1&4&9&11&0&3&11&12&0&0&1&0&9&2&0&10&1\\
\hline
10&3&4&5&9&2&2&3&8&5&7&9&5&12&0&4&6&8&10&12\\
\hline
12&0&3&11&5&1&9&12&12&7&4&4&5&12&3&7&11&6&0&9\\
\hline
9&5&10&2&7&10&4&9&1&6&11&5&8&5&5&11&9&7&12&4\\
\hline
\end{tabular}}

\caption{Partitioning matrix (top) and lifting matrix (bottom) for UNF code with $(\gamma,\kappa,m,z,L)=(4,20,6,13,20)$.}
  \label{fig: MR_UNF_4_20} 
\end{figure}

\end{document}